\documentclass[aps,twocolumn,showpacs,showkeys,superscriptaddress,nofootinbib,floatfix]{revtex4}

\usepackage{comment}

\usepackage{amsfonts}
\usepackage{amsmath}
\usepackage{amssymb}
\usepackage{amsthm}
\usepackage{amscd}

\usepackage{enumitem}

\usepackage{graphicx}
\usepackage{picture}

\usepackage{subfigure}

\usepackage{array}
\usepackage{hhline}
\usepackage{calc}

\newcommand{\T}{\rule{0pt}{3.6ex}}
\newcommand{\B}{\rule[-2.2ex]{0pt}{0pt}}
\newcommand{\M}{\rule{0pt}{3.2ex}}

\theoremstyle{plain}
\newtheorem{theorem}{Theorem}[section]
\newtheorem{proposition}[theorem]{Proposition}

\theoremstyle{definition}
\newtheorem{definition}[theorem]{Definition}
\theoremstyle{remark}
\newtheorem{example}[theorem]{Example}

\newlist{cindep}{enumerate}{1}
\setlist[cindep]{label=\textbf{I\arabic*}:, ref=\textbf{I\arabic*},
  leftmargin=*}
\newlist{qindep}{enumerate}{1}
\setlist[qindep]{label=\textbf{QI\arabic*}:, ref=\textbf{QI\arabic*},
  leftmargin=*}
\newlist{ccindep}{enumerate}{1}
\setlist[ccindep]{label=\textbf{CI\arabic*}:, ref=\textbf{CI\arabic*},
leftmargin=*}
\newlist{qcindep}{enumerate}{1}
\setlist[qcindep]{label={\normalfont\textbf{QCI\arabic*:}},
  ref=\textbf{QCI\arabic*}, leftmargin=*}

\newcommand{\Hilb}[1][]{\ensuremath{\mathcal{H}_{#1}}}
\newcommand{\Ket}[1]{\ensuremath{\left \vert #1 \right \rangle}}
\newcommand{\Bra}[1]{\ensuremath{\left \langle #1 \right \vert}}
\newcommand{\Tr}[2][]{\ensuremath{\text{Tr}_{#1} \left ( #2 \right )
  }}
\newcommand{\Lin}[1][]{\ensuremath{\mathfrak{L} \left ( \Hilb[#1] \right )}}

\newcommand{\Sprod}{\ensuremath{\star}}

\begin{document}

\title{Towards a Formulation of Quantum Theory as a Causally Neutral
  Theory of Bayesian Inference} 
\date{October 22, 2013}
\author{M. S. Leifer} 
\email{matt@mattleifer.info}
\affiliation{Department of Physics and
  Astronomy, University College London, Gower Street, London WC1E 6BT,
  United Kingdom} 
\affiliation{Perimeter Institute for Theoretical Physics, 31 Caroline
  St. N, Waterloo, Ontario, Canada, N2L 2Y5}
\author{Robert W. Spekkens}
\email{rspekkens@perimeterinstitute.ca}
\affiliation{Perimeter Institute for Theoretical Physics, 31 Caroline
  St. N, Waterloo, Ontario, Canada, N2L 2Y5}

\begin{abstract}
  Quantum theory can be viewed as a generalization of classical
  probability theory, but the analogy as it has been developed so far
  is not complete.  Whereas the manner in which inferences are
  made in classical probability theory is independent of the causal
  relation that holds between the conditioned variable and the
  conditioning variable, in the conventional quantum
  formalism, there is a significant difference between how one treats
  experiments involving two systems at a single time and those
  involving a single system at two times. In this article, we develop
  the formalism of \emph{quantum conditional states}, which provides a
  unified description of these two sorts of experiment.  
 In addition, concepts that are distinct in the conventional formalism become unified: channels, sets of states, and positive operator valued measures are all seen to be instances of conditional states; the action of a channel on a state, ensemble averaging, the Born rule, the composition of channels, and nonselective state-update rules are all seen to be instances of belief propagation.  
Using a quantum generalization of Bayes' theorem and the associated notion of Bayesian conditioning, we also show that the remote steering of quantum states can be described within our formalism as a mere updating of beliefs about one system given new information about another, and retrodictive inferences can be expressed using the same belief propagation rule as is used for predictive inferences. 
Finally, we show that previous arguments for interpreting the projection postulate as a quantum generalization of Bayesian conditioning are based on a misleading analogy and that it is best understood as a combination of belief propagation (corresponding to the nonselective state-update map) and conditioning on the measurement outcome.

\end{abstract}


\maketitle

\section{Introduction}

\label{Intro}

Quantum theory can be understood as a non-commutative generalization
of classical probability theory wherein probability measures are
replaced by density operators.  Much of quantum information theory,
especially quantum Shannon theory, can be viewed as the systematic
application of this generalization of probability theory to
information theory.

However, despite the power of this point of view, the conventional
formalism for quantum theory is a poor analogue to classical probability
theory because, in quantum theory, the appropriate mathematical
description of an experiment depends on its causal structure.  For
example, experiments involving a pair of systems at space-like
separation are described differently from those that involve a single
system at two different times.  The former are described by a joint
state on the tensor product of two Hilbert spaces, and the latter by
an input state and a dynamical map on a single Hilbert space.
Classical probability works at a more abstract level than this.  It
specifies how to represent uncertainty prior to, and independently of,
causal structure.  For example, our uncertainty about two random
variables is always described by a joint probability distribution,
regardless of whether the variables represent two space-like separated
systems or the input and output of a classical channel.  Although
channels represent time evolution, they are described mathematically
by conditional probability distributions.  The input state specifies a
marginal distribution, and thus we have the ingredients to define a
joint probability distribution over the input and output variables.
This joint probability distribution could equally well be used to
describe two space-like separated variables.  Therefore, we do not
need to know how the variables are embedded in space-time in advance
in order to apply classical probability theory.  This has the
advantage that it cleanly separates the concept of correlation from
that of causation. The former is the proper subject of probabilistic
inference and statistics.  Within the subjective Bayesian approach to
probability, independence of inference and causality has been
emphasized by de Finetti (\cite{Finetti1974}, Preface pp.\ x--xi):

\begin{quotation}
  Probabilistic reasoning---always to be understood as
  subjective---merely stems from our being uncertain about something.
  It makes no difference whether the uncertainty relates to an
  unforeseeable future, or to an unnoticed past, or to a past
  doubtfully reported or forgotten; it may even relate to something
  more or less knowable (by means of a computation, a logical
  deduction, etc.) but for which we are not willing to make the
  effort; and so on.
\end{quotation}

Thus, in order to build a quantum theory of Bayesian inference, we
need a formalism that is even-handed in its treatment of different
causal scenarios.  There are some clues that this might be possible.
Several authors have noted that that there are close connections, and
often isomorphisms, between the statistics that can be obtained from
quantum experiments with distinct causal arrangements
\cite{Taylor2004, Leifer2006, Leifer2007, Evans2010, Wharton2011,
  Marcovitch2011, Marcovitch2011a}.  Time reversal symmetry is an
example of this, but it is also possible to relate experiments
involving two systems at the same time with those involving a single
system at two times.  The equivalence \cite{Bennett1992} of
prepare-and-measure \cite{Bennett1984} and entanglement-based
\cite{Ekert1991} quantum key distribution protocols is an example of
this, and provides the basis for proofs of the security of the former
\cite{Shor2000}.  Such equivalences suggest that it may be possible to
obtain a causally neutral formalism for quantum theory by describing
such isomorphic experiments by similar mathematical objects.

One of the main goals of this work is to provide this unification for
the case of experiments involving two distinct quantum systems at one
time and those involving a single quantum system at two times, and to
provide a framework for making probabilistic inferences that is
independent of this causal structure.  Both types of experiment can be
described by operators on a tensor product of Hilbert spaces,
differing from one another only by a partial transpose.  Probabilistic
inference is achieved using a quantum generalization of Bayesian
conditioning applied to \emph{quantum conditional states}, which are
the main objects of study of this work.

Quantum conditional states are a generalization of classical
conditional probability distributions.  Conditional probability plays
a key role in classical probability theory, not least due to its role
in Bayesian inference, and there have been attempts to generalize it
to the quantum case.  The most relevant to quantum information are
perhaps the quantum conditional expectation \cite{Umegaki1962} (see
\cite{Petz2008, R'edei2007} for a basic introduction and
\cite{Petz1988} for a review) and the Cerf-Adami conditional density
operator \cite{Cerf1997, Cerf1997a, Cerf1998}.  To date, these have
not seen widespread application in quantum information, which casts
some doubt on whether they are really the most useful generalization
of conditional probability from the point of view of practical
applications.  Quantum conditional states, which have previously
appeared in \cite{Asorey2005, Zhang2007, Leifer2007}, provide an
alternative approach to this problem.  We show that they are useful
for drawing out the analogies between classical probability and
quantum theory, they can be used to describe both space-like and
time-like correlations, and they unify concepts that look distinct in
the conventional formalism.
  
The remainder of the introduction summarizes the contents of this
article.  It is meant to provide a broad overview of the conditional
states formalism, its motivations, and its applications, while
introducing only a minimum of the technical details found in the rest
of the paper.

\subsection{Irrelevance of causal structure to the rules of inference}

Unifying the quantum description of experiments involving two distinct
systems at one time with the description of those involving a single
system at two distinct times requires some modifications to the way
that the Hilbert space formalism of quantum theory is usually set up.
Conventionally, a Hilbert space $\Hilb[A]$ describes a system,
labelled $A$, that persists through time.  Given two such systems, $A$
and $B$, the joint system is described by the tensor product
$\Hilb[AB]=\Hilb[A] \otimes \Hilb[B]$.  In the present work, a Hilbert
space and its associated label should rather be thought of as
representing a localized region of space-time.  Specifically, an
\emph{elementary region} is a small space-time region in which an
agent might possibly make a single intervention in the course of an
experiment, for example by making a measurement or by preparing a
specific state.  Each elementary region is associated with a label and
a Hilbert space, for instance, $A$ and $\Hilb[A]$.

Generally, a \emph{region} will refer to a collection of elementary
regions.  A region that is composed of a pair of disjoint regions,
labelled $A$ and $B$, is ascribed the tensor product Hilbert space
$\Hilb[AB]=\Hilb[A] \otimes \Hilb[B]$.  In contrast to the usual
formalism, this applies regardless of whether $A$ and $B$ describe
independent systems or the same system at two different times.
Because of this, if an experiment involves a system that does persist
through time, then a different label is given to each region it
inhabits, e.g., the input and output spaces for a quantum channel are
assigned different labels.

Although we have motivated our work by the distinction between spatial
and temporal separation, in fact it is \emph{not} the spatio-temporal
relation between the regions that is relevant for how they ought to be
represented in our quantum generalization of probability theory.
Rather, it is the \emph{causal} relation that holds between them which
is important.

More precisely, what is important is the distinction between two
regions that are \emph{causally-related}, which is to say that one has
a causal influence on the other (perhaps via intermediaries), and two
regions that are \emph{acausally-related}, which is to say that
neither has a causal influence on the other (although they may have a
common cause or a common effect, or be connected via intermediaries to
a common cause or a common effect).

The causal relation between a pair of regions cannot be inferred
simply from their spatio-temporal relation.  Consider a relativistic
quantum theory for instance.  Although a pair of regions that are
space-like separated are always acausally-related, a pair of regions
that are time-like separated can be related causally, for instance if
they constitute the input and the output of a channel, or they can be
related acausally, for instance if they constitute the input of one
channel and the output of another. Although time-like separation
implies that a causal connection is \emph{possible}, it is whether
such a connection \emph{actually holds} that is relevant in our
formalism.  The distinction can also be made in non-relativistic
theories, and in theories with exotic causal structure.  Indeed,
causal structure is a more primitive notion than spatio-temporal
structure, and it is all that we need here.

Typically, we shall confine our attention to two paradigmatic examples
of causal and acausal separation (which can be formulated in either a
relativistic or a non-relativistic quantum theory).  Two distinct
regions at the same time, the correlations between which are
conventionally described by a bipartite quantum state, are
acausally-related.  The regions at the input and output of a quantum
channel, the correlations between which are conventionally described
by an input state and a quantum channel, are causally-related
(although there are exceptions, such as a channel which erases the
state of the system and then re-prepares it in a fixed
state).\footnote{Although it is not required here, one can be more
  precise about this distinction as follows. A causal structure for a
  set of quantum regions is represented by a directed acyclic graph
  wherein the nodes are the regions and the directed edges are
  relations of causal dependence (the restriction to acyclic graphs
  prohibits causal loops).  Two regions are said to be
  causally-related if for all paths connecting one to the other in the
  graph, every edge along the path is directed in the same sense.  Two
  systems are said to be acausally-related if for all paths connecting
  one to the other, not every edge along the path is directed in the
  same sense.  When there exist both sorts of paths between a pair of
  nodes, the associated regions are neither purely causally nor purely
  acausally-related.  We do not consider this case in the article.}

We unify the description of Bayesian inference in the two different
causal scenarios in the sense that various formulas are shown to have
precisely the same form, in particular, the relation between joints
and conditionals, the formula for Bayesian inversion and the formula
for belief propagation.

\subsection{Basic elements of the formalism}

Without providing all the details, we summarize the analogues, within
our formalism, of the most basic elements of classical probability
theory.  These are presented in table~\ref{tbl:Intro:Basics}.

\begin{table*}[htb]
  \begin{tabular}{|p{15em}|>{\centering}p{15em}|>{\centering}p{15em}|}
    \hhline{|~|-|-|}
    \multicolumn{1}{p{15em}|}{\T\B}& {\bf Classical} & {\bf Quantum} \tabularnewline
    \hline
    \T State & $P(R)$ & $\rho_A$ \tabularnewline
    \M Joint state & $P(R,S)$ & $\sigma_{AB}$ \tabularnewline
    \M\B Marginalization & $P(S) = \sum_R P(R,S)$ & $\rho_B = \Tr[A]{\sigma_{AB}}$ \tabularnewline
    \hline
    \T Conditional state & $P(S|R)$ & $\sigma_{B|A}$ \tabularnewline
    \M\B & $\sum_S P(S|R) = 1$ & $\Tr[B]{\sigma_{B|A}} = I_A$ \tabularnewline
    \hline
    \T Relation between joint and & $P(R,S) = P(S|R)P(R)$ & $\sigma_{AB} = \sigma_{B|A} \Sprod \rho_A$ \tabularnewline
    \M\B  conditional states & $P(S|R) = P(R,S)/P(R)$ & $\sigma_{B|A} = \sigma_{AB} \Sprod
    \rho_A^{-1}$ \tabularnewline
    \hline
    \T\B Bayes' theorem & $P(R|S) = P(S|R)P(R)/P(S)$ & $\sigma_{A|B} = \sigma_{B|A} \Sprod (\rho_A \rho_B^{-1})$ \tabularnewline
    \hline
    \T\B Belief propagation & $P(S) = \sum_R P(S|R)P(R)$ & $\rho_{B} = \Tr[A]{\sigma_{B|A} \rho_A}$ \tabularnewline
    \hline
  \end{tabular}
  \caption{\label{tbl:Intro:Basics}Analogies between the classical
    theory of Bayesian inference and the conditional states formalism
    for quantum theory.}
\end{table*}

For an elementary region $A$, the quantum analogue of a normalized
probability distribution is a conventional quantum state $\rho_A$,
that is, a positive trace-one operator on $\Hilb[A]$.  For a region
$AB$, composed of two disjoint elementary regions, the analogue of a
joint probability distribution is a trace-one operator $\sigma_{AB}$
on $\Hilb[AB]$.  This operator is not always positive (but we will
nonetheless refer to it as a \emph{state}).  The marginalization
operation is replaced by the partial trace operation, $\textrm{Tr}_A$,
which corresponds to ignoring region $A$.  The role of the marginal
distribution is played by the marginal state $\rho_B =
\Tr[A]{\sigma_{AB}}$.

The quantum analogue of a conditional probability is a
\emph{conditional state} for region $B$ given region $A$.  This is an
operator on $\Hilb[AB]$, denoted $\sigma_{B|A}$, that satisfies
$\Tr[B]{\sigma_{B|A}}= I_A$.

The relation between a conditional state and a joint state is
$\sigma_{B|A} = \sigma_{AB} \Sprod \rho_A^{-1}$, where the
$\Sprod$-product is a particular noncommutative and nonassociative
product, defined by $M \Sprod N \equiv N^{1/2} M N^{1/2}$, where we
have adopted the convention of dropping identity operators and tensor
products, so that $\sigma_{AB} \Sprod \rho_A^{-1}$ is shorthand for
$\sigma_{AB} \Sprod (\rho_A^{-1} \otimes I_B) = (\rho_A^{-1/2} \otimes
I_B) \sigma_{AB} (\rho_A^{-1/2} \otimes I_B)$.

This relation implies that the quantum
analogue of Bayes' theorem, relating $\sigma_{B|A}$ and $\sigma_{A|B}$, is
$\sigma_{A|B} = \sigma_{B|A} \Sprod (\rho_A \rho_B^{-1})$.  

A standard example of inference then proceeds as follows.  Suppose a
conditional state $\sigma_{B|A}$ represents your beliefs about the
relation that holds between a pair of elementary regions.  In this
case, if you represent your beliefs about $A$ by the quantum state
$\rho_A$, then you must represent your beliefs about $B$ by the
quantum state $\rho_B$, where
\begin{equation}
\rho_B = \Tr[A]{\sigma_{B|A} \rho_A}.
\end{equation}
We refer to this map from $\rho_A$ to $\rho_B$ as \emph{belief
  propagation}\footnote{Note that the term ``belief propagation'' has
  also been used to describe message-passing algorithms for performing
  inference on Bayesian networks.  This is not the intended meaning
  here.}.

\subsection{Relevance of causal structure to the form of the state}

In the case of acausally-related regions, it is the joint state that
is easily inferred from the conventional formalism, and the
conditional state that is derived from the joint.  Specifically, if
$A$ and $B$ are acausally-related, then their joint state,
$\sigma_{AB}$, is simply the bipartite state that one would assign to
them in the conventional formalism.  Consequently $\sigma_{AB}$ is a
positive operator in this case.  The conditional state can be inferred
from the rule relating joints to conditionals, namely, $\sigma_{B|A} =
\sigma_{AB} \Sprod \rho_A$.  It follows that $\sigma_{B|A}$ is also a
positive operator.

On the other hand, if $A$ and $B$ are causally-related, then it is the
conditional state that is easily inferred from the conventional
formalism, and the joint state that is derivative.  Specifically, if
the regions are related by a quantum operation $\mathcal{E}_{B|A}$,
then $\sigma_{B|A}$ is defined as the operator on $\Hilb[A]\otimes
\Hilb[B]$ that is Jamio{\l}kowski-isomorphic to
$\mathcal{E}_{B|A}$~\cite{Jamiolkowski1972}.  The joint state is then
inferred from the rule relating joints to conditionals.  One can show
that both $\sigma_{B|A}$ and $\sigma_{AB}$ fail to be positive in
general, but they have positive partial transpose.

Because of this, the set of permissible joint and conditional states
for acausally-related regions is different from the set for
causally-related regions.  To distinguish the two cases, we use the
notation $\rho_{AB}$ and $\rho_{B|A}$ for the acausal case, and
$\varrho_{AB}$ and $\varrho_{B|A}$ for the causal case.

It is important to note that in a classical theory of Bayesian
inference, it is the rules of inference that are independent of the
causal relations that hold among the variables.  The causal relations
can still be relevant, however, for constraining the probability
distribution that is assigned to those variables.  For instance, the
causal relations among a triple of variables are significant for the
sort of probability distribution that can be assigned to them.
Specifically, if variable $R$ is a common cause of variables $S$ and
$T$, while there is no direct causal connection between $S$ and $T$,
then $S$ and $T$ should be conditionally independent given $R$, which
is to say that the joint distribution over these variables is not
arbitrary, but has the form $P(R,S,T)=P(S|R)P(T|R)P(R)$.

In the quantum case, the situation is similar.  The rules of
inference, such as the formula for belief propagation, the formula for
Bayesian inversion, and the relation between the joint and the
conditional, do not depend on the causal relations between the regions
under consideration, but causal relations \emph{do} constrain the set
of operators that can describe joint states.

In fact, the dependence is stronger in the quantum case because the
set of permissible states depends on the causal relation even for a
\emph{pair} of regions.  This is not a feature of a classical theory
of inference: if we consider all the possible joint distributions over
a pair of variables, $R$ and $S$, we find that the set of
possibilities is the same for the case where $R$ and $S$ are
\emph{causally}-related as it is for the case where $R$ and $S$ are
\emph{acausally}-related.

To reiterate: the \emph{fact} that the set of possible states that can
be assigned to a set of regions is constrained by the causal relation
between those regions is \emph{common} to the classical and quantum
theories of inference.  What is particular to the theory of quantum
inference is that even in the case of a \emph{pair} of regions, the
causal relation between the regions is relevant for the set of
possible states that can be assigned to those regions.\footnote{There
  does exist a classical analogue of this dependence on causal structure
  for pairs of regions, but it requires considering the case of a
  classical theory with an epistemic restriction \cite{Bartlett2011}.
  We do not pursue the analogy here.}

\subsection{Recasting conventional quantum notions in terms of conditional states and belief propagation}

The conditional states formalism incorporates the possibility that a
given region is associated to a classical variable rather than a
quantum system.  In this case, the classical variable is represented
by a Hilbert space with a preferred basis, where the different
elements of the basis correspond to different values of the variable,
and any state assigned to that region is diagonal in that basis.  Any
joint state or conditional state involving this region is also
restricted to have this diagonal form.  It follows that for a set of
regions that are all classical, the formalism reproduces the classical
theory of Bayesian inference.

The formalism also yields a new and unified perspective on many
notions in quantum theory.  To see this, it is useful to recall that
measurements, sets of state preparations, and transformations can all
be represented by quantum operations, that is, as completely positive
trace-preserving (CPT) linear maps.  Channels are CPT maps wherein the
input and output spaces are both quantum.  A positive operator valued
measure (POVM) is a CPT map from a quantum input to a classical output
(the measurement outcome).  A set of states is a CPT map from a
classical input (the state index) to a quantum output (the associated
state).  Finally, a quantum instrument, which is a measurement
together with a state update rule for every outcome, can be
represented as a CPT map from a quantum input to a composite output
with a quantum part (the updated state) and a classical part (the
measurement outcome).  Insofar as every CPT map defines a conditional
state, each of these notions in quantum theory is an instance of a
conditional state in our formalism.  This is summarized in the top
half of table~\ref{tbl:CCS:Conventional}.

\begin{table*}[htb]
  \begin{tabular}{|p{18em}|>{\centering}p{15em}|>{\centering}p{15em}|}
    \hhline{|~|-|-|}
    \multicolumn{1}{p{18em}|}{\T\B} & {\bf Conventional Notation} &
    {\bf Conditional States Formalism} \tabularnewline
    \hline
    \T \hspace{0.5em} {\bf Probability distribution of $X$} & $P(X)$ & $\rho_{X}$ \tabularnewline
    \M\B \hspace{0.5em} {\bf Probability that $X=x$} & $P(X=x)$ & $\rho_{X=x}$ \tabularnewline
    \hline
    \T \hspace{0.5em} {\bf Set of states on $A$} & $\{ \rho^A_x \}$ & $\varrho_{A|X}$ \tabularnewline
    \M\B \hspace{0.5em} {\bf Individual state on $A$} & $\rho^A_x$ & $\varrho_{A|X=x}$ \tabularnewline
    \hline
    \T \hspace{0.5em} {\bf POVM on $A$} & $\{ E^A_y \}$ & $\varrho_{Y|A}$ \tabularnewline
    \M\B \hspace{0.5em} {\bf Individual effect on $A$} & $E^A_y$ & $\varrho_{Y=y|A}$ \tabularnewline
    \hline
    \T\B \hspace{0.5em} {\bf Channel from $A$ to $B$} & $\mathcal{E}_{B|A}$ & $\varrho_{B|A}$ \tabularnewline
    \hline
    \M \hspace{0.5em} {\bf Instrument} & $\{ \mathcal{E}^{B|A}_y \}$ & $\varrho_{YB|A}$ \tabularnewline
    \M\B \hspace{0.5em} {\bf Individual Operation} & $\mathcal{E}^{B|A}_y$ & $\varrho_{Y=y,B|A}$ \tabularnewline
    \hhline{|=|=|=|}
    \T\B \hspace{0.5em} {\bf The Born rule}  & $\forall y: P(Y=y) = \Tr[A]{E^A_y \rho_A}$ & $\rho_Y =
    \Tr[A]{\varrho_{Y|A} \rho_A}$ \tabularnewline
    \hline
    \T\B \hspace{0.5em} {\bf Ensemble averaging} & $\rho_A = \sum_x P(X=x) \rho^A_x$ & $\rho_A = \Tr[X]{\varrho_{A|X}
      \rho_X}$ \tabularnewline
    \hline
    \T\B \hspace{0.5em} {\bf Action of a channel
      (Schr{\"o}dinger)} &  $\rho_B = \mathcal{E}_{B|A} \left ( \rho_A \right )$ & $\rho_B =
    \Tr[A]{\varrho_{B|A} \rho_A}$ \tabularnewline
    \hline
    \T\B \hspace{0.5em} {\bf Composition of channels} & $\mathcal{E}_{C|A}
    = \mathcal{E}_{C|B} \circ \mathcal{E}_{B|A}$ & $\varrho_{C|A} =
    \Tr[B]{\varrho_{C|B}\varrho_{B|A}}$ \tabularnewline
    \hline
    \T\B \hspace{0.5em} {\bf Action of a channel (Heisenberg)} & $E^A_y
    = \left (\mathcal{E}_{B|A} \right )^{\dagger} \left ( E^B_y \right
    )$ & $\varrho_{Y|A} =
    \Tr[B]{\varrho_{Y|B}\varrho_{B|A}}$ \tabularnewline
    \hline
    \T\B \hspace{0.5em} {\bf Nonselective state update rule} & $\forall y:P(Y=y) \rho^B_y = \mathcal{E}^{B|A}_y \left ( \rho_A \right )$ &
    $\rho_{YB} = \Tr[A]{\varrho_{YB|A} \rho_A}$ \tabularnewline
    \hline
  \end{tabular}
  \caption{\label{tbl:CCS:Conventional}Translation of concepts and equations from conventional notation to the conditional states formalism.}
\end{table*}

It follows that many relations that seem unrelated in the conventional
formalism all become instances of the belief propagation rule in our
formalism.  This includes the Born rule, the formula for calculating
the average state for an ensemble, the composition of channels, the
state-update rule in a measurement, and the action of a channel in
both the Heisenberg and Schr\"{o}dinger pictures.  This is summarized
in the bottom half of table~\ref{tbl:CCS:Conventional}.

\subsection{Applications of the formalism}

The formalism also accommodates forms of belief propagation that do not
fit into the standard list of the previous section.

One example is the inference made about one system based on the
outcome of a measurement made on another when the two are correlated
by virtue of a common cause.  This reproduces the remote collapse
postulate of quantum theory, which is sometimes called ``remote
steering'' of a quantum state and was made famous by the thought
experiment of Einstein, Podolsky, and Rosen.  It follows that in the
conditional states framework, the steering effect is \emph{merely}
belief propagation (updating beliefs about one system based on new
evidence about another) and does not require any causal influence
from one to the other.  This interpretation has been advocated
previously by Fuchs \cite{Fuchs2003a}.  Our formalism also provides an
elegant derivation of the formula for the set of ensembles to which a
remote system may be steered, previously obtained by conventional
methods in \cite{Verstraete2002}.

Another example of an unconventional form of belief propagation is
retrodiction, that is, inferences about a region based on beliefs
about another region in its future.  We develop a retrodictive
formalism using our quantum Bayes' theorem.  The latter is a necessary
ingredient because the ``givens'' in a retrodiction problem are
typically the descriptions of sets of state preparations, measurements
and channels, each of which corresponds to a conditional wherein the
conditioning system is to the past of the conditioned system.  We use
Bayes' theorem to invert each of these conditionals to ones wherein
the conditioning system is to the future of the conditioned system.
Then, one can use these conditionals to propagate one's beliefs
backwards in time, that is, to update one's beliefs about the past
based on new evidence in the present.  This application of our
formalism is a good example of how one can achieve causal neutrality:
belief propagation backward in time follows the same rules as belief
propagation forward in time.  The retrodictive formalism we devise
coincides with the one introduced in \cite{Barnett2000, Pegg2002,
  Pegg2002a} in the case of unbiased sources, but differs in the
general case, retaining a closer analogy with classical Bayesian
inference.

In the case where a quantum system is passed through a channel
(possibly noisy), the Bayesian inversion of the conditional associated
to this channel, when interpreted as a quantum operation itself, is
the Barnum-Knill approximate error correction map \cite{Barnum2002}.
It follows that this error correction scheme is the quantum analogue
of the following classical error correction scheme: based on a
channel's output, compute a posterior distribution over inputs
(i.e.\ classical retrodiction) and then sample from the latter.

In the case where a quantum system is prepared in one of a set of
states, the Bayesian inversion of the conditional associated to this
set of states (a ``quantum given classical'' conditional) is a
conditional associated to a measurement (a ``classical given quantum''
conditional).  Indeed, we find that in these contexts, our quantum
Bayes' theorem reproduces the well-known rule relating sets of states
to Positive Operator Valued Measures (POVMs)~\cite{Hughston1993,
  Leifer2006, Leifer2007}.  The POVM obtained as the Bayesian
inversion of an ensemble of states turns out to be the ``pretty-good''
measurement for distinguishing those states~\cite{Hausladen1994,
  Belavkin1975, Belavkin1975a}.  Therefore, the latter, like the
Barnum-Knill recovery operation, can be understood as a quantum
analogue of sampling from the posterior.

Similarly, the Bayesian inversion of the conditional associated to a
measurement is a conditional associated to a set of states.  For this
case, our quantum Bayes' theorem reproduces a rule proposed by Fuchs
as a quantum analogue of Bayes' theorem~\cite{Fuchs2003a}.

Finally, we show that our notion of conditioning does not include the
projection postulate as a special case, and that previous arguments to
the contrary (i.e.\ in favour of the projection postulate being viewed
as an instance of Bayesian conditioning) \cite{Bub1978, Bub2007a} are
based on a misleading analogy.  Within the conditional states
formalism, the projection postulate is best described as the
application of a belief propagation rule (a non-selective update map),
followed by conditioning (the selection).  This is broadly in line
with the treatment of quantum measurements advocated by Ozawa
\cite{Ozawa1997, Ozawa1998}.  In support of the argument that the
projection postulate is not a type of conditioning, we provide a
conditional state version of the argument that all informative
measurements must be disturbing, which may be of independent interest
due to its close relationship to entanglement monogamy.

\subsection{Structure of the Paper}

The remainder of this paper is structured as follows.

The relevant aspects of classical conditional probability are reviewed
in \S\ref{CCP}.  \S\ref{CS} introduces quantum conditional states and
the basic concepts of quantum Bayesian inference for a pair of
regions.  The distinction between conditional states for
causally-related and acausally-related regions is discussed here.
This section also provides a detailed discussion of the translations
from the conventional formalism to the conditional states formalism
that are highlighted in table~\ref{tbl:CCS:Conventional}.

\S\ref{Bayes} introduces our quantum version of Bayes' theorem and
discusses its applications, in particular, the connection with the
update rule proposed by Fuchs, the correspondence between POVMs and
ensemble decompositions of a density operator, the ``pretty good''
measurement, and the Barnum-Knill recovery map.  In
\S\ref{Bayes:Retro}, we develop the retrodictive formalism for quantum
theory and describe how it relates to the one introduced in
\cite{Barnett2000, Pegg2002, Pegg2002a}. Finally, in
\S\ref{Bayes:Remote}, the acausal analogue of the symmetry between
prediction and retrodiction is discussed in the context of remote
measurement.

\S\ref{Cond} discusses quantum Bayesian conditioning.  After a brief
discussion of the general problem of conditioning a quantum region on
another quantum region, we focus on conditioning a quantum region on a
classical variable.  This is the correct way to update quantum states
in light of classical data, regardless of the causal relationship
between the two.  Various examples of this are discussed in
\S\ref{Cond:Exa}, including the case of the remote steering
phenomenon.  \S\ref{Cond:Direct} concerns how to understand within our
formalism the rules for updating quantum states after a nondestructive
quantum measurement, in particular, how to understand the projection
postulate.

In \S\ref{Related}, we discuss related work.  Quantum conditional
states are compared to other proposals for quantum generalizations of
conditional probability in \S\ref{Related:Comp} and the conditional
states formalism is compared to several recently proposed operational
reformulations of quantum theory in \S\ref{Related:Op}.

\S\ref{Lim} discusses limitations of the conditional states framework.
These arise because the classical operation of taking the product of a
conditional and a marginal probability distribution to form a joint
distribution is replaced by a noncommutative and nonassociative
operation on the corresponding operators in the quantum case.  Because
of this, unlike in classical probability, an equation involving
conditional states does not necessarily remain valid when both sides
are conditionalized on an additional variable.  This is discussed in
\S\ref{Lim:Cond}.  \S\ref{Lim:CJS} discusses the reasons why causal
joint states are limited to two elementary regions.
\S\ref{Lim:CJS:MCS} discusses why they cannot be applied to mixed
causal scenarios, such as two acausally related regions with a third
region causally related to one of the other two, and
\S\ref{Lim:CJS:MTS} discusses the difficulties with generalizing the
notion to multiple time steps, where we have three or more causally
related regions.

\S\ref{Open} discusses an open question about when assignments of
conditional states are compatible with one another.  That not all
conditional assignments are compatible can be shown via the monogamy
of entanglement.  Indeed, this incompatibility seems to be a more
basic notion, of which monogamy is a consequence.  Finally, we
conclude in \S\ref{Conc}.

\section{Classical Conditional Probability}

\label{CCP}

In this section, the basic definitions and formalism of classical
conditional probability are reviewed, with a view to their quantum
generalization in \S\ref{CS}.

Let $R$ denote a (discrete) random variable, $R=r$ the event that $R$
takes the value $r$, $P(R=r)$ the probability of event $R=r$, and
$P(R)$ the probability that $R$ takes an arbitrary unspecified value.
Finally, $\sum_{R}$ denotes a sum over the possible values of $R$.

A conditional probability distribution is a function of two random
variables $P(S|R)$, such that for each value $r$ of $R$, $P(S|R = r)$
is a probability distribution over $S$.  Equivalently, it is a
positive function of $R$ and $S$ such that
\begin{equation}
  \label{eq:CCP:CP}
  \sum_S P(S|R) = 1
\end{equation}
independently of the value of $R$.

Given a probability distribution $P(R)$ and a conditional probability
distribution $P(S|R)$, a joint distribution over $R$ and $S$ can be
defined via
\begin{equation}
  \label{eq:CCP:JointCP}
  P(R,S) = P(S|R) P(R),
\end{equation}
where the multiplication is defined element-wise, i.e.\ for all values
$r, s$ of $R$ and $S$, $P(R = r, S = s) = P(S = s|R = r)P(R = r)$.

Conversely, given a joint distribution $P(R,S)$, the marginal
distribution over $R$ is defined as
\begin{equation}
  P(R) = \sum_S P(R,S),
\end{equation}
and the conditional probability of $S$ given $R$ is
\begin{equation}
  \label{eq:CCP:CPJoint}
  P(S|R)=\frac{P(R,S)}{P(R)}.
\end{equation}
Note that eq.~\eqref{eq:CCP:CPJoint} only defines a conditional
probability distribution for those values $r$ of $R$ such that $P(R =
r) \neq 0$.  The conditional probability is undefined for other values
of $R$.

The chain rule for conditional probabilities states that a joint
probability over $n$ random variables $R_1,R_2, \ldots, R_n$ can be
written as
\begin{multline}
  \label{eq:CCP:Chain}
  P(R_1,R_2, \ldots, R_n) = P(R_n|R_{1},R_2, \ldots , R_{n-1}) \\
  \times P(R_{n-1}|R_1,R_2, \ldots, R_{n-2}) \ldots P(R_2|R_1) P(R_1).
\end{multline}

Finally, note that the process of marginalizing a distribution over a
set of variables commutes with the process of conditioning on
a disjoint set of variables, as illustrated in the following
commutative diagram.
\begin{equation}
  \begin{CD}
    P(R,S,T) @>\sum_R>> P(S,T) \\
    @VV\times P(T)^{-1}V @VV\times P(T)^{-1}V \\
    P(R,S|T) @>\sum_R>> P(S|T)
  \end{CD}
\end{equation}

\section{Quantum Conditional States}

\label{CS}

In this section, the quantum analogue of conditional probability --- a
conditional state --- is introduced.  We also discuss how the states
assigned to disjoint regions are related via a quantum analogue of the
belief propagation rule $P(S) = \sum_R P(S|R)P(R)$.  There is a small
difference between conditional states for acausally-related and
causally-related regions.  The acausal case is discussed in
\S\ref{ACS}-\S\ref{ACS:ABP}.  On the other hand,
\S\ref{CCS}-\ref{CCS:QI} mainly concern the causal case, wherein we
find that quantum dynamics, ensemble averaging, the Born rule,
Heisenberg dynamics, and the transition from the initial state to the
ensemble of states resulting from a measurement can all be represented
as special cases of quantum belief propagation.  Acausal analogise of
some of these ideas are also developed in these sections.

\subsection{Acausal Conditional States}
\label{ACS}

We begin by defining conditional states for acausally-related regions.
This scenario, and its classical analogue, are depicted in
fig.~\ref{fig:ACS:Acausal}. The definition proceeds in analogy with
the classical treatment given in \S\ref{CCP}.  The convention of using
$A,B,C,\ldots$ to label quantum regions that are analogous to
classical variables $R, S, T, \ldots$ is adopted throughout.  The
labels $X, Y, Z, \ldots$ are reserved for classical variables
associated with preparations and measurements, which remain classical
when we pass from probability theory to the quantum analogue.

\begin{figure}[htb]
  \centering
  \subfigure[]{
    \includegraphics[scale=0.6]{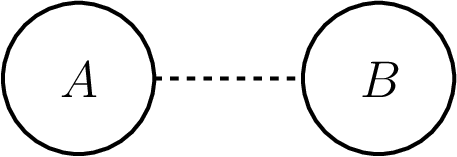}
    \label{fig:ACS:Entangled}
  }
  \subfigure[]{
    \includegraphics[scale=0.6]{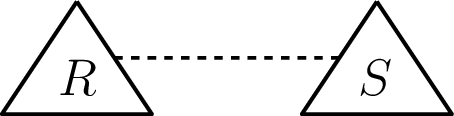}
    \label{fig:ACS:Correlated}
  }
  \caption{\label{fig:ACS:Acausal}Acausally-related quantum and
    classical regions.  Classical variables are denoted by triangles and
    quantum regions by circles (this convention is suggested by the
    shape of the convex set of states in each theory).  The dotted
    line represents acausal correlation.  \subref{fig:ACS:Entangled}
    Two quantum regions in an arbitrary joint state (possibly
    correlated).  \subref{fig:ACS:Correlated} Two classical variables
    with an arbitrary joint probability distribution (possibly
    correlated).}
\end{figure}

The analogue of a probability distribution $P(R)$ assigned to a random
variable $R$ is a quantum state (density operator) $\rho_A$ acting on
a Hilbert space $\Hilb[A]$.  When there are two disjoint regions with
Hilbert spaces $\Hilb[A]$ and $\Hilb[B]$, the tensor product
$\Hilb[AB] = \Hilb[A] \otimes \Hilb[B]$ describes the composite
region. The quantum analogue of a joint distribution $P(R,S)$ is a
density operator $\rho_{AB}$ of the composite region, defined on
$\Hilb[AB]$.  The analogue of marginalization over a variable is the
partial trace over a region. These analogies are set out in the top
half of table~\ref{tbl:ACS:Analogs}.

\begin{table}[htb]
  \begin{tabular}{|>{\centering}p{12em}|>{\centering}p{12em}|}
    \hline
    \T\B {\bf Classical Probability} & {\bf Quantum Theory} \tabularnewline
    \hhline{|=|=|}
    \T $P(R)$ & $\rho_A$ \tabularnewline
    \M $P(R,S)$ & $\rho_{AB}$ \tabularnewline
    \M\B $P(S) = \sum_R P(R,S)$ & $\rho_B = \Tr[A]{\rho_{AB}}$ \tabularnewline
    \hline
    \T $\sum_S P(S|R) = 1$ & $\Tr[B]{\rho_{B|A}} = I_A$ \tabularnewline
    \M $P(R,S) = P(S|R)P(R)$ & $\rho_{AB} = \rho_{B|A} \Sprod \rho_A$
    \tabularnewline
    \M\B $P(S|R) = P(R,S)/P(R)$ & $\rho_{B|A} = \rho_{AB} \Sprod
    \rho_A^{-1}$ \tabularnewline
    \hline
  \end{tabular}
  \caption{\label{tbl:ACS:Analogs}Analogies between classical
    probability theory for two random variables and quantum theory for
    two acausally-related regions.}
\end{table}

In analogy to the classical case, where $P(S|R)$ is a positive
function that satisfies $\sum_S P(S|R) = 1$, an acausal conditional
state for $B$ given $A$ is defined as follows.
\begin{definition}
  \label{def:ACS:ACS}
  An \emph{acausal conditional state} for $B$ given $A$ is a positive operator
  $\rho_{B|A}$ on $\Hilb[AB]=\Hilb[A]\otimes \Hilb[B]$ that satisfies
  \begin{equation}
    \label{eq:ACS:CS}
    \Tr[B]{\rho_{B|A}} = I_A,
  \end{equation}
  where $I_A$ is the identity operator on $\Hilb[A]$.
\end{definition}

To provide an analogy with eq.~\eqref{eq:CCP:JointCP}, a method of
constructing a joint state on $\Hilb[AB]$ from a reduced state on
$\Hilb[A]$ and a conditional state on $\Hilb[AB]$ is required.  This
is given by
\begin{equation}
  \label{eq:ACS:Joint}
  \rho_{AB} = (\rho_A^{\frac{1}{2}} \otimes I_B) \rho_{B|A}
  (\rho_A^{\frac{1}{2}} \otimes I_B).
\end{equation}

Eq.~\eqref{eq:ACS:Joint} involves two constructions that appear
repeatedly in what follows.  Firstly, the operators
$\rho_A^{\frac{1}{2}}$ and $\rho_{B|A}$ are combined via
multiplication, but they are defined on different spaces.  To solve
this problem, $\rho_A^{\frac{1}{2}}$ is expanded to an operator on
$\Hilb[AB]$ by tensoring it with $I_B$.  To simplify notation, the
identity operators required to equalize the Hilbert spaces of two
operators will be left implicit, so that if $M_{AB}$ is an operator on
$\Hilb[AB]$ and $N_{BC}$ is an operator on $\Hilb[BC]$ then
$M_{AB}N_{BC} = \left (M_{AB} \otimes I_C \right ) \left ( I_A \otimes
  N_{BC} \right )$ and an equation like $M_{AB} = N_{BC}$ is
interpreted as $M_{AB} \otimes I_C = I_A \otimes N_{BC}$.  This
notation allows us to omit tensor product symbols where convenient,
since $ M_A \otimes N_B= \left ( M_A \otimes I_B \right ) \left ( I_A
  \otimes N_B \right )= M_A N_B$.

Secondly, rather than simply multiplying $\rho_{B|A}$ with $\rho_A$ in
eq.~\eqref{eq:ACS:Joint}, $\rho_{B|A}$ is conjugated by
$\rho_A^{\frac{1}{2}}$.  This ensures that the resulting joint
operator is positive.  To define a notation for this conjugation, let
$M$ and $N$ be positive operators on a Hilbert space $\Hilb$.  Then
define a (non-associative and non-commutative) product $M \Sprod N$
via
\begin{equation}
  M \Sprod N= N^{\frac{1}{2}} M N^{\frac{1}{2}}.
\end{equation}
With these conventions, eq.~\eqref{eq:ACS:Joint} can be rewritten as
\begin{equation}
  \label{eq:ACS:JointStar}
  \rho_{AB} = \rho_{B|A} \Sprod \rho_A,
\end{equation}
which looks a lot closer to eq.~\eqref{eq:CCP:JointCP} than
eq.~\eqref{eq:ACS:Joint} does.

Starting with a joint state $\rho_{AB}$ and its reduced state $\rho_A
= \Tr[B]{\rho_{AB}}$, a conditional state can be defined via
\begin{equation}
  \label{eq:ACS:CPJoint}
  \rho_{B|A} = \rho_{AB} \Sprod \rho_A^{-1},
\end{equation}
which is the analogue of eq.~\eqref{eq:CCP:CPJoint}.

As with eq.~\eqref{eq:CCP:CPJoint} there are problems with this
formula if $\rho_A$ is not supported on the entire Hilbert space
$\Hilb[A]$.  In that case eq.~\eqref{eq:ACS:CPJoint} is to be
understood as an equation on the Hilbert space $\text{supp}(\rho_A)
\otimes \Hilb[B]$, where $\text{supp}(\rho_A)$ denotes the support of
$\rho_A$ (the span of the eigenvectors of $\rho_A$ having nonzero
eigenvalues).  Because of this, the resulting conditional density
operator satisfies $\Tr[B]{\rho_{B|A}} = I_{\text{supp}(\rho_A)}$
rather than eq.~\eqref{eq:ACS:CS}.

The analogies between the classical and quantum relations between
conditionals, marginals and joints are set out in the bottom half of
table~\ref{tbl:ACS:Analogs}.  As these analogies suggest, the
$\Sprod$-product notation allows equations from classical probability
to be generalized to quantum theory by replacing functions by
operators, products by $\Sprod$-products, and division by
$\Sprod$-products with the inverse.  However, whilst this is a useful
way of postulating results in the conditional states formalism, one
has to take care of the non-associativity and non-commutativity of the
$\Sprod$-product when making such generalizations.

To provide an analogy with the chain rule of eq.~\eqref{eq:CCP:Chain}
it is helpful to adopt the convention that, in the absence of
parentheses, $\Sprod$-products are evaluated right-to-left.  Then,
given $n$ disjoint acausally-related regions $A_1,A_2,\ldots,A_n$ with
Hilbert space $\Hilb[A_1A_2\ldots A_n] = \bigotimes_{j = 1}^n
\Hilb[A_j]$, the joint state can be written as
\begin{multline}
  \label{eq:ACS:Chain}
  \rho_{A_1A_2 \ldots A_n} = \rho_{A_n|A_{1}A_2 \ldots A_{n-1}} \\
  \Sprod \rho_{A_{n-1}|A_1A_2 \ldots A_{n-2}} \Sprod \ldots \Sprod
  \rho_{A_2|A_1} \Sprod \rho_{A_1}.
\end{multline}

Finally, note that the process of marginalizing a conditional state
over a region commutes with the process of conditioning on a disjoint
region, as illustrated in the following commutative diagram:
\begin{equation}
  \begin{CD}
    \rho_{ABC} @>\text{Tr}_C>> \rho_{AB} \\
    @VV\rho_A^{-\frac{1}{2}}\left (\cdot \right)\rho_A^{-\frac{1}{2}}V
    @VV\rho_A^{-\frac{1}{2}}\left (\cdot
    \right)\rho_A^{-\frac{1}{2}}V \\
    \rho_{BC|A} @>\text{Tr}_C>> \rho_{B|A}
  \end{CD}
\end{equation}

\begin{example}[Classical States]
  \label{exa:ACS:Class}
  As one might expect, classical conditional probability is a special
  case of the quantum constructions outlined above.  To see this, the
  classical variables have to be encoded in quantum regions in some
  way, and we adopt the convention of using the same letter to denote
  the classical variable and the corresponding quantum region.  Thus,
  $\Hilb[R], \Hilb[S], \Hilb[T], \ldots$ refer to quantum regions that
  encode classical random variables $R, S, T, \ldots$, as opposed to
  $\Hilb[A], \Hilb[B], \Hilb[C], \ldots$, which are general quantum
  regions.

  For classical random variables $R$ and $S$, pick Hilbert spaces
  $\Hilb[R]$ and $\Hilb[S]$ with dimension equal to the number of
  distinct values of $R$ and $S$ respectively, and choose orthonormal
  bases $\left \{ \Ket{r} \right \}$ for $\Hilb[R]$ and $\left \{
    \Ket{s} \right \}$ for $\Hilb[S]$ labelled by the possible values
  of $R$ and $S$.  Then, joint, marginal and conditional probability
  distributions are encoded as operators via
  \begin{align}
    \rho_{RS} & = \sum_{r,s} P(R = r, S = s) \Ket{r}\Bra{r}_{R}
    \otimes \Ket{s}\Bra{s}_{S} \label{eq:ACS:ClassJoint}\\
    \rho_R & = \sum_r P(R = r) \Ket{r}\Bra{r}_{R} = \Tr[S]{\rho_{RS}} \\
    \rho_{S|R} & = \sum_{r,s} P(S = s| R = r) \Ket{r}\Bra{r}_{R}
    \otimes \Ket{s}\Bra{s}_{S}. \label{eq:ACS:ClassCond}
  \end{align}
  Using eqs.~(\ref{eq:CCP:CP}-\ref{eq:CCP:CPJoint}), it is
  straightforward to check that these operators satisfy
  eqs.~\eqref{eq:ACS:CS}, \eqref{eq:ACS:JointStar} and
  \eqref{eq:ACS:CPJoint}.

  In order to unify the notation for classical variables and quantum
  regions, the operators $\rho_{RS}, \rho_R$ and $\rho_{S|R}$ are
  often used to directly represent the functions $P(R,S), P(R)$ and
  $P(S|R)$ without introducing the classical functions explicitly.
  Whenever states and conditional states have subscripts $R, S, T$ or
  $X, Y, Z$, they are implicitly assumed to be of this classical form.
  If needed, the classical functions can be read off from
  eqs. (\ref{eq:ACS:ClassJoint}-\ref{eq:ACS:ClassCond}).
\end{example}

\begin{example}[Pure Conditional States]
  \label{exa:ACS:Pure}
  A pure conditional state is one that is of the form
  $\rho_{B|A}=\Ket{\psi}\Bra{\psi}_{B|A}$ for some vector
  $\Ket{\psi}_{B|A} \in \Hilb[AB]$.  Since $\Tr[B]{\rho_{B|A}} = I_A$,
  and $I_A$ has all eigenvalues equal to $1$, the Schmidt
  decomposition of $\Ket{\psi}_{B|A}$ is of the form
  \begin{equation}
    \label{eq:ACS:PureCS}
    \Ket{\psi}_{B|A}=\sum_{k}\Ket{u_{k}}_A \otimes \Ket{v_{k}}_B,
  \end{equation}
  where $\{\Ket{u_k}\}$ is an orthonormal basis for $\Hilb[A]$ and
  $\{\Ket{v_k}\}$ is an orthonormal basis for $\Hilb[B]$.  This
  implies that a pure conditional state only exists if
  $\text{dim}(\Hilb[A]) \leq \text{dim}(\Hilb[B])$ because otherwise
  there would not be enough orthonormal vectors on the $B$ side to
  enforce $\Tr[B]{\rho_{B|A}} = I_A$\footnote{If $\Ket{\psi}_{B|A}$
    derives from a joint pure state $\Ket{\psi}_{AB}$via
    eq.~\eqref{eq:ACS:CPJoint} then this only implies that
    $\text{dim}(\text{supp}(\rho_A)) \leq \text{dim}(\Hilb[B])$, which
    is always true because the ranks of $\rho_A$ and $\rho_B$ are
    equal.}.

  Since all the Schmidt coefficients are the same, the bases
  $\{\Ket{u_k}\}$ and $\{\Ket{v_k}\}$ are highly non-unique.  The
  conditional state $\Ket{\psi}_{B|A}$ itself only determines the
  relationship between the two Schmidt bases, i.e.\ for any basis in
  $\Hilb[A]$ it determines a corresponding basis in $\Hilb[B]$.  To
  see this, fix a reference basis, $\{\Ket{j}\}$, for $\Hilb[A]$ in
  order to define a complex conjugation operation.  Next, define an
  isometry $U_{B|A} = \sum_{k} \Ket{v_k}_B\Bra{u_k^*}_A$, where $^*$
  denotes complex conjugation in the $\{\Ket{j}\}$ basis.  Then, if
  $\{\Ket{w_k}_A\}$ is any other basis for $\Hilb[A]$,
  eq.~\eqref{eq:ACS:PureCS} can be rewritten as
  \begin{equation}
    \label{eq:ACS:PureCSArb}
    \Ket{\psi}_{B|A} = \sum_{k} \Ket{w_k}_A  \otimes U_{B|A'} \Ket{w_k^*}_{A'},
  \end{equation}
  where $A'$ labels a second copy of $\Hilb[A]$.  With respect to the
  reference basis $\{\Ket{j}\}$, this simplifies to
  \begin{equation}
    \label{eq:ACS:PureCSCan}
    \Ket{\psi}_{B|A} = U_{B|A'} \Ket{\Phi^+}_{AA'},
  \end{equation}
  where $\Ket{\Phi^+}_{AA'} = \sum_j \Ket{jj}_{AA'}$.

  Let $\rho_A$ be an arbitrary density operator on $\Hilb[A]$ with
  eigendecomposition $\rho_A = \sum_k p_k \Ket{w_k}\Bra{w_k}_{A}$.
  Combining this with $\Ket{\psi}\Bra{\psi}_{B|A}$ via
  eq.~\eqref{eq:ACS:JointStar} in order to define a joint state gives
  the projector onto the pure state
  \begin{equation}
    \Ket{\psi}_{AB} = \rho_A^{\frac{1}{2}} \Ket{\psi}_{B|A}.
  \end{equation}
  Combining this with eq.~\eqref{eq:ACS:PureCSArb} gives the Schmidt
  decomposition
  \begin{equation}
    \Ket{\psi}_{AB} = \sum_k \sqrt{p_k} \Ket{w_k}_A \otimes U_{B|A'}
    \Ket{w_k^*}_{A'}. 
  \end{equation}
  Since an arbitrary pure joint state is of this form, this shows that
  pure conditional states determine pure joint states when combined
  with arbitrary reduced states and, conversely, the conditional state
  of a pure joint state is always pure.

  Note that, using eq.~\eqref{eq:ACS:PureCSCan} instead of
  eq.~\eqref{eq:ACS:PureCSArb} gives
  \begin{equation}
    \Ket{\psi}_{AB} = \rho_{A}^{1/2} U_{B|A^{\prime
      }}\Ket{\Phi^+}_{AA^{\prime}},
  \end{equation}
  which is a well known canonical decomposition of a bipartite pure
  state.
\end{example}

\subsection{Acausal Belief Propagation}

\label{ACS:ABP}

Suppose you characterize your beliefs about two classical variables,
$R$ and $S$, by specifying a marginal probability distribution $P(R)$
and a conditional probability distribution $P(S|R)$.  Then, you can
compute the probability distribution you ought to assign to $S$ via
\begin{equation}
  \label{eq:ACS:CBP}
  P(S) = \sum_R P(S|R) P(R).
\end{equation}
This is called the \emph{classical belief propagation rule} (also
known as the law of total probability).  It follows from calculating
the joint distribution $P(R,S)= P(S|R) P(R)$ and then marginalizing
over $R$.

The belief propagation rule can be thought of as specifying a linear
map $\Gamma_{S|R}$ from the space of probability distributions over
$R$ to the space of probability distributions over $S$ that preserves
positivity and normalization.  This is defined as
\begin{equation}
  \label{eq:ACS:CBPMap}
  \Gamma_{S|R} \left ( P(R) \right ) \equiv \sum_R P(S|R) P(R).
\end{equation}

Propagating beliefs about a quantum region $A$ to an acausally-related
region $B$ works in a similar way.  If you specify a reduced state
$\rho_A$ and a conditional state $\rho_{B|A}$ then your state for $B$
is determined by the \emph{acausal quantum belief propagation rule}
\begin{equation}
\label{eq:ACS:QBP}
\rho_B = \Tr[A]{\rho_{B|A} \rho_A}
\end{equation}
which follows from the fact that the joint state is $\rho_{AB} =
\rho_{B|A} \Sprod \rho_A$, so that $\rho_B = \Tr[A]{\rho_{B|A} \Sprod
  \rho_A}$, and from the cyclic property of the trace.

As in the classical case, acausal belief propagation can also be
viewed as a linear map $\mathfrak{E}_{B|A}$ from states on $A$ to
states on $B$ that preserves positivity and normalization, defined by
\begin{equation}
  \label{eq:ACS:QBPMap}
  \mathfrak{E}_{B|A} \left ( \rho_A \right ) \equiv \Tr[A]{\rho_{B|A} \rho_A}.
\end{equation}

The linear map so defined is clearly positive because it maps states
to states. It is not completely positive in general, but its
composition with a transpose on $A$ is completely positive.  The map
$\mathfrak{E}_{B|A}$ is in fact identical to the map associated to
$\rho_{B|A}$ via the Jamio{\l}kowski isomorphism
\cite{Jamiolkowski1972}, which is a familiar construction in quantum
information theory.  These facts are consequences of the following
theorem.

\begin{theorem}[Jamio{\l}kowski Isomorphism]
  \label{thm:ACS:Jamiol}
  Let $\mathfrak{E}_{B|A}:\Lin[A] \to \Lin[B]$ be a linear map and let
  $M_{AC} \in \Lin[AC]$, where $\Hilb[C]$ is a Hilbert space of
  arbitrary dimension.  Then, the action of $\mathfrak{E}_{B|A}$ on
  $\Lin[A]$ (tensored with the identity on $\Lin[C]$) is given by
  \begin{equation}
    \label{eq:ACS:reverseJamiol}
    (\mathfrak{E}_{B|A} \otimes \mathcal{I}_C) \left ( M_{AC} \right )
    = \Tr[A]{\rho_{B|A} M_{AC}},
  \end{equation}
  where $\rho_{B|A}\in \Lin[AB]$ is given by
  \begin{equation}
    \label{eq:ACS:Jamiol}
    \rho_{B|A} \equiv (\mathfrak{E}_{B|A'} \otimes \mathcal{I}_{A}) \left (
      \sum_{j,k} \Ket{j}\Bra{k}_A \otimes \Ket{k}\Bra{j}_{A'} \right ).
  \end{equation}
  Here, $A'$ labels a second copy of $A$, $\mathcal{I}_A$ is the
  identity superoperator on $\Lin[A]$, and $\{\Ket{j} \}$ is an
  orthonormal basis for $\Hilb[A]$.

  Furthermore, the operator $\rho_{B|A}$ is an acausal conditional
  state, i.e.\ it satisfies definition~\ref{def:ACS:ACS}, if and only
  if $\mathfrak{E}_{B|A}\circ T_A$ is completely-positive and
  trace-preserving (CPT), where $T_A: \Lin[A] \to \Lin[A]$ denotes the
  linear map implementing the partial transpose relative to some
  basis.
\end{theorem}

The proof is provided in appendix~\ref{Appendix:Proofs}.

\subsection{Causal Conditional States}

\label{CCS}

The analogy between conditional probabilities and conditional states
presented so far is not complete. In conventional quantum theory, the
tensor product $\Hilb[AB] = \Hilb[A] \otimes \Hilb[B]$ is used to
represent a joint system with two subsystems, so that the conditional
state $\rho_{B|A}$ refers to the state of two subsystems at a given
time.  However, for classical conditional probabilities, there is no
corresponding requirement that the two random variables $R$ and $S$
appearing in $P(S|R)$ should have any particular causal relation to
one another.  Indeed, $R$ might equally well represent the input to a
classical channel and $S$ the output, i.e.\ they may be causally
related. This is illustrated in fig.~\ref{fig:CCS:Stochastic}.  If
this is indeed the case, then the classical belief propagation rule of
eqs.~\eqref{eq:ACS:CBP} and \eqref{eq:ACS:CBPMap} can be interpreted
as stochastic dynamics.

\begin{figure}[htb]
  \begin{minipage}{5cm}
    \raggedright
    \subfigure[]{
      \includegraphics[scale=0.4]{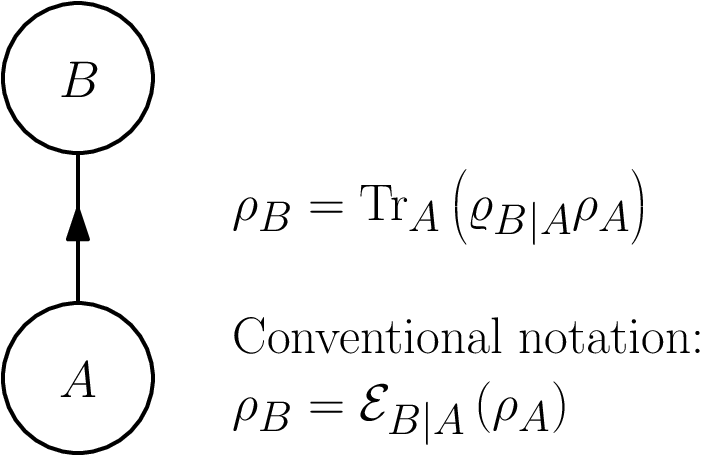}
      \label{fig:CCS:CPMap}
    }
    \subfigure[]{
      \includegraphics[scale=0.4]{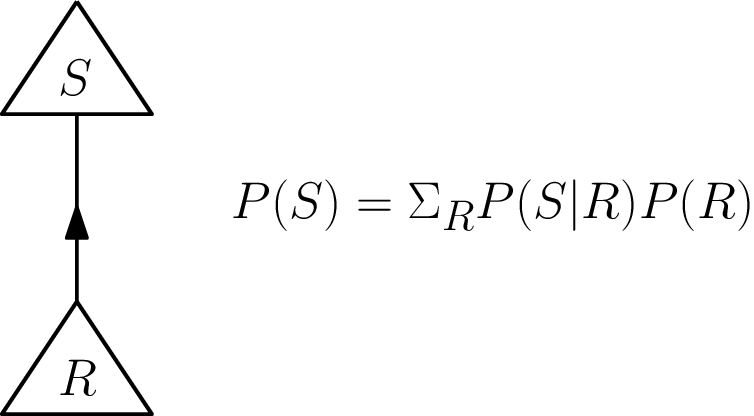}
      \label{fig:CCS:Stochastic}
    }
  \end{minipage}
  \caption{\label{fig:CCS:Timelike}Causally-related quantum and
    classical regions.  The arrows represent the direction of causal
    influence. \subref{fig:CCS:CPMap} General quantum dynamics.  $A$ is
    the input to a CPT map and $B$ is the output.
    \subref{fig:CCS:Stochastic} Classical stochastic dynamics.  $R$ is
    the input to a classical channel and $S$ is the output.}
\end{figure}

In order to formulate quantum theory as a causally neutral theory of
Bayesian inference, the same formalism should be used to describe
causally-related regions as is used to describe acausally-related
regions.  In particular, if $A$ and $B$ are two causally-related
regions, as depicted in fig.~\ref{fig:CCS:CPMap}, then it ought to be
possible to define a quantum conditional state for $B$ given $A$ as an
operator on the tensor product $\Hilb[AB] = \Hilb[A] \otimes
\Hilb[B]$. Towards this end, we make the following definition.
\begin{definition}
  \label{def:CCS:CCS}
  A \emph{causal conditional state} of $B$ given $A$ is an operator
  $\varrho_{B|A}$ on $\Hilb[AB]$ that can be written as
  \begin{equation}
    \varrho_{B|A} = \rho_{B|A}^{T_A},
  \end{equation}
  for some acausal conditional state $\rho_{B|A}$, where $^{T_A}$
  denotes the partial transpose in some basis on $\Hilb[A]$.
\end{definition}
Thus, the set of causal conditional states is just the image under a
partial transpose on the conditioning region of the set of acausal
conditional states.  Note that, although the partial transpose is
basis dependent, its image on the set of acausal conditional states is
not and therefore neither is our definition of a causal conditional
state.  Also, because the set of acausal conditional states is mapped
to itself by the full transpose (i.e.\ the transpose on $AB$), a
partial transpose over the conditioned region $B$, rather than the
conditioning region $A$, could alternatively have been used to define
a causal conditional state. Due to the partial transpose, causal
conditional states are not positive operators in general, but they are
always \emph{locally} positive, i.e. $\Bra{\psi}_A \otimes
\Bra{\phi}_B \varrho_{B|A}\Ket{\psi}_A \otimes \Ket{\phi}_B \geq 0$
for all $\Ket{\psi}_A \in \Hilb[A], \Ket{\phi}_B \in \Hilb[B]$.

In this section, we show that defining causal conditional states in
this way allows us to implement quantum belief propagation across
causally-related regions using the same formula as one uses for
quantum belief propagation across acausally-related regions, namely by
a rule of the form $\rho_B = \Tr[A]{\varrho_{B|A} \rho_A}$.  Belief
propagation for dynamics with a quantum input and a quantum output are
treated in \S\ref{CCS:CPT}.  \S\ref{CCS:CondBP} treats belief
propagation for causal conditional states themselves, which
corresponds to composition of dynamical maps.  \S\ref{CCS:TJS}
introduces the notion of a causal joint state, which is analogous to
the joint distribution of input and output variables for a classical
channel.  \S\ref{CCS:Hybrid} introduces the idea of a
quantum-classical hybrid, which is a composite of a quantum region and
a classical variable.  This allows dynamics with a classical input and
quantum output (and vice versa) to be described in terms of causal
conditional states.  These correspond to ensemble preparation
procedures and measurements, as discussed in \S\ref{CCS:States} and
\S\ref{CCS:POVM}.  In \S\ref{CCS:Heis}, the Heisenberg picture is
translated into the conditional states formalism.  In \S\ref{CCS:QI},
the most general type of state update rule that can occur after a
measurement -- a quantum instrument -- is described in terms of causal
conditional states.  Table~\ref{tbl:CCS:Conventional} summarizes the
translation of these concepts from conventional notation to the
conditional states formalism.

\subsection{Quantum Channels as Causal Belief Propagation}

\label{CCS:CPT}

Conventionally, the transition from a region $A$ to a causally-related
later region $B$ is described by a dynamical CPT map
$\mathcal{E}_{B|A}: \Lin[A] \rightarrow \Lin[B]$ such that, if
$\rho_A$ is the state of $A$ and $\rho_B$ is the state of $B$, then
$\rho_B = \mathcal{E}_{B|A} \left ( \rho_A \right )$.  However, causal
conditional states can provide an alternative representation of
quantum dynamics, as we will show.  First note the following
isomorphism.

\begin{theorem}
  \label{thm:CCS:Jamiol}
  Let $\mathcal{E}_{B|A}:\Lin[A] \to \Lin[B]$ be a linear map and let
  $\varrho_{B|A}\in \Lin[AB]$ be the Jamio{\l}kowski-isomorphic
  operator, as defined in eq.~\eqref{eq:ACS:Jamiol}.  Then,
  $\varrho_{B|A}$ is a causal conditional state, i.e.\ it satisfies
  definition~\ref{def:CCS:CCS}, if and only if $\mathcal{E}_{B|A}$ is
  CPT.
\end{theorem}
\begin{proof}
  Define $\rho_{B|A} \equiv \varrho_{B|A}^{T_A}$ and let
  $\mathfrak{E}_{B|A}$ be the linear map that is
  Jamio{\l}kowski-isomorphic to $\rho_{B|A}$.  It follows that
  $\mathcal{E}_{B|A}=\mathfrak{E}_{B|A}\circ T_A$.  Recalling the
  relation between causal and acausal conditional states,
  $\varrho_{B|A}$ is a causal conditional state if and only if
  $\rho_{B|A}$ is an acausal conditional state.  Recalling
  Theorem~\ref{thm:ACS:Jamiol}, $\rho_{B|A}$ is an acausal conditional
  state if and only if $\mathfrak{E}_{B|A}\circ T_A$ is CPT.  It
  follows that $\varrho_{B|A}$ is a causal conditional state if and
  only if $\mathcal{E}_{B|A}$ is CPT.
\end{proof}

Together with theorem~\ref{thm:ACS:Jamiol}, this implies that the
action of a CPT map $\mathcal{E}_{B|A}$ on an operator $M_{AC}$ is
given by
\begin{equation}
  \label{eq:CCS:reverseJamiol}
  \mathcal{E}_{B|A} \left ( M_{AC} \right ) = \Tr[A]{\varrho_{B|A} M_{AC}},
\end{equation}
where $\varrho_{B|A}$ is the Jamio{\l}kowski isomorphic operator to
$\mathcal{E}_{B|A}$.

Quantum dynamics may be represented by causal conditional states as
follows.
\begin{proposition}
  \label{prop:CCS:Jamiol}
  Let $\varrho_{B|A}$ be the causal conditional state that is
  Jamio{\l}kowski-isomorphic to a CPT map $\mathcal{E}_{B|A}$ that
  describes a quantum dynamics.  If the initial state of region $A$ is
  $\rho_A$, then the state of $B$, conventionally written as
  \begin{equation}
    \label{eq:CCS:Dynamics}
    \rho_B = \mathcal{E}_{B|A} \left ( \rho_A \right ),
  \end{equation}
  can be expressed in the conditional states formalism as
  \begin{equation}
    \label{eq:CCS:Evolve}
    \rho_B = \Tr[A]{\varrho_{B|A} \rho_A},
  \end{equation}
  in analogy with the classical belief propagation rule
  eq.~(\ref{eq:ACS:CBP}).
\end{proposition}
We call eq.~(\ref{eq:CCS:Evolve}) the \emph{causal quantum belief
  propagation rule}. It follows from eq.~\eqref{eq:CCS:reverseJamiol}.

Fig.~\ref{fig:CCS:Timelike} and the fourth and eighth lines of
table~\ref{tbl:CCS:Conventional} summarize how this representation of
quantum dynamics contrasts with the conventional representation, with
fig.~\ref{fig:CCS:Timelike} emphasizing the analogy between the
classical and quantum belief propagation rules.

\subsection{Causal Joint States}

\label{CCS:TJS}

\S\ref{ACS} showed that a joint state $\rho_{AB}$ of two acausally
related regions can be decomposed into a reduced state $\rho_A$ and an
acausal conditional state $\rho_{B|A}$.  Similarly, two causally
related regions can be described by an input state $\rho_A$ and a
causal conditional state $\varrho_{B|A}$, but so far there is no
causal analogue of a joint state.  This is addressed by making the
following definition, in analogy with eq.~\eqref{eq:ACS:JointStar},
\begin{definition}
  A \emph{causal joint state} of two causally-related regions, $A$ and
  $B$, is an operator on $\Hilb[AB]$ of the form
  \begin{equation}
    \label{eq:CPT:CCSJoint}
    \varrho_{AB} = \varrho_{B|A} \Sprod \rho_A,
  \end{equation}
  where $\rho_A$ is a state on $\Hilb[A]$ and $\varrho_{B|A}$ is a
  causal conditional state of $B$ given $A$.
\end{definition}

Note that the reduced state on $A$ of $\varrho_{AB}$ is the initial
state (input to the channel) and the reduced state on $B$ is
\begin{equation}
  \rho_B = \Tr[A]{\varrho_{B|A} \Sprod \rho_A},
\end{equation}
which, by the cyclic property of the trace and
proposition~\ref{prop:CCS:Jamiol}, is the final state (output of the
channel).

It is not too difficult to see that a causal joint state
$\varrho_{AB}$ is the partial transpose of an acausal joint state on
$\Hilb[AB]$.  Specifically, $\varrho_{AB}^{T_A} = \rho_{B|A} \Sprod
\rho_A^{T_A}$, where $\rho_{B|A} = \varrho_{B|A}^{T_A}$ is an acausal
conditional state and $\rho_A^{T_A}$ is a valid reduced state because
the transpose preserves positivity.

Thus, just as a causal conditional state for $B$ given $A$ is an
operator on $\Hilb[A]\otimes\Hilb[B]$ that can be obtained as the
partial transpose over $A$ of an acausal conditional state
$\rho_{B|A}$, a causal joint state on $AB$ is simply an operator on
$\Hilb[A]\otimes\Hilb[B]$ that can be obtained as the partial
transpose over $A$ of an acausal joint state $\rho_{AB}$.

\begin{example}[Unitary Dynamics]
  Suppose a region $A$ is assigned the state $\rho_A$ with
  eigendecomposition $\rho_A=\sum_j p_j \Ket{u_j}\Bra{u_j}$.  $A$ is
  then mapped to a region $B$, which has a Hilbert space of the same
  dimension as that of $A$, by an isometry $U_{B|A} = \sum_j
  \Ket{v_j}_B\Bra{u_j}_A$.  Since the Jamio{\l}kowski isomorphism is
  basis independent, the causal conditional state associated with the
  map $\mathcal{E}_{B|A}\left ( \cdot \right ) = U_{B|A} \left ( \cdot
  \right )\left ( U_{B|A} \right )^{\dagger}$ can be written in the
  eigenbasis of $\rho_A$ as
  \begin{align}
    \label{eq:CCS:UnitCond}
    \varrho_{B|A} & = \sum_{j,k} \Ket{u_j}\Bra{u_k}_A \otimes U_{B|A'}
    \Ket{u_k}\Bra{u_j}_{A'} \left (U_{B|A'}\right )^{\dagger} \\
    & = \sum_{j,k} \Ket{u_j}\Bra{u_k}_A \otimes \Ket{v_k}\Bra{v_j}.
  \end{align}
  It follows that the causal joint state for $AB$ is
  \begin{equation}
    \label{eq:CCS:UnitJoint}
    \varrho_{AB} = \sum_{j,k} \sqrt{p_j p_k} \Ket{u_j}\Bra{u_k}_A
    \otimes  \Ket{v_k}\Bra{v_j}_{B}.
  \end{equation}
  Note that the pure causal conditional state $\varrho_{B|A}$ is the
  partial transpose over $A$ of a pure acausal conditional state, as
  can be seen by comparison with eq.~\eqref{eq:ACS:PureCS} from
  example~\ref{exa:ACS:Pure}.  Also, the pure causal joint state
  $\varrho_{AB}$ is the partial transpose over $A$ of a pure acausal
  joint state $\rho_{AB} = \Ket{\psi}\Bra{\psi}_{AB}$, where
  \begin{equation}
    \Ket{\psi}_{AB} = \sum_j \sqrt{p_j} \Ket{u_j}_A \otimes \Ket{v_j}_B.
  \end{equation}
\end{example}

Causal joint states can be given an operational interpretation similar
to that of acausal joint states by specifying a procedure to perform
tomography on them (see \cite{Leifer2006,Leifer2007}).  The motivation
for introducing them here is that they allow quantum Bayesian
inference to be developed in a way that is blind to the distinction
between acausally and causally-related regions.  This is discussed in
\S\ref{Bayes}.

Note that whereas acausal joint states may involve more than two
acausally-related regions, causal joint states are thus far only
well-defined for \emph{two} causally-related regions.  The reasons for
this limitation are discussed in \S\ref{Lim}.

\subsection{Quantum-Classical Hybrid Regions}

\label{CCS:Hybrid}

An ensemble preparation procedure can be represented by a CPT map from
a classical variable to a quantum region and a measurement can be
represented as a CPT map from a quantum region to a classical
variable.  Therefore, these processes can be represented by
conditional states for which either the conditioned or conditioning
region is classical as a special case of
proposition~\ref{prop:CCS:Jamiol}.  In order to compare this to the
conventional formalism, we need to describe how composite regions
consisting of a classical variable and a quantum region are
represented in the conditional states formalism.  Such composites are
called quantum-classical hybrid regions.

As in example~\ref{exa:ACS:Class}, the classical variable $X$ is
associated with a Hilbert space $\Hilb[X]$ equipped with a preferred
basis $\left \{ \Ket{x} \right \}$ that represents the possible values
of $X$.  The quantum region is associated with a Hilbert space
$\Hilb[A]$ and the joint region with the tensor product $\Hilb[X A] =
\Hilb[X] \otimes \Hilb[A]$.  In order to preserve the classical nature
of $X$, states and conditional states on $\Hilb[X A]$ are restricted
to be of the following form:
\begin{definition}
  A \emph{hybrid operator} on $\Hilb[XA]$ is an operator of the form
  \begin{equation}
    \label{eq:CCS:HybridOp}
    M_{XA} = \sum_x \Ket{x}\Bra{x}_{X} \otimes M_x^A,
  \end{equation}
  where $\left \{ \Ket{x} \right \}$ is a preferred basis for
  $\Hilb[X]$ and $\{ M_x^A \}$ is a set of operators acting on
  $\Hilb[A]$, labelled by the values of $X$.  The operators $M_x^A$ are
  referred to as the \emph{components} of $M_{XA}$.
\end{definition}
When $M_{XA}$ is a state, this ensures that the reduced state on $X$
is diagonal in the preferred basis and that there can be no
entanglement between the quantum and classical regions.

Although our primary interest is in causal hybrids, since these are
relevant to preparations and measurements, one can also have acausal
hybrids.  As shown below, the set of acausal hybrid conditional states
and the set of acausal hybrid joint states are invariant under partial
transpose.  It follows that the set of acausal hybrid states and the
set of causal hybrid states are the same.  In particular, this means
that, unlike fully quantum causal states, hybrid causal states not
only have positive partial transpose but are themselves positive, and
unlike fully quantum acausal states, hybrid acausal states not only
are positive but also have positive partial transpose.

Therefore, for hybrid states, the notational distinction between
$\rho$ and $\varrho$ serves merely as a reminder of the causal
arrangement of the regions under consideration.  This in contrast with
the fully quantum case, where the distinction also has significance
for the mathematical properties of the operator.  When making claims
about the mathematical properties of hybrid conditional states that
are independent of causal structure, the notation $\sigma$ is used.
According to these conventions, any formula expressed in terms of
$\sigma$'s will yield a valid formula about hybrid states if the
$\sigma$'s are replaced by either $\rho$'s or $\varrho$'s.

For hybrid regions, there are two possible types of conditional state,
depending on whether the conditioning is done on the quantum or the
classical region.  In the case of conditioning on the classical
variable, a hybrid conditional state $\sigma_{A|X}$ is a positive
operator on $\Hilb[XA]$ satisfying
\begin{equation}
  \Tr[A]{\sigma_{A|X}} = I_X.
\end{equation}
In the case of conditioning on the quantum region, a hybrid
conditional state $\sigma_{X|A}$ is again a positive operator on
$\Hilb[XA]$, but this time satisfying
\begin{equation}
  \Tr[X]{\sigma_{X|A}} = I_A.
\end{equation}

\subsection{Ensemble Averaging as Belief Propagation}

\label{CCS:States}

A hybrid conditional state of the form $\sigma_{A|X}$ is a quantum
state conditioned on a classical variable. The causal interpretation
of such states is a process that takes a classical variable as input
and outputs a quantum state.  This is just an ensemble preparation
procedure.  In such a preparation procedure, a classical random
variable $X$ is sampled from a probability distribution $P(X)$ (by
flipping coins, rolling dice, or any other suitable method).
Depending on the value $x$ of $X$ thereby obtained, one of a set of
quantum states $\{\rho^A_x\}$ is prepared for a quantum region $A$.
If you do not know the value of $X$, then you should assign the
ensemble average state $\rho_A = \sum_x P(X=x) \rho^A_x$ to $A$. This
scenario is depicted in fig.~\ref{fig:CCS:QPreparation}. A quantum
preparation procedure has an obvious classical analogue wherein the
quantum region $A$ is replaced by a classical variable $R$ that is
prepared in one of a set of probability distributions $P(R|X=x)$
depending on the value $x$ of $X$. If you do not know the value of
$X$, then the classical belief propagation rule specifies that you
should assign the probability distribution $P(R)=\sum_X P(R|X)P(X)$ to
$R$.  This case is illustrated in fig.~\ref{fig:CCS:CPreparation}.
This section shows that a set of density operators can be represented
by a hybrid conditional state of the form $\sigma_{A|X}$ and that the
formula for the ensemble average state in a preparation procedure is a
special case of quantum belief propagation.

\begin{figure}[htb]
  \begin{minipage}{5.2cm}
    \raggedright
    \subfigure[]{
      \includegraphics[scale=0.4]{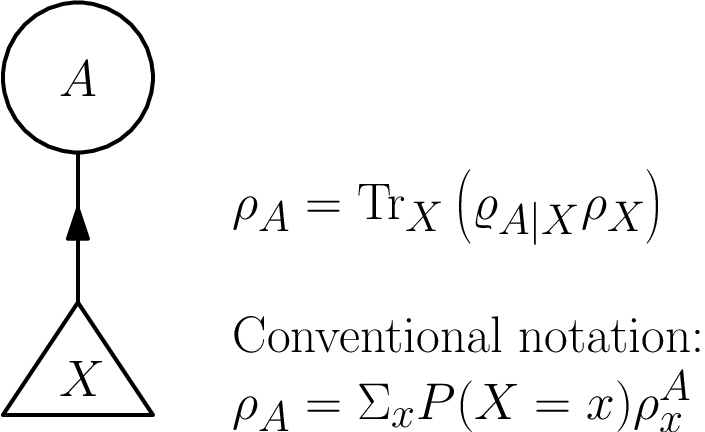}
      \label{fig:CCS:QPreparation}
    }
    \subfigure[]{
      \includegraphics[scale=0.4]{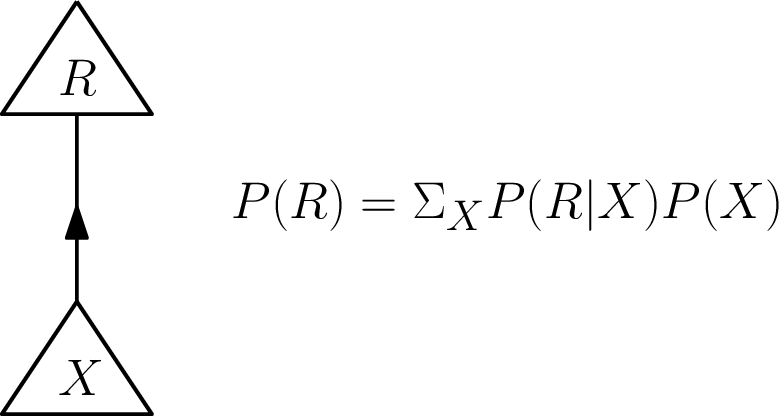}
      \label{fig:CCS:CPreparation}
    }
  \end{minipage}
  \caption{\label{fig:CCS:Preparation}Quantum and classical
    preparation procedures.  \subref{fig:CCS:QPreparation} A quantum
    preparation procedure is a process that takes a classical variable
    $X$ as input and outputs a quantum region $A$ in one of a set
    $\{\rho^A_x\}$ of states, depending on the value of $X$.  It is
    mathematically equivalent to the special case of a CPT map where
    the input is classical. \subref{fig:CCS:CPreparation} A classical
    preparation procedure is a process that takes a variable $X$ as
    input and outputs one of a set $\{P(R|X=x)\}$ of probability
    distributions over $R$, depending on the value of $X$.  It is
    mathematically equivalent to the stochastic dynamics depicted in
    fig.~\ref{fig:CCS:Stochastic}.}
\end{figure}

\begin{theorem}
  \label{thm:CCS:States}
  Let $\sigma_{A|X}$ be a hybrid operator, so that by
  eq.~(\ref{eq:CCS:HybridOp}) it can be written as
  \begin{equation}
    \label{eq:CCS:States}
    \sigma_{A|X} = \sum_x \rho^A_x \otimes \Ket{x}\Bra{x}_X,
  \end{equation}
  for some set of operators $\{ \rho^A_x\}$.  Then, $\sigma_{A|X}$
  satisfies the definition of both an acausal and a causal
  conditional state for $A$ given $X$ iff each of the components
  $\rho^A_x$ is a normalized state on $\Hilb[A]$.
\end{theorem}
The proof is provided in appendix~\ref{Appendix:Proofs}.

It is often convenient to use the notation $\sigma_{A|X=x} = \Bra{x}_X
\sigma_{A|X} \Ket{x}_X = \rho^A_x$ for the components of a conditional
state of this form.

If we adopt the convention that the partial transpose on the Hilbert
space of a classical variable is performed in its preferred basis then
it has no effect on hybrid operators.  Thus, acausal conditional
states of the form $\rho_{A|X}$ are invariant under partial transpose
on $X$ and this is why acausal and causal conditional states that are
conditioned on the classical variable have the same form\footnote{Even
  if we do not adopt the convention of evaluating partial transposes
  in the preferred basis, the sets of acausal and causal conditional
  states are still isomorphic.  If $\left \{ \Ket{x}_X \right \}$ is
  the preferred basis for acausal states then this amounts to choosing
  a different preferred basis $\left \{ \Ket{x^*}_X \right \}$ for
  causal states, where $^*$ is complex conjugation in the basis used
  to define the partial transpose.  However, this is an unnecessary
  complication that is avoided by adopting the recommended
  convention.}. The remainder of this section concerns the causal
interpretation of such states in terms of preparation procedures, so
we shift to the notation $\varrho_{A|X}$.

\begin{proposition}
  \label{prop:CCS:States}
  Let $\varrho_{A|X}$ be the causal hybrid conditional state with
  components given by a set of states $\{\rho^A_x\}$.  The ensemble
  average state arising from a preparation procedure that samples a
  value $x$ of a classical variable $X$ from the distribution $P(X)$
  and prepares the state $\rho^A_x$, is given by
  \begin{equation}
  \rho_A = \sum_x  P(X=x) \rho^A_x.
  \end{equation}
  This can be expressed in the conditional states formalism via the
  quantum belief propagation rule as
  \begin{equation}
    \label{eq:CCS:Ensemble}
    \rho_A = \Tr[X]{\varrho_{A|X} \rho_X},
  \end{equation}
  where $\rho_X = \sum_x P(X=x) \Ket{x}\Bra{x}_X$.
\end{proposition}
This result follows simply from substituting the definition of
$\rho_X$ and $\varrho_{A|X}$ into eq.~(\ref{eq:CCS:Ensemble}).

Fig.~\ref{fig:CCS:Preparation} and the second and seventh lines of
table~\ref{tbl:CCS:Conventional} summarize how the representation of a
preparation procedure within the conditional states formalism
contrasts with the conventional representation and how the latter
generalizes the analogous classical expression.

It should be noted that theorem~\ref{thm:CCS:States} can alternatively
be derived as a special case of theorem~\ref{thm:CCS:Jamiol} and
proposition~\ref{prop:CCS:States} as a special case of
proposition~\ref{prop:CCS:Jamiol}.  This follows from the fact that a
preparation procedure can be represented by a CPT map
$\mathcal{E}_{A|X}$ from a classical variable to a quantum region
(sometimes called a CQ map).  The map is defined on diagonal states
$\rho_X$ via
\begin{equation}
  \label{eq:Hybrid:CQ}
  \mathcal{E}_{A|X}(\rho_X) = \sum_x \Bra{x}_X \rho_X \Ket{x}_X \rho^A_x.
\end{equation}
By proposition~\ref{prop:CCS:Jamiol}, the conditional state associated
with the preparation is the Jamio{\l}kowski isomorphic operator to
this map. Eq.~\eqref{eq:CCS:Ensemble} is then obtained as a special
case of eq.~\eqref{eq:CCS:Evolve} where the input is classical.

\subsection{The Born Rule as Belief Propagation}

\label{CCS:POVM}

A hybrid conditional state of the form $\sigma_{Y|A}$ is a classical
probability distribution conditioned on a quantum region. The causal
interpretation of such states is as a process that takes a quantum
region as input and outputs a classical variable.  This is just a
measurement.  The most general kind of measurement on a quantum region
$A$ is conventionally represented by a POVM $\{E^A_y\}$ with the
classical variable $Y$ ranging over the possible outcomes.  If the
state of the region is $\rho_A$, then the probability of obtaining
outcome $y$ is given by the Born rule as $P(Y=y) = \Tr[A]{E^A_y
  \rho_A}$. This scenario is depicted in
fig.~\ref{fig:CCS:QMeasurement}.  In the classical analogue, the
quantum region $A$ is replaced by a classical variable $R$, the state
$\rho_A$ is replaced by a distribution $P(R)$ and the POVM is replaced
by a (possibly noisy) classical measurement described by a set of
response functions $\{ P(Y=y|R) \}$, i.e.\ a set of functions of $R$
labelled by the values of $Y$, where $P(Y=y|R=r)$ specifies the
probability of obtaining the outcome $y$ given that $R=r$.  The
overall probability of obtaining the outcome $y$ is given by
$P(Y=y)=\sum_R P(Y=y|R)P(R)$, which is just another instance of belief
propagation. This case is illustrated in
fig.~\ref{fig:CCS:CMeasurement}.  In analogy to this, the remainder of
this section shows that the components of a conditional state
$\sigma_{Y|A}$ form a POVM, and that the Born rule can be written as
quantum belief propagation with respect to a causal conditional state
of this form.

\begin{figure}[htb]
  \begin{minipage}{5.1cm}
    \raggedright
    \subfigure[]{
      \includegraphics[scale=0.4]{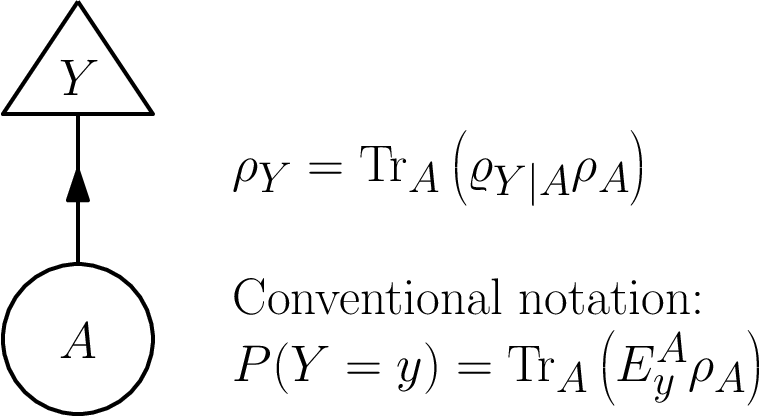}
      \label{fig:CCS:QMeasurement}
    }
    \subfigure[]{
      \includegraphics[scale=0.4]{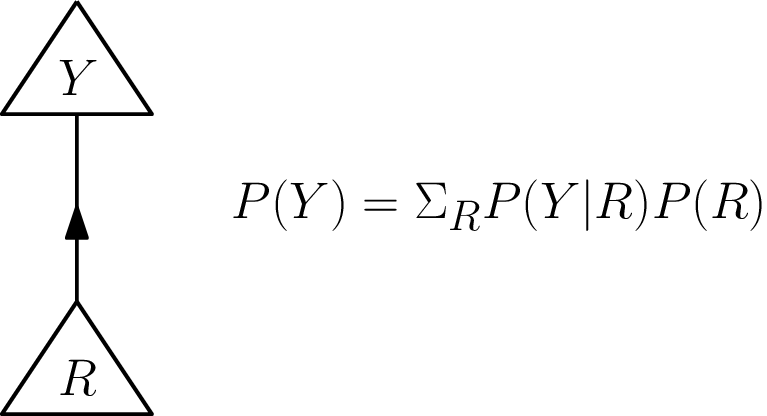}
      \label{fig:CCS:CMeasurement}
    }
  \end{minipage}
  \caption{\label{fig:CCS:Measurement}Quantum and classical
    measurements.  \subref{fig:CCS:QMeasurement} A quantum measurement
    is a process that takes a quantum region $A$ as input and outputs
    a classical variable $Y$.  It is mathematically equivalent to the
    special case of a CPT map where the output is classical.
    \subref{fig:CCS:CMeasurement} A classical (noisy) measurement is a
    process that takes a variable $R$ as input and outputs a variable
    $Y$ that depends on $R$, possibly in a coarse-grained or
    non-deterministic way.  It is mathematically equivalent to the
    stochastic dynamics depicted in fig.~\ref{fig:CCS:Stochastic}.}
\end{figure}

\begin{theorem}
  \label{thm:CCS:POVM}
  Let $\sigma_{Y|A}$ be a hybrid operator so that, by
  eq.~\eqref{eq:CCS:HybridOp}, it can be written in the form
  \begin{equation}
    \label{eq:CCS:POVM}
    \sigma_{Y|A} = \sum_y \Ket{y}\Bra{y}_Y \otimes E^A_y,
  \end{equation}
  for some set of operators $\{ E^A_y \}$.  Then, $\sigma_{Y|A}$
  satisfies the definition of both an acausal and a causal conditional
  state for $Y$ given $A$ iff the set $\{ E^A_y \}$ is a POVM on
  $\Hilb[A]$, i.e.\ each $E^A_y$ is positive and $\sum_y E^A_y = I_A$.
\end{theorem}
The proof is given in appendix~\ref{Appendix:Proofs}.

It is sometimes useful to use the notation $\sigma_{Y=y|A} = \Bra{y}_Y
\sigma_{Y|A} \Ket{y}_Y = E^A_y$ for the components of conditional
states of this form.

Unlike the case of hybrid states that are conditioned on a classical
variable, conditional states of the form $\sigma_{Y|A}$ are not
invariant under partial transpose on the conditioning region.
However, taking the partial transpose over $A$ of $\rho_{Y|A}$ yields
another valid acausal conditional state, because $\left \{ \left
    (E^A_y \right )^{T_A} \right \}$ is a POVM iff $\left \{ E^A_y
\right \}$ is.  The remainder of this section concerns the causal
interpretation of such states, so the notation $\varrho_{Y|A}$ is
adopted.

\begin{proposition}
  Consider a measurement of a POVM $\{E^A_y\}$ on a quantum region $A$
  in state $\rho_A$.  Let $\varrho_{Y|A}$ be the causal conditional
  state with components $E^A_y$.  The Born rule,
  \begin{equation}
    P(Y=y) = \Tr[A]{E^A_y \rho_A},
  \end{equation}
  can then be expressed in the conditional states formalism as the
  quantum belief propagation rule
  \begin{equation}
    \label{eq:CCS:CondBorn}
    \rho_Y = \Tr[A]{\varrho_{Y|A} \rho_A},
  \end{equation}
  where $\rho_Y = \sum_y P(Y=y) \Ket{y}\Bra{y}_Y$.
\end{proposition}
This is easily verified by substituting the definition of
$\varrho_{Y|A}$ from eq.~(\ref{eq:CCS:POVM}) into
eq.~(\ref{eq:CCS:CondBorn}).

This representation of a measurement as a causal conditional state and
of the Born rule as an instance of belief propagation is summarized in
fig.~\ref{fig:CCS:Measurement} and the third and sixth lines of
table~\ref{tbl:CCS:Conventional}.

Once again, these results can be understood as a special case of
theorem~\ref{thm:CCS:Jamiol} and proposition~\ref{prop:CCS:Jamiol} by
recognizing that a POVM may be represented as a map from a quantum
region to a classical variable (sometimes called a QC map).
Specifically, if the probability distribution $P(Y)$ is represented by
a diagonal state $\rho_Y$, then the measurement can be represented by
the CPT map $\mathcal{E}_{Y|A}$ defined by
 \begin{equation}
   \label{eq:Hybrid:QC}
   \rho_Y = \mathcal{E}_{Y|A}(\rho_A) = \sum_y \Tr[A]{E^A_y
     \rho_A}\Ket{y}\Bra{y}_Y.
\end{equation}
The causal conditional state $\varrho_{Y|A}$ appearing in
eq.~\eqref{eq:CCS:CondBorn} is simply the Jamio{\l}kowski isomorphic
operator to this map.

\subsection{Belief Propagation of Conditional States}

\label{CCS:CondBP}

Consider three causally-related regions $A$, $B$ and $C$, such that
$B$ is in the future of $A$ and $C$ is in the future of $B$.  If the
dynamics is Markovian then it can be described by first applying a CPT
map $\mathcal{E}_{B|A}$ to $A$ followed by a CPT map
$\mathcal{E}_{C|B}$ to $B$.  This scenario is illustrated in
fig.~\ref{fig:CCS:QCondBP}.  If we are only interested in regions $A$
and $C$ then region $B$ can be eliminated from the description by
composing the two maps to obtain $\mathcal{E}_{C|A} =
\mathcal{E}_{C|B} \circ \mathcal{E}_{B|A}$, where $\mathcal{E} \circ
\mathcal{F} \left ( \cdot \right ) \equiv \mathcal{E} \left (
  \mathcal{F} \left ( \cdot \right ) \right )$.

\begin{figure}[htb]
  \begin{minipage}{5cm}
    \raggedright
    \subfigure[]{
      \includegraphics[scale=0.4]{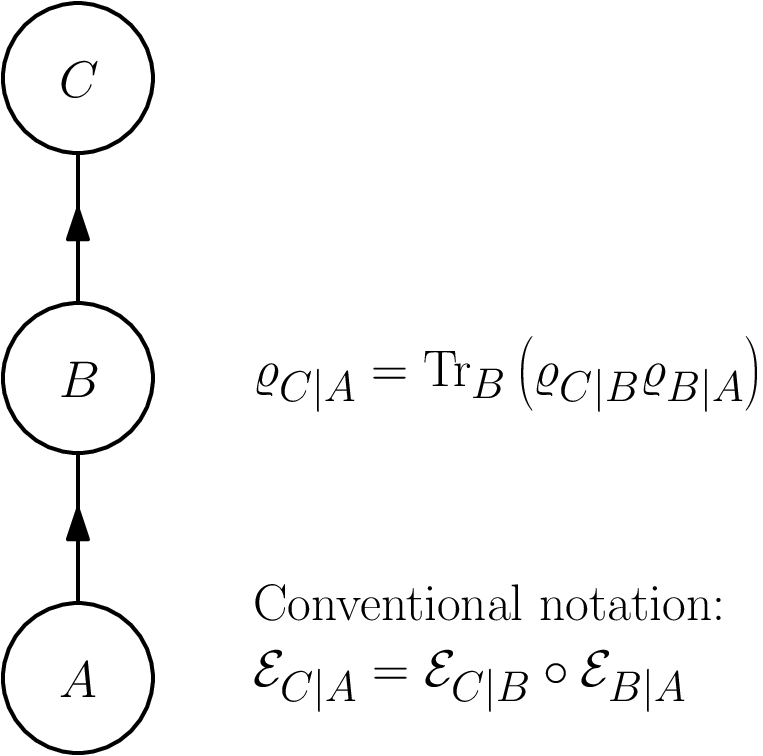}
      \label{fig:CCS:QCondBP}
    }
    \subfigure[]{
      \includegraphics[scale=0.4]{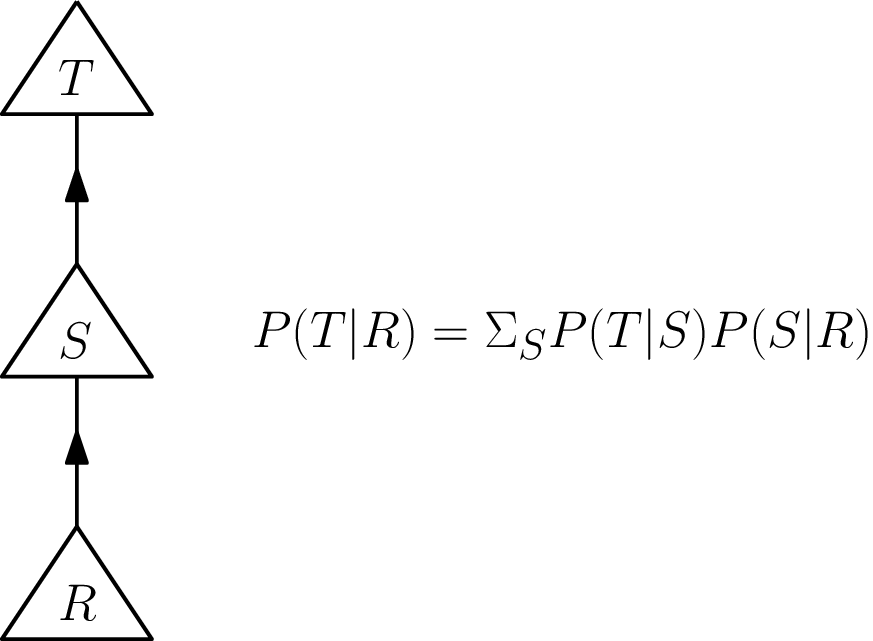}
      \label{fig:CCS:CCondBP}
    }
  \end{minipage}
  \caption{\label{fig:CCS:CondBP}Propagating causal conditional states
    and conditional probability distributions.
    \subref{fig:CCS:QCondBP} Quantum case.  \subref{fig:CCS:CCondBP}
    Classical case.}
\end{figure}

In the conditional states formalism the CPT maps are replaced by the
Jamio{\l}kowski isomorphic causal conditional states $\varrho_{B|A}$,
$\varrho_{C|B}$ and $\varrho_{C|A}$, and thus $\mathcal{E}_{C|A} =
\mathcal{E}_{C|B} \circ \mathcal{E}_{B|A}$ should be replaced by a
formula for $\varrho_{C|A}$ in terms of $\varrho_{B|A}$ and
$\varrho_{C|B}$.

As an aid to intuition, consider the classical analogue of this scenario
as depicted in fig.~\ref{fig:CCS:CCondBP}.  Here, the variable $S$ is
in the future of $R$, and $T$ is in the future of $S$.  The three
variables are related by a Markovian dynamics, described by the
conditional probability distributions $P(S|R)$ and $P(T|S)$.  An
initial probability distribution $P(R)$ can be propagated into the
future to obtain $P(T)$ in two steps.  First we propagate from $R$ to
$S$ to obtain $P(S) = \sum_R P(S|R)P(R)$, and then from $S$ to $T$ to
obtain $P(T)= \sum_{S}P(T|S)P(S)$.  Combining these steps gives $P(T)
= \sum_{R,S} P(T|S)P(S|R)P(R)$, so defining the conditional
probability distribution
\begin{equation}
  P(T|R) = \sum_S P(T|S)P(S|R),
\end{equation}
allows the belief propagation from $R$ to $T$ to be performed in a
single step via $P(T) = \sum_R P(T|R)P(R)$.

The quantum analogue of this is given by the following theorem.
\begin{theorem}
  \label{thm:CCS:CondBP}
  Let $\mathcal{E}_{B|A}$, $\mathcal{E}_{C|B}$ and $\mathcal{E}_{C|A}$
  be linear maps such that $\mathcal{E}_{C|A} = \mathcal{E}_{C|B}
  \circ \mathcal{E}_{B|A}$.  Then, the Jamio{\l}kowski isomorphic
  operators, $\varrho_{B|A}$, $\varrho_{C|B}$ and $\varrho_{C|A}$
  satisfy
  \begin{equation}
    \label{eq:CCS:CondBP}
    \varrho_{C|A} = \Tr[B]{\varrho_{C|B}\varrho_{B|A}}.
  \end{equation}
  Conversely, if three operators satisfy eq.~\eqref{eq:CCS:CondBP},
  then the Jamio{\l}kowski isomorphic maps satisfy $\mathcal{E}_{C|A}
  = \mathcal{E}_{C|B} \circ \mathcal{E}_{B|A}$.
\end{theorem}

The proof is provided in appendix~\ref{Appendix:Proofs}.

Eq.~\eqref{eq:CCS:CondBP} can be regarded as a belief propagation rule
for causal conditional states.  It propagates beliefs about $B$,
conditional on $A$, into the future to obtain beliefs about $C$,
conditional on $A$.

A similar formalism can be developed for the propagation of acausal
conditional states across acausally-related regions.  However, for
present purposes, it is more interesting to consider a situation of
mixed causality, wherein a causal conditional state is propagated
across two acausally-related regions.  This is used in the application
to steering developed in \S\ref{Cond:Remote}.

Consider the scenario depicted in fig.~\ref{fig:CCS:AcausalCausal}.
Initially, a state $\rho_{AB}$ is assigned to regions $A$ and $B$,
which are acausally-related.  A CPT map $\mathcal{E}_{C|B}$
(alternatively represented by a causal conditional state
$\varrho_{C|B}$) is then applied to region $B$ to obtain the state of
$AC$.  In this scenario, $C$ is causally-related to $B$, but acausally
related to $A$.  By theorem~\ref{thm:ACS:Jamiol}, the state of $AC$ is
given by
\begin{equation}
  \rho_{AC} = \Tr[B]{\varrho_{C|B}\rho_{AB}}.
\end{equation}
Now, $\rho_{AB} = \rho_{B|A}\Sprod \rho_A$ and $\rho_{AC} = \rho_{C|A}
\Sprod \rho_A$, so we have
\begin{equation}
  \rho_{C|A} \Sprod \rho_A = \Tr[B]{\varrho_{C|B}\rho_{B|A}} \Sprod \rho_A,
\end{equation}
where we have used the fact that $\rho_A$ commutes with
$\varrho_{C|B}$.  Taking the $\Sprod$-product of this equation with
$\rho_A^{-1}$ then gives
\begin{equation}
  \label{eq:CCS:AcausalCausal}
  \rho_{C|A} = \Tr[B]{\varrho_{C|B}\rho_{B|A}},
\end{equation}
in analogy with eq.~\eqref{eq:CCS:CondBP}.  This can again be viewed
as a belief propagation rule for conditional states, but this time an
acausal conditional state is being propagated through a causal
conditional state.

\begin{figure}[htb]
  \includegraphics[scale=0.4]{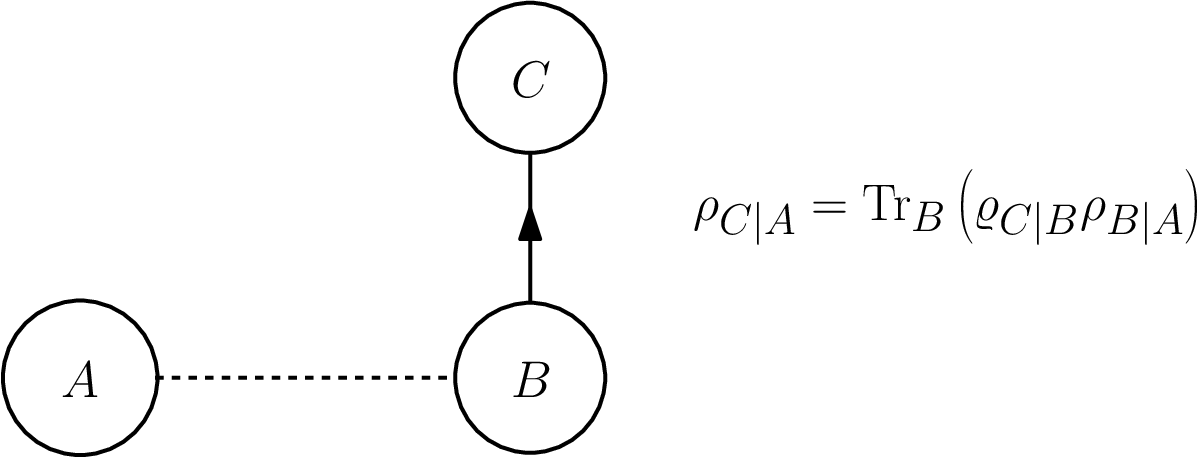}
  \caption{\label{fig:CCS:AcausalCausal}Propagating an acausal
    conditional state through a causal conditional state.}
\end{figure}

\subsection{The Heisenberg Picture}

\label{CCS:Heis}

In \S\ref{CCS:CPT}, quantum evolution from an early region $A$ to a
late region $B$ was described by a map from states on $\Hilb[A]$ to
states on $\Hilb[B]$.  However, dynamics can alternatively be
represented in terms of observables rather than states.  This is
simply the textbook distinction between the Schr\"{o}dinger picture
and the Heisenberg picture.  In the Heisenberg picture, a temporal
evolution is described by a map from the space of observables on the
late region $B$ to the space of observables on the early region $A$.
The observables are usually represented by self-adjoint operators and
the dynamics by unitary operations, but this can be generalized to
take account of generalized measurements and CPT dynamics.  In this
generalization, Heisenberg dynamics consists of a map from POVM
elements (also known as effects) on $B$ to POVM elements on $A$.  In
other words, effects are evolved backwards in time in the Heisenberg
picture.

In order to describe the Heisenberg picture for CPT maps, it is
necessary to define the notion of a dual map.
\begin{definition}
  \label{def:CCS:Dual}
  The \emph{dual map} $\left ( \mathcal{E}_{B|A} \right
  )^{\dagger}:\Lin[B] \rightarrow \Lin[A]$ of a linear map
  $\mathcal{E}_{B|A}:\Lin[A] \rightarrow \Lin[B]$ is the unique map
  that satisfies
  \begin{equation}
    \Tr[A]{\left ( \mathcal{E}_{B|A} \right )^{\dagger}
      \left (N_B \right ) M_A} = \Tr[B]{N_B \mathcal{E}_{B|A} \left (
        M_A \right )}
  \end{equation}
  for all $M_A \in \Lin[A], N_B \in \Lin[B]$.
\end{definition}

Note that the input space for the dual map is the output space for the
original map and vice versa.  The notational convention:
\begin{equation}
  \mathcal{E}^{\dagger}_{A|B} \equiv \left ( \mathcal{E}_{B|A}
  \right )^{\dagger},
\end{equation}
is adopted in order to make this clear.

If $\mathcal{E}_{B|A}$ is the map describing time evolution in the
Schr\"{o}dinger picture through the formula
$\rho_B=\mathcal{E}_{B|A}(\rho_A)$, then the same evolution is
described in the Heisenberg picture by the dual map
$\mathcal{E}^{\dagger}_{A|B}$ through the formula
$E^A_y=\mathcal{E}^{\dag}_{A|B}(E^B_y)$ where $\{ E_y^A \}$ and $\{
E_y^B \}$ are POVMs. This follows from the condition that the two
pictures should be operationally equivalent, i.e.\ they should assign
the same probabilities.  To see this, imagine that the evolution from
region $A$ to $B$ is followed by a measurement on $B$ yielding an
outcome $Y$.  This scenario is depicted in
fig.~\ref{fig:CCS:TimeStepHeis}. The probability of observing the
effect $E^B_y$ after a preparation of $\rho_A$ followed by an
evolution $\mathcal{E}_{B|A}$ is expressed in the Schr\"{o}dinger
picture as $\Tr[B]{E^B_y \rho_B}$ where
$\rho_B=\mathcal{E}_{B|A}(\rho_A)$, while it is expressed in the
Heisenberg picture as $\Tr[A]{E^A_y \rho_A}$ where
$E^A_y=\mathcal{E}^{\dagger}_{A|B}(E^B_y)$.  The definition of the
dual map ensures that the two expressions for the probability are
equivalent, i.e. $\Tr[B]{E^B_y \rho_B} = \Tr[A]{E^A_y \rho_A}$.

A CPT map $\mathcal{E}_{B|A}$ can always be written in a Kraus
decomposition
\begin{equation}
  \mathcal{E}_{B|A} \left (\cdot \right ) = \sum_{\mu} K_{\mu} \left (
    \cdot \right ) K_{\mu}^{\dagger},
\end{equation}
where $K_{\mu}:\Hilb[A] \rightarrow \Hilb[B]$.  The dual map
$\mathcal{E}_{A|B}^{\dagger}$ can then be obtained by taking the
adjoint of the operators in the Kraus decomposition, i.e.,
\begin{equation}
  \mathcal{E}_{A|B}^{\dagger} \left (\cdot \right ) = \sum_{\mu}
  K_{\mu}^{\dagger} \left ( \cdot \right ) K_{\mu}.
\end{equation}
Thus, if $\Hilb[A]$ and $\Hilb[B]$ are isomorphic and
$\mathcal{E}_{B|A}$ is a unitary operation, i.e.
$\mathcal{E}_{B|A}(\cdot)=U (\cdot) U^{\dag}$ for some unitary
operator $U$, then $\mathcal{E}^{\dag}_{A|B}(\cdot)=U^{\dag}(\cdot)
U$.  This is the familiar special case of the Heisenberg picture for
unitary dynamics.

In order to translate the Heisenberg picture into the conditional
states formalism, first represent the POVM $\{E^B_y\}$ by a
conditional state $\varrho_{Y|B}$, the CPT map $\mathcal{E}_{B|A}$ by
a causal conditional state $\varrho_{B|A}$, and the POVM $\{E^A_y\}$
by a conditional state $\varrho_{Y|A}$.  Secondly, note that
fig.~\ref{fig:CCS:TimeStepHeis} is just a special case of
fig.~\ref{fig:CCS:QCondBP} from \S\ref{CCS:CondBP} in which the final
region is classical, so the three conditional states are related by
eq.~\eqref{eq:CCS:CondBP}, i.e.
\begin{equation}
  \label{eq:CCS:CondHeis}
  \varrho_{Y|A} = \Tr[B]{\varrho_{Y|B}\varrho_{B|A}}.
\end{equation}
In \S\ref{CCS:CondBP}, this was described as a belief propagation
formula for causal conditional states because $\varrho_{Y|B}$ was
regarded as defining a map from $\varrho_{B|A}$ to $\varrho_{Y|A}$,
propagating beliefs about $B$, conditional on $A$, into the future.
However, in the context of Heisenberg dynamics, we instead regard
$\varrho_{B|A}$ as defining a map from $\varrho_{Y|B}$ to
$\varrho_{Y|A}$, in the opposite direction to the flow of time.  It
remains to show that eq.~\eqref{eq:CCS:CondHeis} is equivalent to the
conventional description of Heisenberg dynamics in terms of dual maps.

\begin{figure}[htb]
  \includegraphics[scale=0.4]{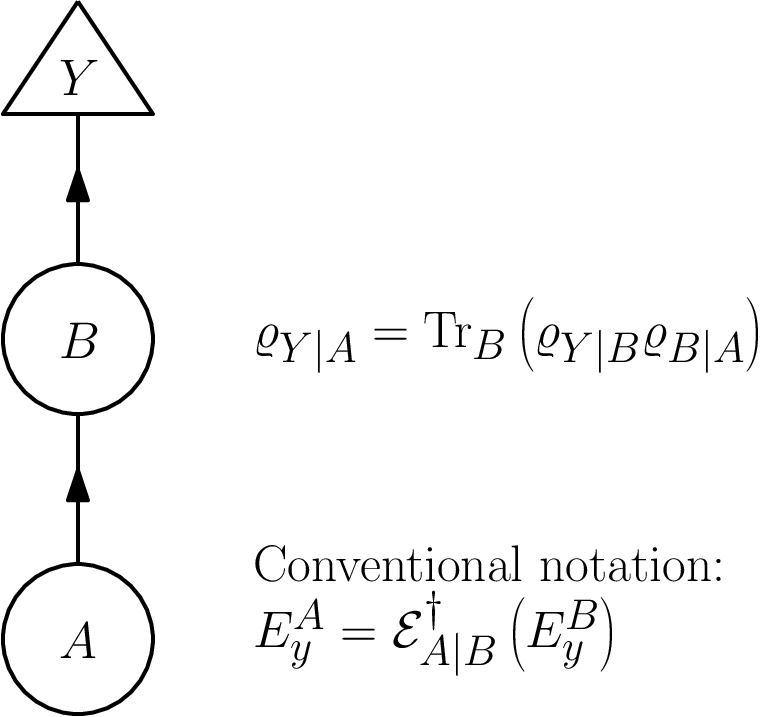}
  \caption{\label{fig:CCS:TimeStepHeis} Dynamics in the Heisenberg
    picture.}
\end{figure}

\begin{theorem}
  \label{thm:CCS:Duals}
  Let $\varrho_{B|A}$ be the causal joint state that is
  Jamio{\l}kowski isomorphic to the CPT map $\mathcal{E}_{B|A}$.
  Then, the action of the dual map $\mathcal{E}^{\dagger}_{A|B}$ on an
  operator $M_{BC}$ is given by
  \begin{equation}
    \label{eq:CCS:DualMap}
    (\mathcal{E}^{\dagger}_{A|B}\otimes \mathcal{I}_C) \left (
      M_{BC} \right ) = \Tr[B]{M_{BC} \varrho_{B|A}}
  \end{equation}
\end{theorem}
\begin{proof}
  By definition~\ref{def:CCS:Dual}, the dual map to
  $\mathcal{E}_{B|A}$ is the unique linear map
  $\mathcal{E}^{\dagger}_{A|B}$ that satisfies
  \begin{equation}
    \Tr[A]{\mathcal{E}^{\dagger}_{A|B} \left ( N_B \right )M_A} =
    \Tr[B]{N_B \mathcal{E}_{B|A} \left ( M_A\right )}
  \end{equation}
  for all operators $M_A$ and $N_B$.  Using the Jamio{\l}kowski
  isomorphism, theorem~\ref{thm:CCS:Jamiol}, the right hand side can
  be written as
  \begin{align}
    \Tr[B]{N_B \mathcal{E}_{B|A} \left ( M_A\right )} & = \Tr[AB]{N_B
      \varrho_{B|A} M_A} \\
    & = \Tr[A]{\Tr[B]{N_B \varrho_{B|A}}M_A}.
  \end{align}
  The only way this can equal $\Tr[A]{ \mathcal{E}^{\dagger}_{A|B}
    \left ( N_B \right )M_A}$ for all $M_A$ is if
  $\mathcal{E}^{\dagger}_{A|B} \left ( N_B \right ) = \Tr[B]{N_B
    \varrho_{B|A}}$.  Eq.~\eqref{eq:CCS:DualMap} then follows by
  linear extension to $\Hilb[BC]$.
\end{proof}

Combining this with eq.~\eqref{eq:CCS:CondHeis} gives the following
proposition.
\begin{proposition}
  \label{prop:CCS:Heis}
  Let $\varrho_{B|A}$ be the causal conditional state associated with
  a quantum evolution described by the CPT map $\mathcal{E}_{B|A}$ and
  let $\varrho_{Y|A}$ and $\varrho_{Y|B}$ be the hybrid conditional
  states associated with the POVMs $\{ E_y^A \}$ and $\{ E_y^B \}$,
  such that $\{E^A_y\}$ is obtained from $\{ E^B_y\}$ by the
  Heisenberg picture dynamics.  The conventional description of
  evolution in the Heisenberg picture,
  \begin{equation}
    \label{eq:CCS:HeisDynamics}
    E^A_y = \mathcal{E}^{\dag}_{A|B} \left ( E^B_y \right ),
 \end{equation}
 can be expressed in the conditional states formalism as
 \begin{equation}
   \label{eq:CCS:HeisEvolve}
   \varrho_{Y|A} = \Tr[B]{\varrho_{Y|B} \varrho_{B|A}}.
 \end{equation}
\end{proposition}
This follows straightforwardly from theorem~\ref{thm:CCS:Duals} and
theorem~\ref{thm:CCS:POVM}.

As in \S\ref{CCS:CondBP}, similar reasoning can be applied to other
causal scenarios.  Consider the special case of
fig.~\ref{fig:CCS:AcausalCausal} in which $C$ is replaced by a
classical variable $Y$.  This is depicted in
fig.~\ref{fig:CCS:Remote}.  Two acausally-related regions, $A$ and
$B$, are assigned a state $\rho_{AB}$ and then the POVM $\left \{
  E^B_y \right \}$ (alternatively represented by a conditional state
$\varrho_{Y|B}$) is measured on region $B$.  This is the type of
scenario that occurs in an EPR experiment.  By measuring the region
$B$, information is obtained about the remote region $A$ and we are
interested in how the state of $A$ is correlated with the measurement
outcome $Y$.

\begin{figure}[htb]
  \includegraphics[scale=0.4]{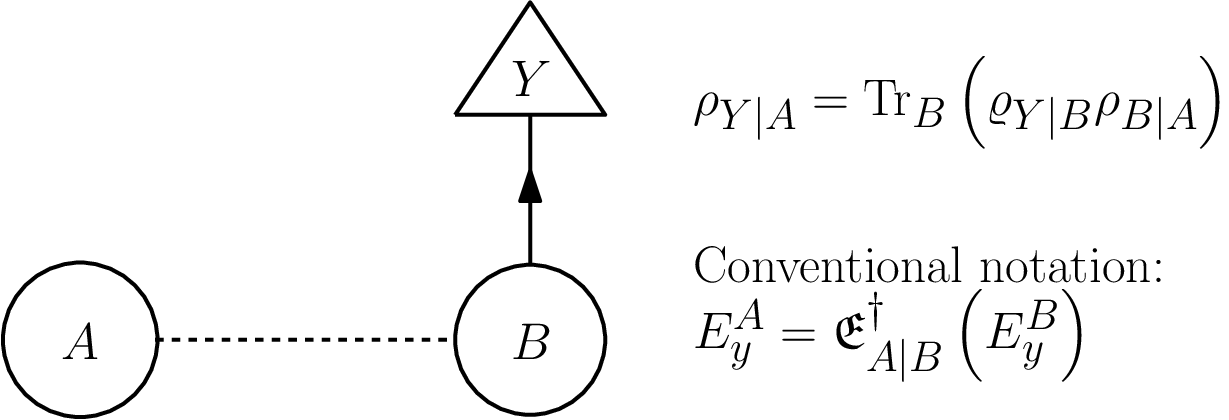}
  \caption{\label{fig:CCS:Remote}Heisenberg evolution for a remote
    measurement.}
\end{figure}

The belief propagation formula for this scenario,
eq.~\eqref{eq:CCS:AcausalCausal}, gives
\begin{equation}
  \label{eq:CCS:CondHeisAC}
  \rho_{Y|A} = \Tr[B]{\varrho_{Y|B}\rho_{B|A}}.
\end{equation}
The components of the conditional state $\rho_{Y|A}$ specify a POVM
$\left \{ E^A_y \right \}$.  This POVM can be thought of as describing
the effective measurement that gets performed on $A$ when we actually
measure region $B$.  When combined with $\rho_A$, the conditional
state $\rho_{Y|A}$ specifies the ensemble of states for region $A$
associated with the different measurement outcomes via $\rho_{AY} =
\rho_{Y|A} \Sprod \rho_A$.  In terms of components, this is
\begin{equation}
  \rho_{YA} = \sum_Y \Ket{y}\Bra{y}_Y \otimes \rho_A^{\frac{1}{2}}
  E^A_y \rho_A^{\frac{1}{2}},
\end{equation}
so the unnormalized state of $A$ corresponding to the outcome $Y=y$ is
\begin{equation}
  P(Y=y)\rho^A_y = \rho_A^{\frac{1}{2}} E^A_y \rho_A^{\frac{1}{2}},
\end{equation}
where $P(Y=y) = \Tr[A]{E^A_y\rho_A} = \Tr[B]{E^B_y \rho_B}$ is the
Born rule probability for the measurement outcome $Y=y$.

If $\rho_{B|A}$ in eq.~\eqref{eq:CCS:CondHeisAC} is thought of as
specifying a map from $\varrho_{Y|B}$ to $\rho_{Y|A}$, then this map
is analogous to a Heisenberg picture dynamics, except that the
propagation is across acausally-related rather than causally-related
regions.  If $\mathfrak{E}_{B|A}$ is the Jamio{\l}kowski isomorphic
map to $\rho_{B|A}$ then, by theorem~\ref{thm:CCS:Duals}, the POVM
elements are related by
\begin{equation}
  \label{eq:CCS:HeisAC}
  E^A_y = \mathfrak{E}_{A|B}^{\dagger} \left ( E^B_y \right ).
\end{equation}
Mathematically, the only difference between this and a Heisenberg
picture map is that $\mathfrak{E}_{B|A} \circ {T_A}$ is completely
positive rather than the map $\mathfrak{E}_{B|A}$ itself.  This is
just a reflection of the fact that we are propagating across acausally
related, rather than causally-related, regions.

A similar expression to eq.~\eqref{eq:CCS:HeisAC} has appeared in the
context of quantum steering \cite{Verstraete2002}, although there it
is written in terms of the Choi, rather than Jamio{\l}kowski,
isomorphic map so there is a transpose in the expression.  We develop
this application in \S\ref{Cond:Remote}.  To our knowledge, the
connection to Heisenberg dynamics has not previously been recognized.

\subsection{Quantum Instruments as Causal Belief Propagation}

\label{CCS:QI}

Describing a measurement by a POVM is adequate for determining the
outcome probabilities of the measurement via the Born rule.  However,
one might also wish to describe how the state of a post-measurement
region is correlated with the measurement result. In the conventional
formalism, the transformative aspect of a measurement is represented
by a quantum instrument.
\begin{definition}
  \label{def:quantuminstrument}
  Given quantum regions $A$ and $B$, and a classical variable $Y$, a
  \emph{quantum instrument} is a set $\{\mathcal{E}^{B|A}_y\}$ of CPT
  maps $\mathcal{E}^{B|A}_y:\Lin[A]\to\Lin[B]$ such that the operators
  \begin{equation}
    \label{eq:CCS:QIDef}
    E^A_y = \mathcal{E}^{\dagger}{}^{A|B}_y \left ( I_B \right )
  \end{equation}
  form a POVM.
\end{definition}

If $A$ and $B$ represent causally-related regions before and after a
measurement and $Y$ represents the measurement outcome, then a quantum
instrument can be used to determine the subnormalized state
$P(Y=y)\rho^B_y$ of $B$ when the outcome is known via
\begin{equation}
  \label{eq:CCS:QInstrument}
  P(Y=y) \rho^B_y = \mathcal{E}^{B|A}_y \left ( \rho_{A} \right ).
\end{equation}
It can also be used to compute the outcome probabilities for the
measurement by simply tracing over $B$ in
eq.~\eqref{eq:CCS:QInstrument} to obtain
\begin{equation}
  \label{eq:CCS:QInstrumentProb}
  P(Y=y) = \Tr[B]{\mathcal{E}^{B|A}_y \left ( \rho_{A} \right )}.
\end{equation}
Using eq.~\eqref{eq:CCS:QIDef}, this can be written as
\begin{align}
  P(Y=y) & = \Tr[B]{\mathcal{E}^{B|A}_y \left ( \rho_{A} \right )} \\
  & = \Tr[B]{I_B\mathcal{E}^{B|A}_y \left ( \rho_{A} \right )} \\
  & = \Tr[A]{\mathcal{E}^{\dagger}{}^{A|B}_y \left ( I_B \right )
    \rho_A} \\
  & = \Tr[A]{E^A_y \rho_A},
\end{align}
which is just the Born rule with respect to the POVM defined by the
instrument.

Whilst an instrument defines a unique POVM, each POVM corresponds to
more than one quantum instrument.  When performing a measurement of a
particular POVM, \emph{any} of the quantum instruments that correspond
to it via eq.~\eqref{eq:CCS:QIDef} may be obtained, depending on how
the measurement is implemented.  Conversely, the set of instruments
corresponding to a given POVM exhaust the possible post-measurement
transformations.  This includes, for example, the situation in which
the system being measured is absorbed by the detector, which
corresponds to choosing the trivial Hilbert space for $B$,
i.e. $\Hilb[B] = \mathbb{C}$.

Despite the freedom in choosing a quantum instrument, certain kinds of
instrument are usually considered particularly important.  For
measurements associated with a \emph{projector}-valued measure
$\{\Pi^A_y\}$, the possible quantum instruments include the
L{\"u}ders-von Neumann projection postulate as a special case by
taking $\Hilb[B]$ to have the same dimension as $\Hilb[A]$ and
$\mathcal{E}^{B|A}_y (\rho_A) = \mathcal{I}_{B|A} \left (\Pi^A_y
  \rho_A \Pi^A_y \right)$, where $\mathcal{I}_{B|A}$ is an isometry
between $\Hilb[A]$ and $\Hilb[B]$.  For general POVMs the rule
$\mathcal{E}^{B|A}_y (\rho_A) = \mathcal{I}_{B|A} \left ( \left (
    E^A_y \right )^{\frac{1}{2}} \rho_A \left ( E^A_y \right
  )^{\frac{1}{2}} \right )$, which is sometimes taken as a natural
generalization of the projection postulate, is also included as a
special case.

A general measurement procedure, where there is an initial quantum
region, a classical outcome of the measurement, and a quantum region
after the measurement, is depicted in fig.~\ref{fig:CCS:QInstrument}.
Note that the final quantum region $B$ may depend causally both on the
initial quantum region $A$ and on the outcome $Y$.

In order to understand how quantum instruments are represented in the
conditional states formalism, it is helpful to first look at the
classical analogue.  This is a scenario wherein a classical measurement
is made upon a classical system, which persists after the measurement,
and in general the measurement procedure is permitted to disturb the
state of the system.  The variable $R$ describes the system before the
measurement and the variable $S$ describes the system after the
measurement.  This is in line with the quantum treatment, in which
distinct regions are given distinct labels.  The outcome of the
measurement is again denoted by $Y$. This scenario is depicted in
fig.~\ref{fig:CCS:CInstrument}.  The measurement is then described by
a conditional probability $P(Y,S|R)$.  This can equivalently be
thought of as a set of subnormalized conditional probabilities for $S$
given $R$, $\{ P(Y=y,S|R) \}$, one for each outcome $y$, which is the
analogue of a quantum instrument. The joint distribution over $Y$ and
$S$, when the input distribution is $P(R)$, is then given by
\begin{equation}
  \label{eq:CCS:CInstrument}
  P(Y,S)=\sum_R P(Y,S|R) P(R).
\end{equation}
Furthermore, the set of response functions $\{P(Y=y|R)\}$ associated
with such a measurement is easily computed from $P(Y,S|R)$ by
marginalizing over $S$,
\begin{equation}
  \label{eq:CCS:CInstrumentRespfn}
  P(Y|R)=\sum_S P(Y,S|R).
\end{equation}

In the conditional states formalism, eqs.~\eqref{eq:CCS:QInstrument}
and \eqref{eq:CCS:QIDef} are replaced by straightforward analogise of
eqs.~\eqref{eq:CCS:CInstrument} and \eqref{eq:CCS:CInstrumentRespfn}
for a causal hybrid state $\varrho_{YB|A}$ of a classical variable $Y$
and quantum region $B$, conditioned on a quantum region $A$.

\begin{figure}[htb]
  \begin{minipage}{7.4cm}
    \raggedright
    \subfigure[]{
      \includegraphics[scale=0.4]{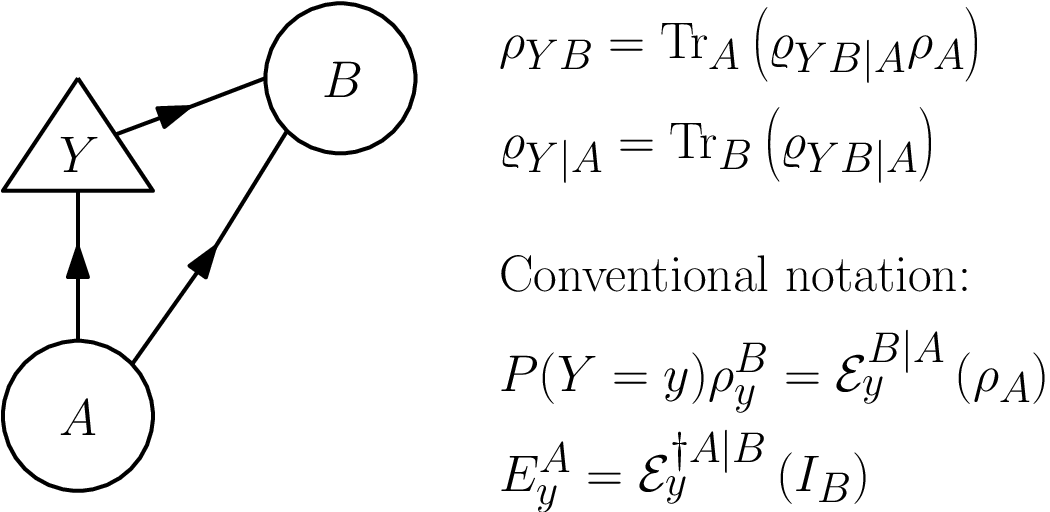}
      \label{fig:CCS:QInstrument}
    }
    \subfigure[]{
      \includegraphics[scale=0.4]{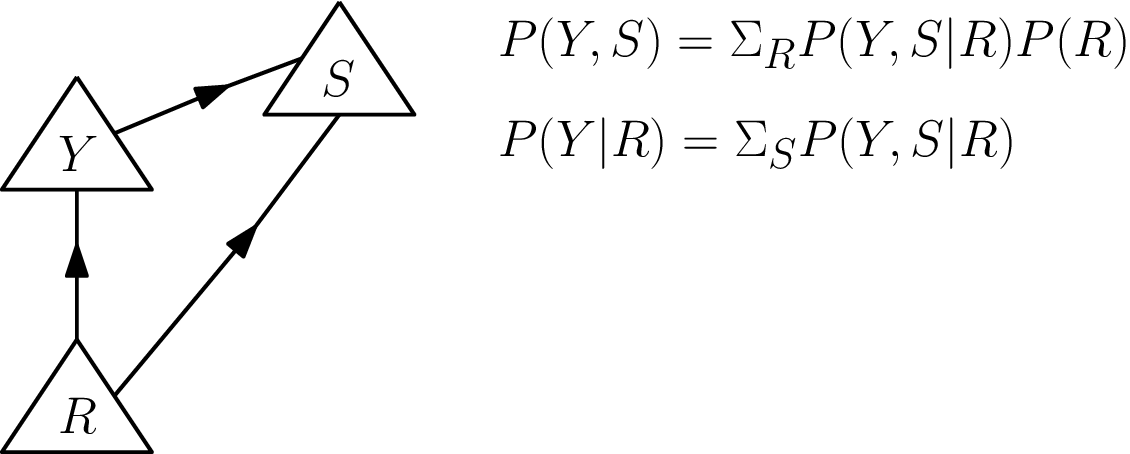}
      \label{fig:CCS:CInstrument}
    }
  \end{minipage}
  \caption{\label{fig:CCS:Instrument} Measurements and their
    associated state-update rules.  \subref{fig:CCS:QInstrument} A
    quantum instrument, representing how the state of a quantum
    persistent system changes after a measurement.
    \subref{fig:CCS:CInstrument} The classical analogue of a quantum
    instrument, representing how the state of a classical persistent
    system changes after a general (possibly disturbing) measurement.}
\end{figure}

\begin{theorem}
  \label{thm:CCS:Instrument}
  Let $\varrho_{YB|A}$ be an operator on $\Hilb[YAB]$ of the form
  \begin{equation}
    \varrho_{YB|A} \equiv \sum_y \Ket{y}\Bra{y}_{Y} \otimes \varrho_{Y=y,B|A},
  \end{equation}
  where
  \begin{equation}
    \varrho_{Y=y,B|A}\equiv
    \mathcal{E}_y^{B|A'} \left ( \sum_{j,k}
      \Ket{j}\Bra{k}_A \otimes \Ket{k}\Bra{j}_{A'} \right )
  \end{equation}
  are the Jamio{\l}kowski isomorphic operators to maps
  $\mathcal{E}^{B|A}_y: \Lin[A] \to \Lin[B]$.  Then $\varrho_{YB|A}$
  is a causal conditional state if and only if $\{
  \mathcal{E}^{B|A}_y \}$ is a quantum instrument.
\end{theorem}
The proof is similar to those of theorems~\ref{thm:CCS:Jamiol},
\ref{thm:CCS:States} and \ref{thm:CCS:POVM} and is left to the reader.

Theorem~\ref{thm:CCS:Instrument} allows the transformative aspect of a
quantum measurement to be represented by conditional states as
follows.
\begin{proposition}
  \label{prop:CCS:Instrument}
  Let $\varrho_{YB|A}$ be the causal conditional state associated with
  the instrument $\{ \mathcal{E}^{B|A}_y \}$.  Then, when the
  measurement corresponding to this instrument is made with input
  state $\rho_A$, the state-update rule in conventional notation is
  given by
  \begin{equation}
    \label{eq:CCS:QInstrumentConv}
    P(Y=y) \rho^B_y = \mathcal{E}^{B|A}_y \left ( \rho_{A} \right ).
  \end{equation}
  In analogy with the classical expression in
  eq.~\eqref{eq:CCS:CInstrument}, this can be expressed in the
  conditional states framework as
  \begin{equation}
    \label{eq:CCS:QInstrumentCS}
    \rho_{YB} = \Tr[A]{\varrho_{YB|A} \rho_A},
  \end{equation}
  where $\rho_{YB} = \sum_y P(Y=y) \Ket{y}\Bra{y}_Y \otimes \rho^B_y$.
  Furthermore, the conventional expression for the relation between a
  POVM and a quantum instrument,
  \begin{equation}
    E^A_y = (\mathcal{E}^{\dagger}_y)^{A|B} \left ( I_B \right ),
  \end{equation}
  can be expressed simply as
  \begin{equation}
    \label{eq:CCS:QInstrumentPOVM}
    \varrho_{Y|A} =  \Tr[B]{\varrho_{YB|A}},
  \end{equation}
  in analogy with the classical expression in
  eq.~\eqref{eq:CCS:CInstrumentRespfn}.
\end{proposition}
The proof of eq.~\eqref{eq:CCS:QInstrumentCS} consists of applying
proposition~\ref{prop:CCS:Jamiol}, in particular
eq.~\eqref{eq:CCS:Evolve}, to eq.~\eqref{eq:CCS:QInstrumentConv} for
every value of $Y$.  Eq.~(\ref{eq:CCS:QInstrumentPOVM}) follows from
applying theorem~\ref{thm:CCS:Duals} to each element of the
instrument.

Note that, from the perspective of the conditional states framework,
the fact that there are many quantum instruments consistent with a
given POVM is no more surprising than the fact that in classical
probability theory there are many joint distributions consistent with
a given marginal distribution.

Finally, note that for a quantum instrument, the map
$\mathcal{E}_{B|A}(\rho_A) = \sum_y \mathcal{E}^{B|A}_y(\rho_A)$ is
CPT and represents the \emph{non-selective} state-update rule, i.e.\
the one that you should apply if you know that the measurement has
been made but do not know its outcome.  In the conditional states
framework, if you know that a measurement associated with the causal
conditional state $\varrho_{YB|A}$ has been performed, but you do not
know the outcome, then you simply marginalize over $Y$ to obtain the
causal conditional state $\varrho_{B|A}$.  Quantum belief propagation
from $A$ to $B$ using $\varrho_{B|A}$ is the non-selective
state-update rule.

Table~\ref{tbl:CCS:Conventional} provides a summary of how dynamics,
ensemble preparations and measurements are represented in the
conditional states formalism as compared to the conventional
formalism.

\section{Quantum Bayes' Theorem}
\label{Bayes}

This section develops a quantum generalization of Bayes' theorem that
relates the conditional states $\rho_{B|A}$ ($\varrho_{B|A}$) and
$\rho_{A|B}$ ($\varrho_{A|B}$).  Formally, the quantum Bayes' theorem
is the same for acausal and causal conditional states, so this
represents a success in our project to develop a causally neutral
theory of quantum Bayesian inference.  In \S\ref{Bayes:Gen}, the
quantum Bayes' theorem is introduced for two quantum regions.  When
written in terms of conventional notation, it reproduces the
Barnum-Knill approximate error correction map \cite{Barnum2002}.
\S\ref{Bayes:Hybrid} specializes to the hybrid case, which provides a
rule for relating sets of states to POVMs.  In conventional notation,
this reproduces the definition of the ``pretty-good'' measurement
\cite{Hausladen1994, Belavkin1975, Belavkin1975a} and a quantum analogue
of Bayes' theorem previously advocated by Fuchs \cite{Fuchs2003a}.  As
an application of the quantum Bayes' theorem, we develop a
retrodictive formalism for quantum theory in \S\ref{Bayes:Retro} in
which states are evolved backwards in time.  This demonstrates that
the conditional states formalism is causally neutral with respect to
the direction of time.  Finally, in \S\ref{Bayes:Remote}, the acausal
analogue of the symmetry between prediction and retrodiction is
discussed in the context of remote measurement.

\subsection{General Quantum Bayes' Theorem}

\label{Bayes:Gen}

Recall that the classical Bayes' theorem is
\begin{equation}
  P(R|S)=\frac{P(S|R)P(R)}{P(S)},
\end{equation}
which is derived by noting two expressions for the joint probability
in terms of conditionals and marginals
\begin{align}
  P(R,S)  &  =P(R|S)P(S) \label{eq:BR1}\\
  & = P(S|R)P(R). \label{eq:BR2}
\end{align}

Quantum conditional states can be used to derive a quantum analogue of
Bayes' theorem.  For acausal conditional states, the two analogous
expressions to eqs.~\eqref{eq:BR1} and \eqref{eq:BR2} are
\begin{align}
  \rho_{AB}  &  = \rho_{A|B} \Sprod \rho_B \\
  & =\rho_{B|A} \Sprod \rho_A.
\end{align}
Combining these gives
\begin{equation}
  \label{eq:Bayes:Bayes}
  \rho_{A|B}= \rho_{B|A} \Sprod \left ( \rho_A \rho_B^{-1} \right ),
\end{equation}
which is a quantum analogue of Bayes' theorem for acausal conditional
states.

Classically, the distribution $P(S)$ that appears in the denominator
of Bayes' theorem is usually computed via belief propagation as $P(S)
= \sum_R P(S|R)P(R)$.  This gives the alternative form of Bayes'
theorem
\begin{equation}
  \label{eq:LaplacianBayes}
  P(R|S)=\frac{P(S|R)P(R)}{\sum_R P(S|R)P(R)}.
\end{equation}
Similarly, noting that $\rho_B = \Tr[A]{\rho_{B|A} \Sprod \rho_A} =
\Tr[A]{\rho_{B|A} \rho_A}$, the quantum Bayes' theorem for acausal
conditionals can be written as
\begin{equation}
  \label{eq:Bayes:BayesLH}
  \rho_{A|B}= \rho_{B|A} \Sprod \left \{ \rho_A \left [ \Tr[A]{\rho_{B|A}
        \rho_A}\right ]^{-1} \right \}.
\end{equation}

Now consider the case of two causally-related regions.  Suppose that a
region $A$, described by the state $\rho_A$, is mapped to $B$ by a CPT
map $\mathcal{E}_{B|A}$ that is Jamio{\l}kowski-isomorphic to the
causal conditional state $\varrho_{B|A}$.  Recall from \S\ref{CCS:TJS}
that the conditional state and input state can be used to define a
causal joint state $\varrho_{AB} = \varrho_{B|A} \Sprod \rho_A$.  Now,
we can try to define a new causal conditional state $\varrho_{A|B}$
via an analogous decomposition of the causal joint state, namely,
$\varrho_{AB} = \varrho_{A|B} \Sprod \rho_B$, where $\rho_B =
\Tr[A]{\varrho_{B|A} \rho_A}$ is the output state of the channel.
Equating the two expressions for the joint state $\varrho_{AB}$, we
obtain an expression for $\varrho_{A|B}$, which can be regarded as the
causal version of the quantum Bayes' theorem,
\begin{equation}
  \label{eq:Bayes:Temporal1}
  \varrho_{A|B} = \varrho_{B|A} \Sprod \left ( \rho_A \rho_B^{-1} \right ),
\end{equation}
and which can also be written as
\begin{equation}
 \varrho_{A|B} = \varrho_{B|A} \Sprod \left \{ \rho_A \left [ \Tr[A]{\varrho_{B|A}
      \rho_A} \right ]^{-1} \right \}. \label{eq:Bayes:Temporal2}
\end{equation}

In order for this to make sense, it must be checked that
$\varrho_{A|B}$ is indeed a valid causal conditional state.  Taking
the partial transpose over $B$ of eq.~(\ref{eq:Bayes:Temporal1}), we
have $\varrho_{A|B}^{T_B} = \varrho_{B|A}^{T_B} \Sprod \left ( \rho_A
  (\rho_B^{T_B})^{-1} \right )$.  Given that $\rho_B$ is a valid state
(positive and normalized), and given that the set of such states is
mapped to itself by the transpose, $\rho_B^{T_B}$ is also a valid
state.  Furthermore, given that $\varrho_{B|A}$ is a valid causal
conditional state, and the fact that the set of such states are mapped
to the set of valid acausal conditional states by the partial
transpose, it follows that $\varrho_{B|A}^{T_B}$ is a valid acausal
conditional state. But then, by the acausal quantum Bayes' theorem
given in eq.~\eqref{eq:Bayes:Bayes}, $\varrho_{A|B}^{T_B}$ is a valid
acausal conditional state, which implies that $\varrho_{B|A}$ is a
valid causal conditional state.

It is instructive to see how the causal and acausal versions of Bayes'
theorem appear in conventional notation.

For the causal version, suppose that the causal conditional state
$\varrho_{B|A}$ is associated, via the Jamio{\l}kowski-isomorphism,
with a quantum channel $\mathcal{E}_{B|A}$ and its Bayesian inversion
$\varrho_{A|B}$ is associated with a quantum channel
$\mathcal{F}_{A|B}$.  Then, eq.~\eqref{eq:Bayes:Temporal1} is
equivalent to
\begin{equation}
  \label{eq:Bayesianinversionchannel}
  \mathcal{F}_{A|B}(\cdot) = \rho_A^{\frac{1}{2}}\left \{
    \mathcal{E}^{\dag}_{A|B}\left [\rho_B^{-\frac{1}{2}} (\cdot)
      \rho_B^{-\frac{1}{2}}\right ] \right \} \rho_A^{\frac{1}{2}}.
\end{equation}
where
\begin{equation}
  \rho_B = \mathcal{E}_{B|A}(\rho_A).
\end{equation}
The converse of this relation, wherein $\mathcal{E}_{B|A}$ is
expressed in terms of $\mathcal{F}_{A|B}$, is obtained by simply
exchanging $A$ and $B$ as well as $\mathcal{E}$ and $\mathcal{F}$.

The proof that eq.~\eqref{eq:Bayes:Temporal1} translates into
eq.~\eqref{eq:Bayesianinversionchannel} is straightforward.  From the
Jamio{\l}kowski isomorphism (theorem~\ref{thm:ACS:Jamiol}),
$\mathcal{F}_{A|B}(\cdot) = \Tr[B]{\varrho_{A|B} (\cdot)}$, which
implies that
\begin{align}
  \mathcal{F}_{A|B}(\cdot) & = \Tr[B]{\rho_A^{\frac{1}{2}}
    \rho_B^{-\frac{1}{2}} \varrho_{B|A} \rho_B^{-\frac{1}{2}}
    \rho_A^{\frac{1}{2}} (\cdot)} \\
  & = \rho_A^{\frac{1}{2}} \Tr[B]{\rho_B^{-\frac{1}{2}} (\cdot)
    \rho_B^{-\frac{1}{2}} \varrho_{B|A}} \rho_A^{\frac{1}{2}},
\end{align}
where the first step follows from eq.~\eqref{eq:Bayes:Temporal1} and
expanding the $\Sprod$-product, and the second step uses the cyclic
property of the trace. Eq.~\eqref{eq:Bayesianinversionchannel} then
follows from the representation of dual maps in terms of conditional
states given in theorem~\ref{thm:CCS:Duals}.

The map $\mathcal{F}_{A|B}$ is recognizable as the Barnum-Knill
recovery map for the channel $\mathcal{E}_{B|A}$ \cite{Barnum2002}.
This map is known to achieve near-optimal quantum error correction in
situations where the input state and channel are known.  To our
knowledge, its connection with Bayesian inversion has not previously
been noted.  It suggests that the best way of thinking about
$\mathcal{F}_{A|B}$ is not as an error correction map, but rather as a
means for accurately capturing your beliefs about region $A$ given
your beliefs about region $B$.

A similar result holds for the acausal case. Suppose that the linear
map that is Jamio{\l}kowski-isomorphic to the acausal conditional
state $\rho_{B|A}$ is $\mathfrak{E}_{B|A}$ and the one associated to
its Bayesian inversion $\rho_{A|B}$ is $\mathfrak{F}_{A|B}$ (recall
that only $\mathfrak{E}_{B|A} \circ T_A$ and $\mathfrak{F}_{A|B} \circ
T_A$ are CPT maps).  Following the same reasoning used above,
eq.~\eqref{eq:Bayes:Bayes} can be rewritten as
\begin{equation}
  \mathfrak{F}_{A|B}(\cdot) = \rho_A^{\frac{1}{2}}\left \{
    \mathfrak{E}^{\dag}_{A|B}\left [\rho_B^{-\frac{1}{2}} (\cdot)
      \rho_B^{-\frac{1}{2}}\right ] \right \} \rho_A^{\frac{1}{2}},
\end{equation}
where
\begin{equation}
  \rho_B = \mathfrak{E}_{B|A}(\rho_A).
\end{equation}

\subsection{Bayes' Theorem for Quantum-Classical Hybrids}

\label{Bayes:Hybrid}

For quantum-classical hybrids, there are two versions of Bayes'
theorem, depending on whether it is the conditioned region or the
conditioning region that is classical. Recall that the mathematical
form of hybrid conditional states does not depend on whether they are
causal or acausal, and $\sigma$ is the notation used for results that
do not depend on the causal interpretation.  The two versions of
Bayes' theorem are then:
\begin{equation}
  \label{eq:Bayes:HybridLH}
  \sigma_{X|A} = \sigma_{A|X} \Sprod \left ( \rho_X
    \rho_A^{-1} \right ),
\end{equation}
and
\begin{equation}
  \label{eq:Bayes:HybridRH}
  \sigma_{A|X} = \sigma_{X|A} \Sprod \left ( \rho_A \rho_X^{-1} \right ),
\end{equation}
where $\rho_A = \Tr[X]{\sigma_{A|X} \rho_X}$ and $\rho_X =
\Tr[A]{\sigma_{X|A} \rho_A}$.

A hybrid joint state $\sigma_{XA}$ may be decomposed into a hybrid
conditional state and a reduced state via eq.~\eqref{eq:ACS:JointStar}
in two distinct ways: either in terms of a classical reduced state and
a conditional state that is conditioned on the classical part,
\begin{equation}
  \label{eq:Bayes:HybridStates}
  \sigma_{XA} = \sigma_{A|X} \Sprod \rho_X,
\end{equation}
or in terms of a quantum reduced state and a conditional state that is
conditioned on the quantum part,
\begin{equation}
  \label{eq:Bayes:HybridPOVM}
  \sigma_{XA} = \sigma_{X|A} \Sprod \rho_A.
\end{equation}
The hybrid Bayes' theorems of eqs.~\eqref{eq:Bayes:HybridLH} and
\eqref{eq:Bayes:HybridRH} give the rules for converting between these
two decompositions.

To see how these two decompositions appear in conventional notation,
recall from theorem~\ref{thm:CCS:States} that a general hybrid
conditional state $\sigma_{A|X}$ is of the form
\begin{equation}
  \sigma_{A|X} = \sum_x \Ket{x}\Bra{x}_X \otimes \rho_x^A,
\end{equation}
where each $\rho_x^A$ is a normalized density operator on $\Hilb[A]$,
and from theorem~\ref{thm:CCS:POVM} that $\sigma_{X|A}$ is of the form
\begin{equation}
  \sigma_{X|A} = \sum_x \Ket{x}\Bra{x}_X \otimes E_x^A,
\end{equation}
where $\{E_x^A\}$ is a POVM on $A$.  Finally, recall that the
classical state $\rho_X$ is of the form $\rho_X = \sum_x P(X=x)
\Ket{x}\Bra{x}_X$ where $P(X)$ is a classical probability
distribution.

Eq.~\eqref{eq:Bayes:HybridStates} therefore gives a decomposition of a
joint state in terms of a set of states and a probability distribution
via
\begin{equation}
  \sigma_{XA} = \sum_x P(X=x) \Ket{x}\Bra{x}_X \otimes \rho^A_x,
\end{equation}
and eq.~\eqref{eq:Bayes:HybridPOVM} gives a decomposition in terms of
a POVM and a state for $A$ via
\begin{equation}
  \sigma_{XA} = \sum_x \Ket{x}\Bra{x}_X \otimes \rho_A^{\frac{1}{2}} E^A_x
  \rho_A^{\frac{1}{2}}
\end{equation}

In terms of components, the Bayes' theorem of
eq.~\eqref{eq:Bayes:HybridLH} is a rule for determining a POVM from a
probability distribution and a set of states, whilst
eq.~\eqref{eq:Bayes:HybridRH} is a rule for determining a set of
states from a POVM and a state on $A$.  These rules are:
\begin{equation}
  \label{eq:Bayes:HybridLHCon}
  \rho_x^A =  \frac{\rho_A^{\frac{1}{2}} E_x^A
    \rho_A^{\frac{1}{2}}}{\Tr[A]{E_x^A \rho_A}}.
\end{equation}
and
\begin{equation}
  \label{eq:Bayes:HybridRHCon}
  E_x^A =  P(X=x) \rho_A^{-\frac{1}{2}} \rho_x^A
  \rho_A^{-\frac{1}{2}},
\end{equation}
where $\rho_A = \sum_x P(X=x) \rho_x^A$.

These rules have appeared numerous times in the literature,
e.g. \cite{Hughston1993, Leifer2006, Leifer2007}.  In the context of
distinguishing the states in an ensemble, the POVM defined by
eq.~\eqref{eq:Bayes:HybridRHCon} is known as the ``pretty good''
measurement \cite{Hausladen1994, Belavkin1975, Belavkin1975a}.
Eq.~\eqref{eq:Bayes:HybridLHCon} is a rule previously advocated by
Fuchs as a quantum analogue of Bayes' theorem \cite{Fuchs2003a}.  The
fact that these rules are all special cases of a more general quantum
Bayes' theorem goes some way to explaining their utility.

\subsection{Retrodiction and Time Symmetry}

\label{Bayes:Retro}

As an application of the quantum Bayes' theorem, we use it to develop
a retrodictive formalism for quantum theory, in which states are
propagated backwards in time from late regions to early regions.  This
is operationally equivalent to the usual predictive formalism, in
which states are propagated forward in time from early regions to late
regions. The retrodictive description is particularly useful if you
acquire new information about the late region and wish to update your
beliefs about the early region, for instance, when you learn about the
output of a noisy channel and wish to make inferences about its input.
This situation will be considered in \S\ref{Cond:Exa}.

Barnett \textit{et al.} have previously proposed a formalism for
retrodictive inference in quantum theory\cite{Barnett2000, Pegg2002,
  Pegg2002a}.  For unbiased sources --- for which the ensemble average
of the prepared states is the maximally mixed state --- their
formalism is identical to the one presented here, and the quantum
Bayes' theorem provides it with an intuitive derivation.  For biased
sources, their formalism differs from ours. The one we propose has the
advantage that it can be derived as a special case of our general
formalism for quantum Bayesian inference and thereby retains a closer
analogy with classical Bayesian inference.

As emphasized by the de Finetti quote in the introduction, the rules
for making classical probabilistic inferences about the past are the
same as those for making inferences about the future.  By analogy, in
the conditional states formalism, we would expect to be able to
propagate quantum states from future regions to past regions via the
same rules used to propagate them from past regions to future regions.

If the state $\rho_A$ of an early region is mapped to the state
$\rho_B$ of a later region by a CPT map $\mathcal{E}_{B|A}$ then
\begin{equation}
  \rho_{B} = \Tr[A]{\varrho_{B|A} \rho_A},
\end{equation}
where $\varrho_{B|A}$ is the Jamio{\l}kowski isomorphic causal
conditional state to $\mathcal{E}_{B|A}$.  By construction, the causal
conditional state $\varrho_{A|B}$ defined by Bayes' theorem in
eq.~\eqref{eq:Bayes:Temporal1} satisfies $\varrho_{AB} = \varrho_{B|A}
\Sprod \rho_A = \varrho_{A|B} \Sprod \rho_B$, and this causal joint
state has $\rho_A$ and $\rho_B$ as its marginals, so we have
\begin{equation}
  \rho_A = \Tr[B]{\varrho_{A|B} \rho_B}
\end{equation}

In conventional notation, this is equivalent to
\begin{equation}
  \rho_A = \mathcal{E}_{A|B}^{\text{retr}} \left ( \rho_B \right ),
\end{equation}
where
\begin{equation}
  \mathcal{E}_{A|B}^{\text{retr}} \left ( \cdot \right ) \equiv
  \rho_A^{\frac{1}{2}}\mathcal{E}^{\dagger}_{A|B} \left ( \rho_B^{-\frac{1}{2}} \left (
      \cdot \right ) \rho_B^{-\frac{1}{2}}\right ) \rho_A^{\frac{1}{2}}.
\end{equation}
is the map that is Jamio{\l}kowski isomorphic to $\varrho_{A|B}$ and
the superscript $^{\text{retr}}$ is used to indicate that this is a
retrodictive map that propagates states from future to past regions.

For comparison with the retrodictive formalism of \cite{Barnett2000},
consider a simple prepare-and-measure experiment, as depicted in
fig.~\ref{fig:Bayes:PrepMeas}.

\begin{figure}[htb]
  \includegraphics[scale=0.4]{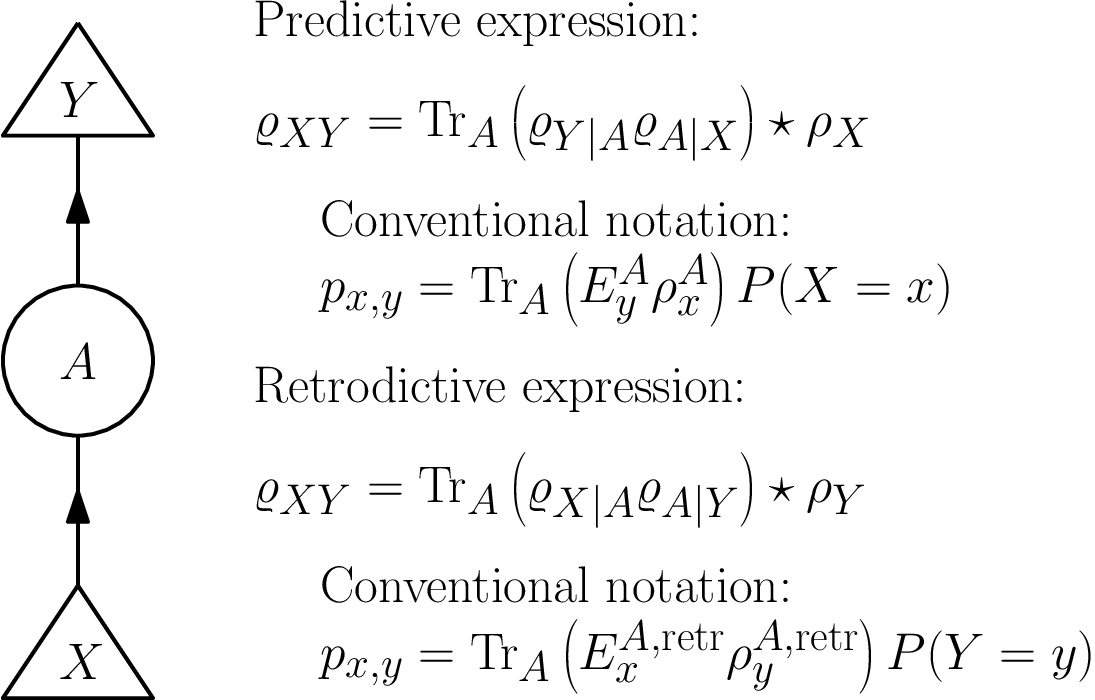}
  \caption{\label{fig:Bayes:PrepMeas}A prepare-and-measure
    experiment in which a preparation procedure is followed by a
    measurement.  We are interested in computing the joint
    probability distribution $P(X,Y)$ of the preparation variable
    and the measurement outcome.  For compactness, $p_{x,y} =
    P(X=x,Y=y)$ is used for the expressions in conventional
    notation.  }
\end{figure}

In the predictive description, the preparation procedure is
characterized by a probability distribution $P(X)$, which can
alternatively be represented by a diagonal state $\rho_X$, and by a
set of states $\rho^A_x$, which can alternatively be represented by a
causal conditional state $\varrho_{A|X}$.  The measurement is
characterized by a POVM $\left \{ E^A_y \right \}$, or alternatively a
causal conditional state $\varrho_{Y|A}$.  This is predictive because
the conditional states, $\varrho_{A|X}$ and $\varrho_{Y|A}$, are
conditioned on regions in their immediate past.

In order to calculate the joint distribution $P(X,Y)$, or equivalently
the causal joint state $\varrho_{XY}$, we need to use the belief
propagation formula for causal conditional states as given in
eq.~\eqref{eq:CCS:CondBP}.  This proceeds as follows:
\begin{enumerate}
\item Propagate the causal conditional state of $A$ given $X$ into the
  future to obtain $\varrho_{Y|X} = \Tr[A]{\varrho_{Y|A}
    \varrho_{A|X}}$.
\item Combine the causal conditional state for $Y$ given $X$ with the
  state for $X$ to obtain $\varrho_{XY} = \rho_{Y|X} \Sprod \rho_X$.
\end{enumerate}

Combining these steps, the predictive expression for $\varrho_{XY}$ is
\begin{equation}
  \label{eq:Bayes:PredictiveTemp}
  \varrho_{XY} = \Tr[A]{\varrho_{Y|A} \varrho_{A|X}} \Sprod \rho_X.
\end{equation}

In conventional notation, if the states $\left \{ \rho^A_x \right \}$
are the components of $\varrho_{A|X}$ and the elements of the POVM
$\{E^A_y\}$ are the components of $\varrho_{Y|A}$, then
eq.~\eqref{eq:Bayes:PredictiveTemp} is equivalent to
\begin{equation}
  P(X=x,Y=y) = \Tr[A]{E^A_y \rho^A_x} P(X=x).
\end{equation}

A retrodictive description of the same experiment can be given,
involving states of regions conditioned on regions in their immediate
future, i.e. $\rho_Y$, $\varrho_{A|Y}$ and $\varrho_{X|A}$.  These
correspond, respectively, to a probability distribution $P(Y)$, a set
of states $\left \{ \rho^{A,\text{retr}}_y \right \}$, and a POVM
$\left \{ E^{A,\text{retr}}_x \right \}$. Note that, in contrast to
the predictive description, the measurement is now being described by
a classical probability distribution and a set of states, which we
call retrodictive states, whilst the preparation is being described by
a POVM, which we call the retrodictive POVM.  The retrodictive
calculation of $\varrho_{XY}$ proceeds as follows:
\begin{enumerate}
\item Propagate the causal conditional state of $A$ given $Y$ into the
  past to obtain $\varrho_{X|Y} = \Tr[A]{\varrho_{X|A}
    \varrho_{A|Y}}$.
\item Combine the causal conditional state for $X$ given $Y$ with the
  state for $Y$ to obtain $\varrho_{XY} = \rho_{X|Y} \Sprod \rho_Y$.
\end{enumerate}

Combining these steps, the retrodictive expression for $\varrho_{XY}$
is
\begin{equation}
  \label{eq:Bayes:RetrodictiveTemp}
  \varrho_{XY} = \Tr[A]{\varrho_{X|A} \varrho_{A|Y}} \Sprod \rho_Y.
\end{equation}

In conventional notation, this is equivalent to
\begin{multline}
  P(X=x,Y=y) = \\ \Tr[A]{E^{A,\text{retr}}_x \rho^{A,\text{retr}}_y}
  P(Y=y).
\end{multline}

The retrodictive and predictive descriptions of the experiment are
related by the quantum Bayes' and quantum belief propagation via
\begin{align}
  \varrho_{X|A} & = \varrho_{A|X} \Sprod \left ( \rho_X \rho_A^{-1}
  \right ), \label{eq:Bayes:RetroPOVMCond} \\
  \text{where} \qquad \rho_A & = \Tr[X]{\rho_{A|X} \rho_X},
\end{align}
and
\begin{align}
  \varrho_{A|Y} & = \varrho_{Y|A} \Sprod \left ( \rho_A \rho_Y^{-1}
  \right ), \label{eq:Bayes:RetroStatesCond}\\
  \text{where} \qquad \rho_Y & = \Tr[A]{\rho_{Y|A} \rho_A}.
\end{align}

In conventional notation, these equations are equivalent to
\begin{align}
  & E^{A,\text{retr}}_x = P(X=x)
  \rho_A^{-\frac{1}{2}}\rho^A_x\rho_A^{-\frac{1}{2}}, \label{eq:Bayes:RetroPOVM}
  \\
  \text{where} \qquad & \rho_A  = \sum_x P(X=x)\rho^A_x,
\end{align}
and
\begin{align}
  & \rho^{A,\text{retr}}_y = \frac{\rho_A^{\frac{1}{2}}E^A_y
    \rho_A^{\frac{1}{2}}}{P(Y=y)}, \label{eq:Bayes:RetroStates} \\
  \text{where} \qquad & P(Y=y) = \Tr[B]{E^A_y \rho_A}.
\end{align}

Eqs.~\eqref{eq:Bayes:RetroPOVMCond} and
\eqref{eq:Bayes:RetroStatesCond} can be used to prove that the
predictive and retrodictive expressions for $\varrho_{XY}$ do indeed
give the same causal joint state.  Starting from
eq.~\eqref{eq:Bayes:RetrodictiveTemp}, we have
\begin{align}
  \varrho_{XY} = & \Tr[A]{\varrho_{X|A} \varrho_{A|Y}} \Sprod \rho_Y
  \\
  = & \text{Tr}_A \left \{ \left ( \varrho_{A|X} \Sprod \left [ \rho_X
        \rho_A^{-1}
      \right ] \right ) \right . \\
  & \left . \left ( \varrho_{Y|A} \Sprod \left [ \rho_A \rho_Y^{-1}
      \right ]
    \right ) \right \} \Sprod \rho_Y \\
  = & \text{Tr}_A \left ( \rho_Y^{\frac{1}{2}}
    \rho_X^{\frac{1}{2}}\rho_A^{-\frac{1}{2}}
    \varrho_{A|X}\rho_X^{\frac{1}{2}} \rho_A^{-\frac{1}{2}} \right ) \\
  & \left . \rho_A^{\frac{1}{2}} \rho_Y^{-\frac{1}{2}} \varrho_{Y|A}
    \rho_A^{\frac{1}{2}} \rho_Y^{-\frac{1}{2}} \rho_Y^{\frac{1}{2}}
  \right )
  \\
  = & \Tr[A]{\rho_Y^{\frac{1}{2}}
    \rho_X^{\frac{1}{2}}\rho_A^{-\frac{1}{2}}
    \varrho_{A|X}\rho_X^{\frac{1}{2}} \rho_Y^{-\frac{1}{2}}
    \varrho_{Y|A} \rho_A^{\frac{1}{2}}}
\end{align}
Since $\rho_Y$ commutes with $\rho_X$, $\rho_A$ and $\varrho_{A|X}$,
the $\rho_Y^{\frac{1}{2}}$ term can be moved forward to cancel with
$\rho_Y^{-\frac{1}{2}}$ term.  The $\rho_A^{-\frac{1}{2}}$ term can be
made to cancel with the last $\rho_A^{\frac{1}{2}}$ term via the
cyclic property of the trace.  This yields
\begin{align}
  \varrho_{XY} & = \Tr[A]{\rho_X^{\frac{1}{2}}
    \varrho_{A|X}\rho_X^{\frac{1}{2}} \varrho_{Y|A}} \\
  & = \rho_X^{\frac{1}{2}} \Tr[A]{\varrho_{A|X} \varrho_{Y|A}}
  \rho_X^{\frac{1}{2}} \\
  & = \Tr[A]{\varrho_{A|X} \varrho_{Y|A}} \Sprod \rho_X,
\end{align}
where, in the second line, we have used the fact that $\rho_X$
commutes with $\rho_{Y|A}$.  Finally, applying the cyclic property of
the trace gives
\begin{equation}
  \varrho_{XY} = \Tr[A]{\varrho_{Y|A} \varrho_{A|X}} \Sprod \rho_X,
\end{equation}
which is the predictive expression for $\varrho_{XY}$.

We are now in a position to show that our formalism coincides with
that of Ref.~\cite{Barnett2000} for the case of unbiased sources,
where $\rho_A = I_A/d$ with $d$ the dimension of $\Hilb[A]$. In this
case, the definitions of retrodictive states and retrodictive POVMs in
\cite{Barnett2000} were
\begin{align}
  E^{A,\text{retr}}_x & = d P(X=x) \rho^A_x,
\end{align}
and
\begin{align}
  \rho^{A,\text{retr}}_y & = \frac{E^A_y}{d P(Y=y)},
\end{align}
which are easily seen to be special cases of
eqs.~\eqref{eq:Bayes:RetroPOVM} and \eqref{eq:Bayes:RetroStates}.

Finally, note that the analysis above can be extended to deal with the
scenario depicted in fig.~\ref{fig:Bayes:Intervene}, in which there is
an intervening channel $\mathcal{E}_{B|A}$ between the preparation and
measurement so that the measurement is now made on region $B$.  In
fact, by making use of the rule for propagating conditional states
given in eq.~\eqref{eq:CCS:CondBP} (or equivalently the Heisenberg
dynamics given in eq.~\eqref{eq:CCS:CondHeis}), the extra region $B$
can be eliminated from the description by defining
\begin{equation}
  \label{eq:Bayes:CondHeis}
  \varrho_{Y|A} = \Tr[A]{\varrho_{Y|B}\varrho_{B|A}},
\end{equation}
which determines the effective measurement on $A$ that is performed by
actually measuring region $B$.  The three operators $\rho_X$,
$\varrho_{A|X}$ and $\varrho_{Y|A}$ then provide a predictive
description of a simple prepare-and-measure experiment, and so the
previous analysis applies.

Specifically, substituting eq.~\eqref{eq:Bayes:CondHeis} into
eq.~\eqref{eq:Bayes:PredictiveTemp} gives the predictive expression
for $\varrho_{XY}$ as
\begin{equation}
  \label{eq:Bayes:PredictiveInter}
  \varrho_{XY} = \Tr[AB]{\varrho_{Y|B}\varrho_{B|A}
    \varrho_{A|X}} \Sprod \rho_X,
\end{equation}
or in conventional notation
\begin{multline}
  P(X=x, Y=y) = \\ \Tr[B]{E^B_y \mathcal{E}_{B|A} \left ( \rho^A_x \right
    )} P(X=x).
\end{multline}

\begin{figure}[htb]
  \includegraphics[scale=0.4]{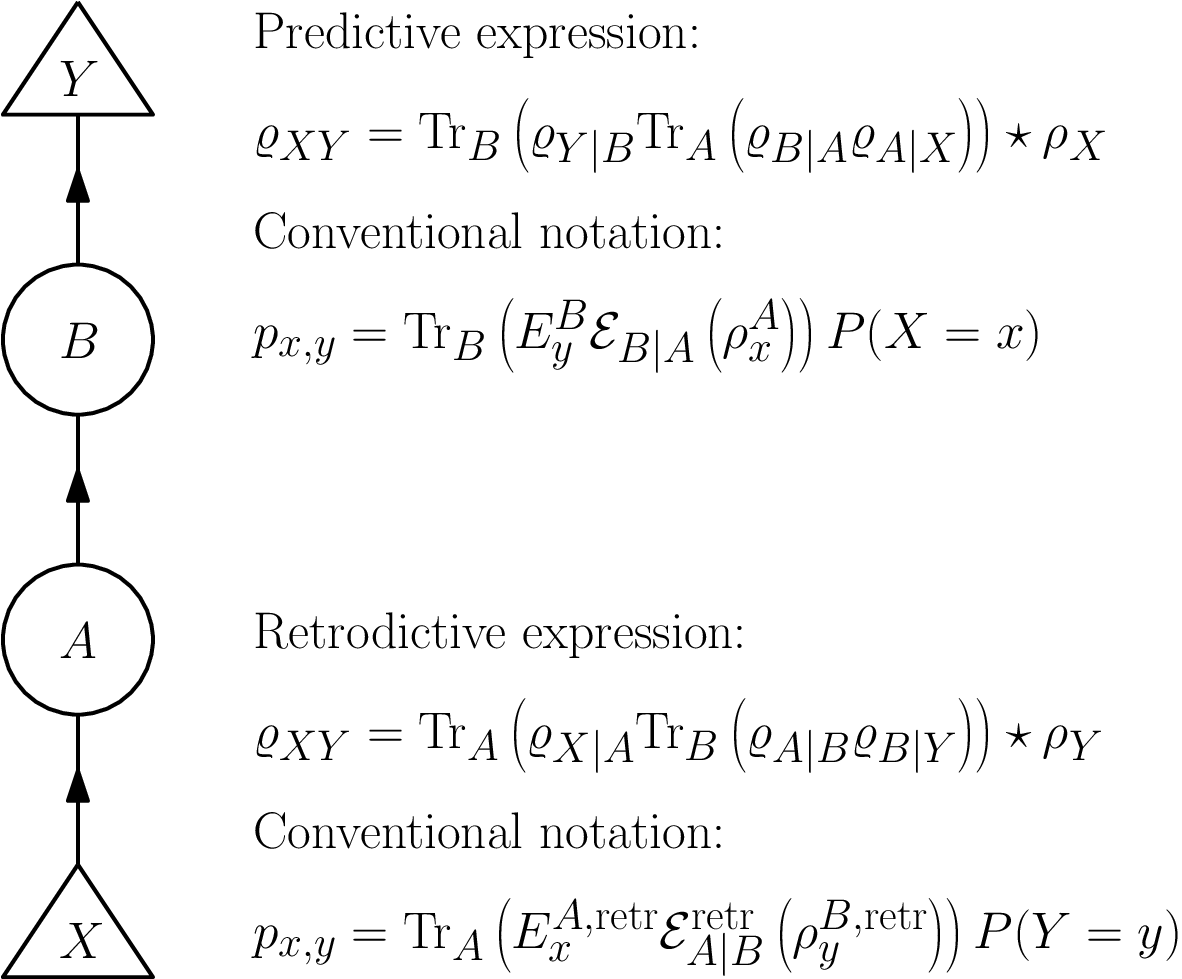}
  \caption{\label{fig:Bayes:Intervene}A prepare-and-measure experiment
    with an intervening channel.  We are interested in computing the
    joint probability distribution $P(X,Y)$ of the preparation
    variable and the measurement outcome. For compactness, $p_{x,y} =
    P(X=x,Y=y)$ is used for the expressions in conventional notation.}
\end{figure}

Similarly, the retrodictive expression is obtained from
eq~\eqref{eq:Bayes:RetrodictiveTemp} by substituting $\varrho_{A|Y} =
\Tr[B]{\varrho_{A|B}\varrho_{B|Y}}$, where $\varrho_{A|B}$ and
$\varrho_{B|Y}$ are obtained from $\varrho_{B|A}$ and $\varrho_{Y|B}$
by Bayes' theorem.  This gives
\begin{equation}
  \label{eq:Bayes:RetrodictiveInter}
  \varrho_{XY} = \Tr[AB]{\varrho_{X|A}\varrho_{A|B}
      \varrho_{B|Y}} \Sprod \rho_Y,
\end{equation}
or in conventional notation
\begin{multline}
  P(X=x,Y=y) = \\ \Tr[A]{E^{A,\text{retr}}_x
    \mathcal{E}^{\text{retr}}_{A|B} \left ( \rho^{B,\text{retr}}_y
    \right )} P(Y=y).
\end{multline}

To sum up, the formalism of quantum conditional states shows that,
just as in the classical case, the rules of quantum Bayesian inference
do not discriminate between prediction and retrodiction.
Specifically, the rules are the same regardless of whether the
propagation is in the same or the opposite direction to the causal
arrows. This reveals an important kind of time-symmetry that is not
apparent in the normal quantum formalism.  More importantly, it shows
that a formalism for quantum Bayesian inference can be found that is
blind to at least this aspect of the causal structure.

\subsection{Remote Measurements and Spatial Symmetry}

\label{Bayes:Remote}

Many of the novel features of quantum theory exhibit themselves in the
correlations that can be obtained between local measurements on a pair
of acausally-related regions.  These include Einstein-Podolsky-Rosen
correlations and Bell correlations.  Inferences about such
measurements can be treated using a formalism that is almost identical
to the predictive and retrodictive expressions for prepare-and-measure
experiments given above.  One simply has to substitute the formula for
propagating a causal conditional state across acausally-related
regions, eq.~\eqref{eq:CCS:AcausalCausal}, for the formula for
propagating them across causally-related regions,
eq.~\eqref{eq:CCS:CondBP}, used above.

For two acausally-related regions, $A$ and $B$, it is self evident
that there is complete symmetry between propagation from one region to
another and back again, i.e.\ if
\begin{equation}
  \rho_B = \Tr[A]{\rho_{B|A} \rho_A}
\end{equation}
then
\begin{equation}
  \rho_A = \Tr[B]{\rho_{A|B} \rho_B}.
\end{equation}
The two conditional states, $\rho_{B|A}$ and $\rho_{A|B}$ are related
by the quantum Bayes' theorem, so Bayes' theorem allows the direction
of belief propagation to be reversed.

To see this spatial symmetry at work, consider the scenario of
measurements being implemented on a pair of acausally-related regions,
$A$ and $B$, as depicted in fig.~\ref{fig:Bayes:EPR}.
\begin{figure}[htb]
  \includegraphics[scale=0.4]{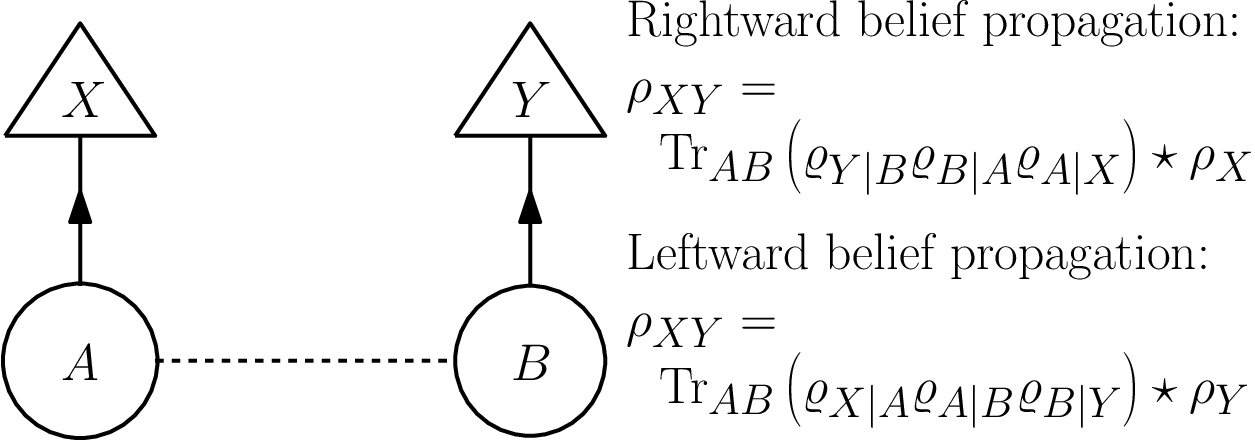}
  \caption{\label{fig:Bayes:EPR} Measurements on a pair of acausally
    related regions, rightward and leftward belief propagation.}
\end{figure}

In the conditional states formalism, the region $AB$ is assigned a
joint state $\rho_{AB}$, and the measurements on $A$ and $B$ are
represented by causal conditional states $\varrho_{X|A}$ and
$\varrho_{Y|B}$.  The joint distribution $P(X,Y)$ over outcomes is
given by the components of the acausal joint state $\rho_{XY}$ via
\begin{equation}
  \label{eq:Bayes:XYcorrelation}
  \rho_{XY} = \Tr[AB]{\left ( \varrho_{X|A}\varrho_{Y|B} \right )
     \rho_{AB}}.
\end{equation}

These correlations can alternatively be calculated by propagating
beliefs about $A$ conditioned on $X$ to beliefs about $Y$ conditioned
on $X$, and also by propagating beliefs about $B$ conditioned on $Y$
to beliefs about $X$ conditioned on $Y$.  These representations are
analogous to the predictive and retrodictive expressions for
prepare-and-measure experiments discussed in the previous section.  We
shall refer to these as \emph{rightward} and \emph{leftward} belief
propagation respectively.

It may seem convoluted to calculate $\rho_{XY}$ via rightward or
leftward belief propagation, when eq.~\eqref{eq:Bayes:XYcorrelation}
already gives a simple expression for it.  However, it is important to
understand how to propagate beliefs across acausally-related regions
in order to deal with the situation in which you obtain new
information about one region and wish to make inferences about the
other.  This is exactly what happens in the analysis of an EPR
experiment.  This problem is known as quantum steering and will be
considered in \S\ref{Cond:Remote}.

First, consider rightward belief propagation.  The aim is to rewrite
eq.~\eqref{eq:Bayes:XYcorrelation} in terms of $\rho_{B|A}$, the state
of $B$ conditioned on the region to its left.  For greater symmetry
with the prepare-and-measure case, we also use $\rho_X$ and
$\varrho_{A|X}$ to describe the left hand wing of the experiment,
whilst retaining $\rho_{Y|B}$ for the right.  Then, we write
$\rho_{AB} = \rho_{B|A} \Sprod \rho_A$ and note that, by Bayes'
theorem,
\begin{equation}
  \varrho_{X|A} = \varrho_{A|X} \Sprod \left ( \rho_X \rho_A^{-1} \right ).
\end{equation}
Substituting these into eq.~\eqref{eq:Bayes:XYcorrelation} gives
\begin{multline}
  \rho_{XY} \\ = \Tr[AB]{\left [\varrho_{A|X} \Sprod \left (
          \rho_X \rho_A^{-1} \right )\right ] \varrho_{Y|B}
    \left [ \rho_{B|A} \Sprod \rho_A \right ]}.
\end{multline}
Expanding the $\Sprod$-products gives
\begin{multline}
  \rho_{XY} = \\
  \Tr[AB]{\rho_X^{\frac{1}{2}}\rho_A^{-\frac{1}{2}}\varrho_{A|X}
    \rho_X^{\frac{1}{2}} \rho_A^{-\frac{1}{2}}
    \varrho_{Y|B} \rho_A^{\frac{1}{2}} \rho_{B|A} \rho_A^{\frac{1}{2}}}.
\end{multline}
All of the $\rho_A$ terms can be cancelled by using the fact that they
commute with $\varrho_{Y|B}$ and $\rho_X$, and by using the cyclic
property of the trace.  Then, we have
\begin{align}
  \rho_{XY} & = \Tr[AB]{\rho_X^{\frac{1}{2}} \varrho_{A|X}
    \rho_X^{\frac{1}{2}}
    \varrho_{Y|B} \rho_{B|A}} \\
  & = \Tr[AB]{\varrho_{Y|B} \rho_{B|A} \rho_X^{\frac{1}{2}}
    \varrho_{A|X} \rho_X^{\frac{1}{2}}} \\
  & =  \Tr[AB]{\varrho_{Y|B} \rho_{B|A}
    \varrho_{A|X}} \Sprod \rho_X,
  \label{eq:Bayes:JointXY}
\end{align}
where we have used the cyclic property of the trace and the fact that
$\rho_X$ commutes with $\rho_{B|A}$ and $\varrho_{Y|B}$.
Eq.~\eqref{eq:Bayes:JointXY} has the same form as the predictive
expression for a prepare-and-measure experiment with an intervening
channel given in eq.~\eqref{eq:Bayes:PredictiveInter}, except that, in
this case, $\rho_{B|A}$ is acausal.  We can also use
eq.~\eqref{eq:CCS:AcausalCausal} to define $\rho_{Y|A} =
\Tr[B]{\varrho_{Y|B}\rho_{B|A}}$, which represents the effective
measurement on region $A$ that is made by measuring region $B$.  Using
this, region $B$ can be eliminated from eq.~\eqref{eq:Bayes:JointXY}
to obtain
\begin{equation}
  \label{eq:Bayes:JointXYNoB}
  \rho_{XY} = \Tr[AB]{\rho_{Y|A} \varrho_{A|X}} \Sprod
  \rho_X.
\end{equation}
This is similar to eq.~\eqref{eq:Bayes:PredictiveTemp}, and, in fact,
is mathematically identical to it because $\rho_{Y|A}$ is a hybrid
conditional state, so its mathematical form does not depend on whether
it is acausal or causal.  Eq.~\eqref{eq:Bayes:JointXYNoB} is a useful
form to use when we want to consider the effect of measuring $B$ on
the remote region $A$, as in quantum steering.

For leftward belief propagation, a similar analysis gives the
expressions
\begin{equation}
  \rho_{XY} = \Tr[AB]{\varrho_{X|A} \rho_{A|B} \varrho_{B|Y}} \Sprod
  \rho_Y,
\end{equation}
which is analogous to eq.~\eqref{eq:Bayes:RetrodictiveInter} and
\begin{equation}
  \label{eq:Bayes:JointXYNoBLeft}
  \rho_{XY} = \Tr[A]{\rho_{X|A} \varrho_{B|Y}} \Sprod \rho_Y,
\end{equation}
which is analogous to eq.~\eqref{eq:Bayes:RetrodictiveTemp}.

As with prediction and retrodiction, there is complete symmetry
between leftward and rightward belief propagation, and there is a
strong symmetry between causal and acausal belief propagation in
general.  This represents progress towards a theory of quantum
Bayesian inference that is completely independent of causal structure.

\section{Quantum Bayesian Conditioning}
\label{Cond}

Classically, Bayesian conditioning is used to update probabilities
when new data are acquired.  Specifically, if you are interested in a
variable $R$, and you learn that some correlated variable $X$ takes
the value $x$, then the theory of Bayesian inference recommends that
you should update your probability distribution for $R$ from the prior
$P(R)$ to the posterior $P_x(R) = P(R|X=x)$\footnote{Strictly
  speaking, this is only a special case of Bayesian conditioning,
  which can be formulated more generally for arbitrary events on a
  sample space rather than just for random variables.  We restrict
  attention to the special case of conditioning one random variable
  upon another for ease of comparison with quantum theory.}.

Bayesian conditioning can be viewed as a two-step process.  First, the
observation that $X=x$ causes you to update your probability
distribution for $X$ from $P(X)$ to $P_x(X)$, where
\begin{equation}
  P_x(X=x') = \delta_{x,x'}.
\end{equation}
Secondly, assuming that the observation of $X$ does not cause you to
change your conditional probabilities $P(R|X)$, the new probability
distribution for $R$ is obtained via belief propagation as
\begin{align}
  P_x(R) & = \sum_X P(R|X)P_x(X) \\
  & = P(R|X=x).
\end{align}

This two step decomposition of conditioning has been emphasized by
Richard Jeffrey \cite{Jeffrey1983, Jeffrey2004} who showed that
whenever an observation causes the value of a random variable to
become certain, and that all probabilities conditioned on that random
variable are unchanged by the observation, then the change in the
probability distribution can be represented as Bayesian conditioning.

The reason for emphasizing this decomposition is that Jeffrey was
interested in situations in which an observation does not cause you to
believe that some variable takes a precise value.  As an example of
this, adapted from \cite{Earman1992}, suppose that $X$ is the color of
a jellybean, which has possible values `red', `green' and `yellow',
and that $R$ is the flavor, which has possible values `cherry',
`strawberry', `lime' and `lemon'.  Suppose that initially, your
probability distribution for $X$, $P(X)$, assigns a probability $1/3$
to each color, and that your observation consists of viewing the
jellybean in the light of a dim candle.  This might not be enough for
you to become certain about the color of the jellybean, but it may
reduce your uncertainty somewhat.  For example, you may now think it
reasonable to assign a probability distribution $P_{\text{post}}(X)$
that gives probability $2/3$ to $X = \text{'red'}$ and $1/6$ each to
$X = \text{'green'}$ and $X = \text{'yellow'}$.  Assuming that the
observation does not cause your conditional probabilities $P(R|X)$ to
change, Jeffrey shows that your posterior probability distribution for
$R$ is obtained by belief propagation via
\begin{equation}
  P_{\text{post}}(R) = \sum_X P(R|X) P_{\text{post}}(X),
\end{equation}
which is known as Jeffrey conditioning.

Many orthodox Bayesians reject the generalization to Jeffrey
conditioning and maintain that the rational way to update
probabilities in the light of data is always via Bayesian
conditioning, not least because most of the apparatus of Bayesian
statistics depends on this.  This position can be defended by
insisting that the sort of situations described above should really be
handled by expanding the sample space to include statements about your
perceptions.  One can show that Jeffrey conditioning can always be
represented as Bayesian conditioning on a larger space in this way.
Counter to this, Jeffrey argues that it is not realistic to construct
such a space, since you do not actually have precise descriptions of
your perceptions.  This argument has been eloquently put by Diaconis
and Zabell \cite{Diaconis1982}:
\begin{quotation}
  For example, suppose we are about to hear one of two recordings of
  Shakespeare on the radio, to be read by either Olivier or Gielgud,
  but are unsure of which, and have a prior with mass $1/2$ on
  Olivier, $1/2$ on Gielgud. After hearing the recording, one might
  judge it fairly likely, but by no means certain, to be by Olivier.
  The change in belief takes place by direct recognition of the voice;
  all the integration of sensory stimuli has already taken place at a
  subconscious level.  To demand a list of objective vocal features
  that we condition on in order to affect the change would be a
  logician's parody of a complex psychological process.
\end{quotation}

The debate over whether Jeffrey conditioning should be subsumed into
Bayesian conditioning is somewhat analogous to a similar argument in
quantum theory about whether POVMs should be regarded as fundamental,
since they can always be represented by projective measurements on a
larger Hilbert space via Naimark extension \cite{Naimark1940}.

In the conditional states formalism, the quantum analogue of Jeffrey
conditioning is straightforward.  If an observation causes you to
change the state you assign to a region $A$ from $\rho_A$ to
$\rho^{\text{post}}_A$, and if your conditional state for another
region $B$, given $A$ (which will either be an acausal state
$\rho_{B|A}$ or a causal state $\varrho_{B|A}$) is unchanged, then
your posterior state for $B$ is determined by belief propagation via
either
\begin{equation}
  \rho^{\text{post}}_B = \Tr[A]{\rho_{B|A} \rho^{\text{post}}_A}
\end{equation}
or
\begin{equation}
  \rho^{\text{post}}_B = \Tr[A]{\varrho_{B|A} \rho^{\text{post}}_A},
\end{equation}
depending on whether $A$ and $B$ are acausally or causally-related.

The question of whether there is a quantum analogue of Bayesian
conditioning is more subtle, as it depends on whether there is a
posterior quantum state for $A$ that is analogous to having certainty
about the value of a classical variable, i.e.\ the point measure
$P_x(X=x') = \delta_{x,x'}$.  Arguably, a pure state could play this
role, since it represents the smallest amount of uncertainty that one
can have about a quantum region.  However, unlike classical point
measures, pure states still assign probabilities other than $0$ and
$1$ to fine-grained measurements, e.g.\ measurements in a
complimentary basis, and there are good reasons to believe that, even
if they represent maximal knowledge, that knowledge is still
incomplete \cite{Spekkens2007a}.

We will not pursue this question further here, but instead focus on
the case of a hybrid region $XA$.  If the the data is the classical
variable (so that one can indeed learn its value) then Bayesian
conditioning has a straightforward generalization.  Upon learning that
$X=x$, the state of $A$ should be updated via
\begin{equation}
  \label{eq:Cond:Hybrid}
  \rho_A \rightarrow \sigma_{A|X=x},
\end{equation}
Recall that the general form of a state conditioned on a classical
variable is $\sigma_{A|X}=\sum_x \rho^A_x \otimes \Ket{x}\Bra{x}_X $,
where $\{ \rho^A_x\}$ is a set of normalized density operators and
that $\sigma_{A|X=x}$ is simply our notation for $\rho^A_x$.  The
elements of this set are called the components of $\sigma_{A|X}$, so
we may describe Bayesian conditioning as replacing $\sigma_{A}$ with
one of the components of $\sigma_{A|X}$.

Note that, as in the classical case, conditioning can be viewed as a
two-step process, wherein first the state of $X$ is updated to the
diagonal density operator for the point measure, $\rho^X_x =
\Ket{x}\Bra{x}_X$, and then belief propagation is used to determine
the posterior state of $A$, as follows
\begin{align}
  \rho^A_x & = \Tr[X]{\sigma_{A|X} \rho^X_x}, \\
  & =\Bra{x}_X \sigma_{A|X} \Ket{x}_X,\\
  & =\sigma_{A|X=x},
\end{align}
where we've made use of the fact that $\sigma_{A|X} = \sum_x
\sigma_{A|X=x} \otimes \Ket{x}\Bra{x}$. Note that this holds
regardless of the causal relation between $A$ and $X$ because a hybrid
conditional such as $\sigma_{A|X}$ does not distinguish between these
causal possibilities.

Recall that the rule for propagating unconditional beliefs about $X$
to beliefs about $A$ is $\rho_A = \Tr[X]{\sigma_{A|X} \rho_X}$.  In
conventional notation, this translates to $\rho_A = \sum_X P(X=x)
\sigma_{A|X=x}$.  Therefore, Bayesian conditioning is a process by
which a state $\rho_A$ is updated to an element $\sigma_{A|X=x}$
within a convex decomposition of $\rho_A$.

\subsection{Examples of Quantum Bayesian Conditioning}

\label{Cond:Exa}

In this section, we consider several examples of Bayesian
conditioning, based on the different causal scenarios discussed in
\S\ref{CS} and \S\ref{Bayes}.  In all these cases, conditioning the
state of a quantum region on a classical variable is the correct thing
to do in order to update the predictions or retrodictions that can be
made about other classical variables correlated with the region.  In
particular, in \S\ref{Cond:Remote}, we develop the application to
quantum steering, showing that the set of states of a region that can
be steered to by making remote measurements can be expressed compactly
in terms of conditioning and belief propagation.

\subsubsection{Conditioning on a Preparation Variable}

\label{Cond:Prep}

Consider again the preparation scenario depicted in
fig.~\ref{fig:CCS:QPreparation}, wherein a quantum region $A$ is
prepared in one of a set of states $\{\varrho_{A|X=x}\}$ depending on
the value of a classical variable $X$ with prior probability
distribution $P(X)$. This can alternatively be described by a
conditional state $\varrho_{A|X}$ and a diagonal state $\rho_X$.  In
this case, Bayesian conditioning of $A$ on the value $x$ of $X$
corresponds to updating from the ensemble average state $\rho_A =
\sum_x P(X=x) \varrho_{A|X=x}$ to the particular state
$\varrho_{A|X=x}$ corresponding to the value $x$ of $X$ that actually
obtains, which is clearly a reasonable thing to do.

The operational significance of this conditioning becomes apparent by
considering a measurement made on region $A$, described by a
conditional state $\varrho_{Y|A}$.  We now have a prepare-and-measure
experiment, where we are interested in making a predictive inference
about the measurement outcome from knowledge of the preparation
variable, as depicted in fig.~\ref{fig:Cond:Predict}.
\begin{figure}[htb]
  \includegraphics[scale=0.4]{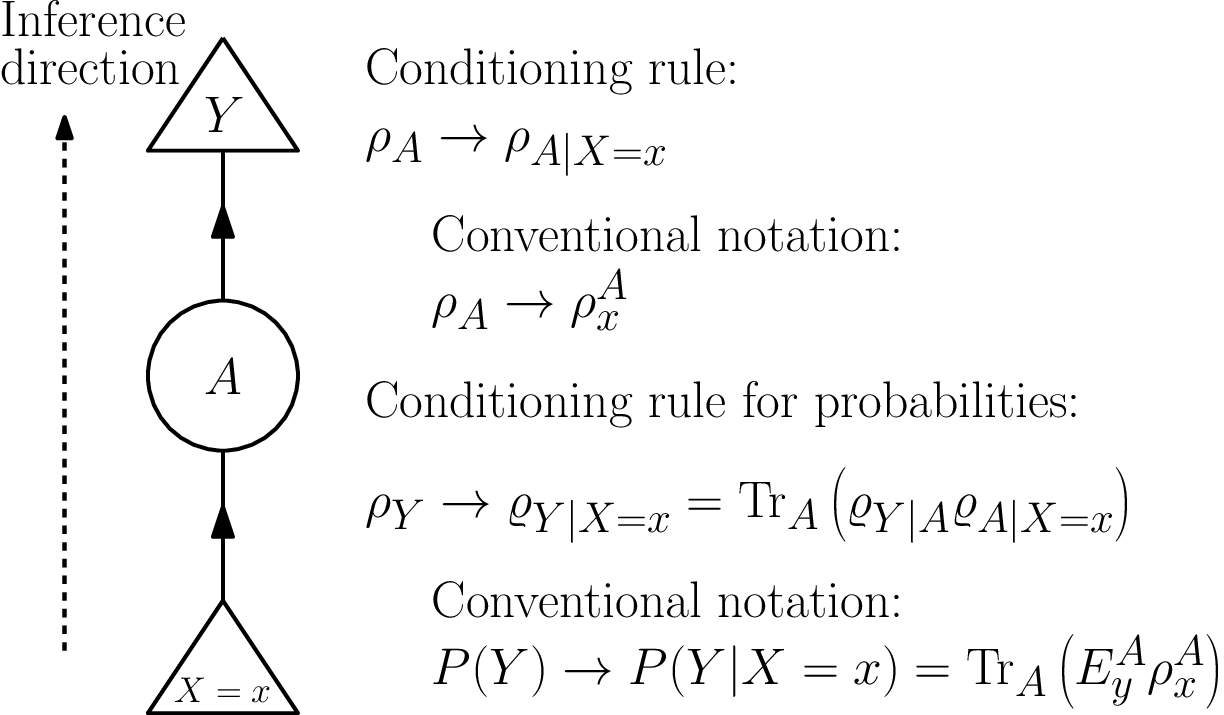}
  \caption{\label{fig:Cond:Predict}Predictive inference in a
    prepare-and-measure experiment.  We are interested in inferring
    the probability of the measurement outcome given knowledge of the
    preparation variable.}
\end{figure}

Eq.~\eqref{eq:Bayes:PredictiveTemp} gives the predictive expression
for the joint probability distribution for this experiment as
\begin{equation}
  \label{Cond:Joint}
  \varrho_{XY} = \Tr[A]{\varrho_{Y|A} \varrho_{A|X}} \Sprod \rho_X.
\end{equation}
From this, we can compute the marginal probability for $Y$ as
\begin{equation}
  \label{Cond:Born}
  \rho_Y = \Tr[A]{\varrho_{Y|A} \rho_A},
\end{equation}
where $\rho_A = \Tr[X]{\varrho_{A|X} \rho_X}$ is the ensemble
average state.  The conditional $\varrho_{Y|X} = \varrho_{XY}
\Sprod \rho_X^{-1}$ is given by
\begin{equation}
  \rho_{Y|X} = \Tr[A]{\varrho_{Y|A} \varrho_{A|X}},
\end{equation}
from which it follows that
\begin{equation}
  \label{Cond:BornCond}
  \rho_{Y|X=x} = \Tr[A]{\varrho_{Y|A} \varrho_{A|X=x}}.
\end{equation}

The transition from eq.~\eqref{Cond:Born} to eq.~\eqref{Cond:BornCond}
is just classical Bayesian conditioning of the probability for $Y$ on
the value of $X$.  Both expressions are representations of the Born
rule in terms of belief propagation, and eq.~\eqref{Cond:BornCond} can
be obtained from eq.~\eqref{Cond:Born} by replacing $\rho_A$ with
$\varrho_{A|X=x}$, which is just quantum Bayesian conditioning.  Thus,
quantum Bayesian conditioning on a preparation variable can be used as
an intermediate step in updating the probability distribution for a
measurement outcome by classical Bayesian conditioning.

Nothing changes if we consider the slightly more complicated scenario
where there is an intermediate channel between the preparation and
measurement, as depicted in fig.~\ref{fig:Bayes:Intervene}.  This is
because, as shown in \S\ref{Bayes:Retro}, the joint probability is
still given by eq.~\eqref{Cond:Joint}, where now $\varrho_{Y|A} =
\Tr[B]{\varrho_{Y|B}\varrho_{B|A}}$ describes the effective
measurement on $A$ that corresponds to the actual measurement made on
the later region $B$.  By similar reasoning, there could be an
arbitrary number of time-steps between the preparation and measurement
and conditioning the state of $A$ on the preparation variable would
still be the correct way to update the Born rule probabilities for
$Y$.

\subsubsection{Conditioning on the Outcome of a Measurement}

\label{Cond:Meas}

When conditioning on a measurement outcome, it is important to recall
that the conditional states formalism assigns states to regions rather
than to persistent systems.  Therefore, when we update the state of a
region by conditioning on a measurement outcome, the resulting
conditionalized state is assigned to the very same region that we
started with.  This is a different concept from the usual state update
rules that occur in the standard quantum formalism, such as the
projection postulate.  These standard rules apply to persistent
systems when we are interested in how the state of a system in a
region before the measurement gets mapped to its state in a region
after the measurement.  Because the updated state belongs to a
different region than the initial state, this is not an example of
pure conditioning in our framework.  Therefore, one should not think
that conditioning on a measurement outcome must reproduce the
projection postulate.  The way in which this kind of state update rule
is handled in the conditional states framework is discussed in
\S\ref{Cond:Direct}.

However, there are several other types of inference for which pure
conditioning on a measurement outcome is the correct update rule to
use.  In particular, conditioning can be used for making retrodictions
about classical variables in the past of the region of interest.  For
instance, in a quantum communication scheme, registering the outcome
of a measurement on the output of the channel leads us to infer
something about which of a set of classical messages was encoded in
the quantum state of the channel’s input.  The use of conditioning
for this sort of inference is the topic of this section.

Recall the measurement scenario depicted in
fig.~\ref{fig:CCS:QMeasurement}.  A measurement with outcomes labelled
by the classical variable $Y$ is implemented upon a quantum region
$A$.  In the predictive formalism, this experiment is described by an
input state $\rho_A$ and a causal conditional state $\varrho_{Y|A}$
(describing the measured POVM).  To condition $A$ on a value $y$ of
$Y$, Bayes' theorem must be applied to determine $\varrho_{A|Y}$, and
then the component $\varrho_{A|Y=y}$ gets picked out by conditioning.
This gives
\begin{equation}
  \label{eq:Cond:Retro}
  \rho_A \to \varrho_{A|Y=y}=\varrho_{Y=y|A}\Sprod \left ( \rho_A
  \rho_{Y=y}^{-1} \right ),
\end{equation}
or in conventional notation
\begin{equation}
  \label{eq:Cond:HybridCon}
  \rho_A  \rightarrow \rho_{y}^{A,\text{retr}} =  \frac{\rho_A^{\frac{1}{2}} E_y^A
    \rho_A^{\frac{1}{2}}}{\Tr[A]{E_y^A \rho_A}}.
\end{equation}
The state $\varrho_{A|Y=y}$ represents the state of a region $A$
\emph{prior to the measurement}, given the outcome of the measurement,
i.e.\ it is a retrodictive state.  Its operational significance is
that it allows one to make inferences about variables involved in the
preparation of $A$.

To see this, consider again the prepare-and-measure experiment, where
now we are interested in making a retrodictive inference from the
measurement outcome to the preparation variable, as depicted in
fig.~\ref{fig:Cond:Retrodict}.
\begin{figure}[htb]
  \includegraphics[scale=0.4]{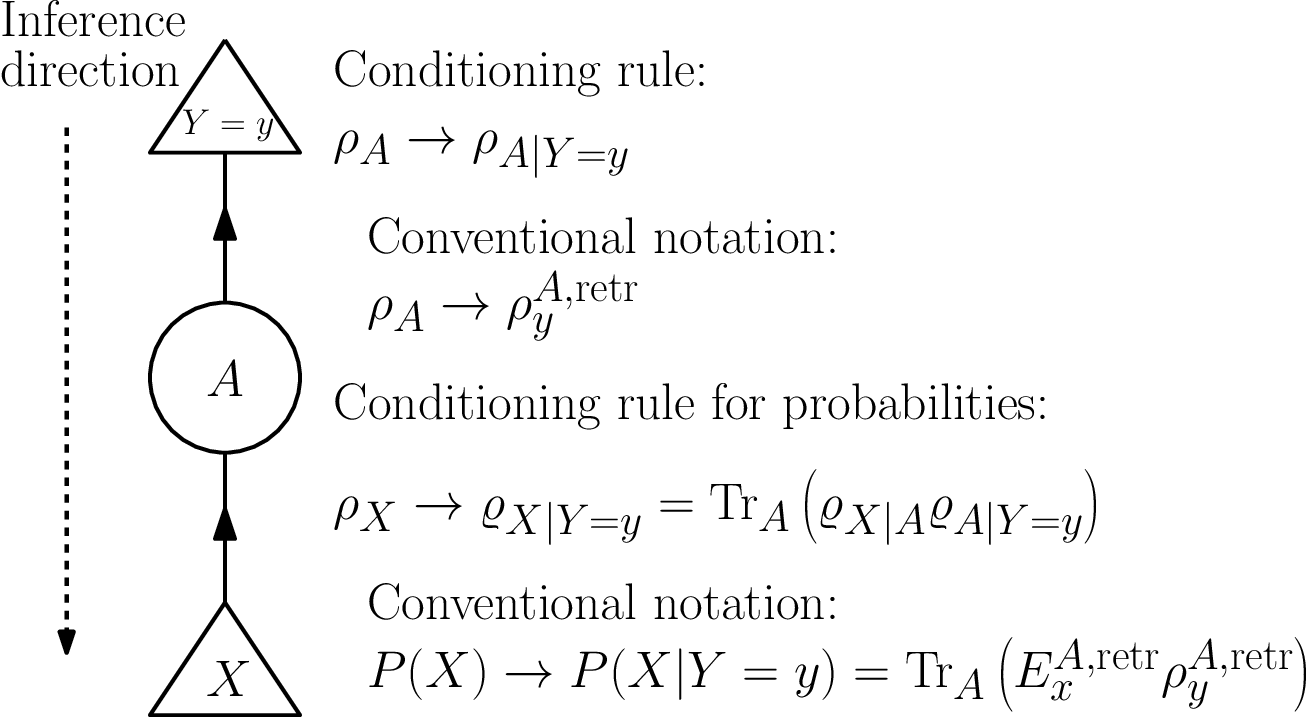}
  \caption{\label{fig:Cond:Retrodict}Retrodictive inference in a
    prepare-and-measure experiment.  We are interested in inferring the
    probability of the preparation variable given knowledge of the
    measurement outcome.}
\end{figure}

Recall from \S\ref{Bayes:Retro} that there is complete symmetry
between the predictive and retrodictive expressions for a
prepare-and-measure experiment under exchange of the preparation
variable $X$ with the measurement variable $Y$.  Thus, everything that
was said regarding the probability distribution for the measurement
outcome in the previous example, applies here to the probability
distribution for the preparation variable.  In particular, the
marginal probability distribution for $X$ is
\begin{equation}
  \rho_{X} = \Tr[A]{\varrho_{X|A} \rho_A},
\end{equation}
and conditioning this on $Y=y$ gives
\begin{equation}
  \rho_{X|Y=y} = \Tr[A]{\varrho_{X|A} \varrho_{A|Y=y}}.
\end{equation}
Both of these expressions are belief propagation representations of
the Born rule, with respect to the retrodictive POVM for $X$.  Thus,
conditioning $A$ on the outcome of the measurement can be used as an
intermediate step in updating the probability distribution for the
preparation variable in light of the measurement outcome.  The
retrodictive state appearing in the retrodictive Born rule simply gets
updated by Bayesian conditioning.

As in the previous example, nothing changes if there are intermediate
channels between the preparation and measurement.  We can simply use
conditional belief propagation to eliminate the additional regions.

\subsubsection{Conditioning on the Outcome of a Remote Measurement:
  Quantum Steering}

\label{Cond:Remote}

Finally, consider the case of a measurement made on a region $B$ that
is acausally-related to a region $A$, as depicted in
fig.~\ref{fig:CCS:Remote}.  This experiment is described by an acausal
joint state $\rho_{AB}$, and a conditional state $\varrho_{Y|B}$,
corresponding to the POVM measured on $B$.  We are interested in how
the state of the remote region $A$ is updated when we learn the
outcome of the measurement made on $B$.  The updated state of $A$
could then be used to predict the outcome of a measurement made on
$A$, corresponding to a causal conditional state $\varrho_{X|A}$.
This scenario is depicted in fig.~\ref{fig:Bayes:EPR}.  The results of
\S\ref{Bayes:Remote} establish that the mathematical description of
this experiment is formally equivalent to a prepare-and-measure
experiment with an intervening channel, the only difference being that
the causal conditional state $\varrho_{B|A}$ is replaced by an acausal
conditional state $\rho_{B|A}$.  This symmetry is enough to establish
that, as in the previous example, conditioning the state of $A$ on
$Y=y$ must be the correct way of updating the Born rule probabilities
for $X$.

However, it is worth developing this example in a little more detail,
since the ``remote collapse'' at $A$ that occurs upon measuring $B$ is
at the core of the EPR argument.  This has led to a study of the
ensembles of states for $A$ that can be obtained by measuring $B$, a
problem known as quantum steering \cite{Schrodinger1935,
  Schrodinger1936, Vujicic1988, Hughston1993, Verstraete2002,
  Wiseman2007, Jones2007, Herbut2009}.  The conditional states
formalism provides an elegant approach to this problem.

Recall from \S\ref{Bayes:Remote}, that the joint probability for $X$
and $Y$ can be computed via leftward belief propagation, which yields
eq.~\eqref{eq:Bayes:JointXYNoBLeft}, i.e.
\begin{equation}
  \label{eq:Cond:JointXYNoBLeft}
  \rho_{XY} = \Tr[A]{\varrho_{X|A} \rho_{A|Y}} \Sprod
  \rho_Y.
\end{equation}
This formula is obtained by first applying Bayesian inversion to
$\varrho_{Y|B}$ to determine
\begin{equation}
  \varrho_{B|Y} = \varrho_{Y|B} \Sprod \left ( \rho_B \rho_Y^{-1}\right ),
\end{equation}
where $\rho_Y = \Tr[B]{\varrho_{Y|B} \rho_B}$ is the Born rule
probability distribution for the measurement outcome.  Then,
conditional belief propagation is used to obtain $\rho_{A|Y} =
\Tr[B]{\rho_{A|B}\varrho_{B|Y}}$.

Eq.~\eqref{eq:Cond:JointXYNoBLeft} is formally equivalent to the
retrodictive expression for a prepare-and-measure experiment, so the
rationale for conditioning is exactly the same as in the previous
example.  Specifically, the marginal probability distribution for $X$
is
\begin{equation}
  \rho_X = \Tr[A]{\varrho_{X|A} \rho_A},
\end{equation}
and conditioning on $Y=y$ gives
\begin{equation}
  \rho_{X|Y=y} = \Tr[A]{\varrho_{X|A} \rho_{A|Y=y}}.
\end{equation}
Thus, conditioning the state of $A$ on $Y=y$ is an appropriate
intermediate step for updating the Born rule probabilities for $X$.

The following proposition summarizes this result and translates it
into conventional notation.
\begin{proposition}
  \label{prop:Cond:Steering}
  Let $\rho_{AB}$ be the joint state of two acausally-related regions.
  Suppose that the POVM corresponding to the conditional state
  $\varrho_{Y|B}$ is measured on $B$ and the outcome $Y=y$ is
  obtained.  Then, the state of region $A$ should be updated from the
  prior $\rho_A = \Tr[B]{\rho_{AB}}$ to $\rho_{A|Y=y}$, where
  \begin{equation}
    \label{eq:Cond:SteerCond}
    \rho_{A|Y} = \Tr[B]{\rho_{A|B} \varrho_{B|Y}},
  \end{equation}
  $\rho_{A|B} = \rho_{AB} \Sprod \rho_B^{-1}$, and
  \begin{equation}
    \label{eq:Cond:SteerBayesCond}
    \varrho_{B|Y} = \varrho_{Y|B} \Sprod (\rho_{B} \rho_Y^{-1}).
  \end{equation}
  Let the components of the conditional state $\varrho_{Y|B}$ be the
  elements of the POVM $\left \{ E^B_y \right \}$ and let
  $\mathfrak{E}_{A|B}$ be the map that is Jamio{\l}kowski-isomorphic
  to $\rho_{A|B}$. Then, the updated state of $A$ is
  \begin{equation}
    \label{eq:Cond:Steer}
    \rho_y^{A} = \mathfrak{E}_{A|B}(\rho_y^{B}),
  \end{equation}
  where
  \begin{equation}
    \label{eq:Cond:SteerBayes}
    \rho_y^{B} = \frac{\rho_B^{1/2} E^B_y \rho_B^{1/2}}{\Tr[B]{E^B_y \rho_B}}.
  \end{equation}
\end{proposition}

In the above analysis, the method used to compute $\rho_{A|Y=y}$ is to
first apply Bayes' theorem to $\varrho_{Y|B}$ and then apply
conditional belief propagation.  This is the acausal analogue of
performing the calculation in the retrodictive formalism.  However, we
could equally well apply belief propagation first to obtain
\begin{equation}
  \rho_{Y|A} = \Tr[B]{\varrho_{Y|B}\rho_{B|A}},
\end{equation}
and then apply Bayes' theorem to obtain
\begin{equation}
  \rho_{A|Y=y} = \rho_{Y=y|A} \Sprod \left (  \rho_A \rho_Y^{-1} \right ).
\end{equation}
This is the acausal analogue of performing the calculation in the
Heisenberg picture, and it is straightforward to check that it gives
the same result.

In conventional notation, this amounts to first determining the
effective POVM on $A$ that is performed when $B$ is measured via
\begin{equation}
  E^A_y = \mathfrak{E}_{A|B}^{\dagger} \left ( E^B_y \right ),
\end{equation}
and then determining the updated state via
\begin{equation}
  \rho^A_y = \frac{\rho_A^{\frac{1}{2}} E^A_y
    \rho_A^{\frac{1}{2}}}{\Tr[B]{E^A_y \rho_A}}.
\end{equation}
Expressions equivalent to these conventional formulas have appeared
previously in \cite{Verstraete2002}, the only difference being the
appearance of a transpose, due to the use of the Choi isomorphism
rather than the Jamio{\l}kowski isomorphism in \cite{Verstraete2002}.

We have seen that the remote collapse $\rho_A \to \rho_{A|Y=y}$ is an
instance of quantum Bayesian conditioning.  Within interpretations of
quantum theory wherein quantum states are considered to describe
reality completely, the change of the state of $A$ as a result of a
distant measurement upon $B$ is sometimes considered to be an instance
of action at a distance.  Indeed, Einstein criticized the Copenhagen
interpretation on exactly these grounds.  On the other hand, the
analysis above shows clearly that if one views quantum theory as a
theory of Bayesian inference, then upon learning the outcome of a
measurement on $B$, all that occurs is that one's beliefs about $A$
are updated.  No change to the physical state of $A$ is required
within such an approach.  This interpretation is bolstered by the
formal equivalence to the case of conditioning a region on the outcome
of a subsequent measurement, which does not seem to imply retrocausal
influences (although for realist interpretations it has been argued
that a imposing a particular symmetry principle does imply
retrocausality in this scenario \cite{Evans2010, Wharton2011}).

Einstein anticipated such an epistemic interpretation of remote
collapse in his writing, as is argued in \cite{Spekkens2010}.
However, he most likely thought that the probabilities represented by
quantum states could be probabilities for the values of physical
variables (possibly hidden), and that these could satisfy classical
probability theory.  However, by virtue of Bell's theorem
\cite{Bell1964}, such an interpretation of remote collapse is not
possible within the standard framework for hidden variable theories.
We evade this no-go result here by understanding remote collapse as
Bayesian updating within a non-commutative probability theory rather
than within classical probability theory.  In itself, this does not
provide a viable realist interpretation of quantum theory, but it does
suggest that an acceptable ontology for quantum theory ought not to
include the quantum state.  Finally, note that our interpretation of
remote steering is broadly in harmony with that of Fuchs
\cite{Fuchs2003a}.  However, the view of quantum theory presented here
differs from that of Fuchs in that he views quantum theory as a
restriction upon classical probability theory whereas we consider it
to be a generalization thereof.

\subsection{Why the Projection Postulate is not an Instance of
  Bayesian Conditioning}

\label{Cond:Direct}

Finally in this section, we deal with the elephant in the room --- the
projection postulate.  Some authors have argued that the projection
postulate is a quantum analogue of Bayesian conditioning (see
e.g. \cite{Bub1978, Bub2007a}).  However, the projection postulate is
not an instance of quantum Bayesian conditioning as defined above.  In
this section, we discuss the relationship between quantum Bayesian
conditioning and the projection postulate (and quantum instruments in
general) at some length, in order to dispel the misconception that
projection is a kind of conditioning.  After pointing out the formal
differences between the two, we explain the different types of update
rule that are associated with a measurement in both classical
probability theory and quantum theory, pointing out where conditioning
and projection fit into this picture.  Then, we explain how, in the
conditional states formalism, the projection postulate should be
thought of as a composite operation, consisting of belief propagation
to a later region followed by Bayesian conditioning.  This is broadly
in line with the treatment of quantum measurements advocated by Ozawa
\cite{Ozawa1997, Ozawa1998}.  Finally, we deal with the possible
objection that, although the classical analogue of the projection
postulate is not just Bayesian conditioning, it can be thought of as
conditioning combined with a simple relabelling of the system
variable.  This argument does not apply in quantum theory because,
unlike in a classical theory, any informative measurement necessarily
disturbs the system being measured.  Although this is well known, we
present a formulation of information-disturbance in terms of
conditional states, which makes it clear that no quantum instrument
can be thought of as conditioning combined with relabelling.  This may
be of interest in its own right, as it emphasizes the similarity
between information-disturbance and other trade-offs in quantum
theory, such as the monogamy of entanglement.

In the conventional formalism, when a projective measurement
$\{\Pi^A_{y}\}$ is made on a system $A$, the L{\"u}ders-von Neumann
projection postulate says that, upon learning the outcome $y$, the
initial state $\rho_A$ should be updated via
\begin{equation}
  \label{eq:ProjPost}
  \rho_{A} \rightarrow \frac{\Pi_y^A\rho_A \Pi_y^A}{\Tr[A]{\Pi_y^A
      \rho_A}}.
\end{equation}
For a general POVM $\{E_y^A\}$, there is a natural generalization of
the projection postulate, given by
\begin{equation}
  \label{eq:GeneralCollapsePost}
  \rho_A \rightarrow \frac{(E_y^A)^{\frac{1}{2}} \rho_A
    (E_y^A)^{\frac{1}{2}}}{\Tr[A]{E_y^A\rho_A}}.
\end{equation}
This generalization has also been proposed as a quantum analogue of
Bayesian conditioning \cite{Jacobs2002, Jacobs2005}.

On the other hand, in \S\ref{Cond:Meas} the rule for conditioning a
state on the outcome of a measurement was found to be
\begin{equation}
  \label{eq:Bayes:HybridCon}
  \rho_A \rightarrow \frac{\rho_A^{\frac{1}{2}} E_y^A
    \rho_A^{\frac{1}{2}}}{\Tr[A]{E_y^A \rho_A}}.
\end{equation}

This is distinct from eq.~\eqref{eq:GeneralCollapsePost} because the
roles of $\rho_A$ and $E_y^A$ have been interchanged.  Furthermore,
eq.~\eqref{eq:GeneralCollapsePost} is not equivalent to an equation of
the form of eq.~\eqref{eq:Bayes:HybridCon} (even allowing a different
POVM to appear therein) because the map $\rho_A \rightarrow \left (
  E_y^A \right )^{\frac{1}{2}} \rho_A \left ( E_y^A \right
)^{\frac{1}{2}}$, like all quantum instruments, is linear on the set
of states on $\Hilb[A]$, whereas the map $\rho \rightarrow
\rho_A^{\frac{1}{2}} E_y^A \rho_A^{\frac{1}{2}}$ is nonlinear on the
state space.

If you are inclined to view the projection postulate as the correct
quantum generalization of Bayesian conditioning, then you could take
this as evidence \emph{against} the idea that the conditional states
formalism provides an adequate theory of quantum Bayesian inference.
We shall therefore go to some length to defend the claim that neither
the projection postulate nor any other quantum instrument is an analogue
of classical Bayesian conditioning.

Consider the scenario depicted in fig.~\ref{fig:Cond:QInstrument}.  A
quantum system located in region $A$ and described by the state
$\rho_A$ is subjected to a measurement with outcomes labelled by $Y$,
and the system persists after the measurement.  At this later time, it
is represented by a region $A'$ which, as always in the present
formalism, we distinguish from region $A$, but which is associated
with a Hilbert space of the same dimension.  As discussed in
\S\ref{CCS:QI}, a quantum instrument, $\{ \mathcal{E}_y^{A'|A} \}$,
determines how the state of $A'$ is related to the state of $A$, where
$\mathcal{E}_y^{A'|A} \left ( \rho_A \right )$ is the unnormalized
state of $A'$ obtained when $Y=y$.

Now consider the classical analogue of this scenario, depicted in
fig.~\ref{fig:Cond:CInstrument}.  A classical system described by the
variable $R$, and assigned a distribution $P(R)$, is subjected to a
(possibly noisy) measurement, resulting in the outcome $Y$, which is a
random variable that depends on $R$ through a conditional probability
distribution $P(Y|R)$.  The system persists after the measurement,
where it is described by a random variable $R'$.  The value of $R'$ is
presumed to depend probabilistically on $R$, and the nature of this
dependence may vary with the outcome $Y$.  This is captured by a
conditional probability $P(Y,R'|R)$, which is the classical analogue of
a quantum instrument.
\begin{figure}[htb]
  \begin{minipage}{7.8cm}
    \raggedright
    \subfigure[]{
      \includegraphics[scale=0.4]{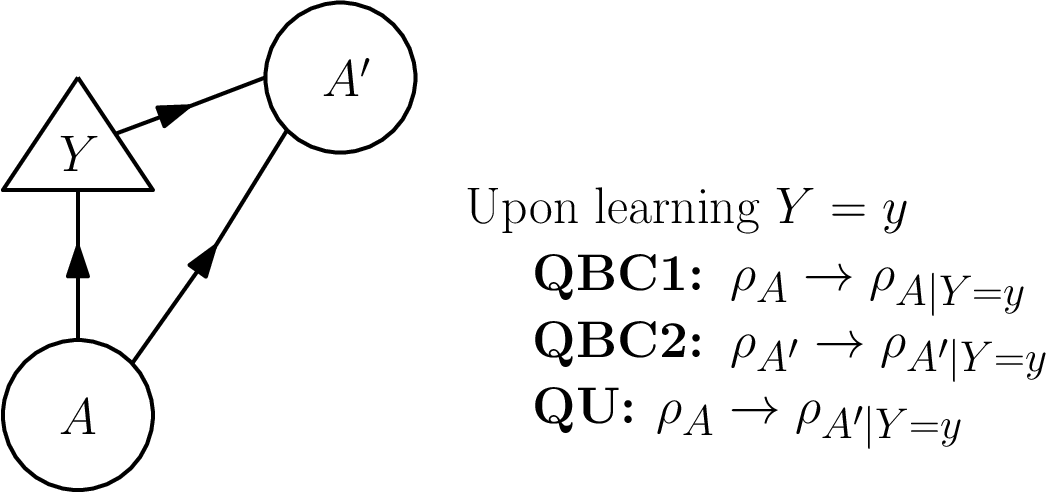}
      \label{fig:Cond:QInstrument}
    }
    \subfigure[]{
      \includegraphics[scale=0.4]{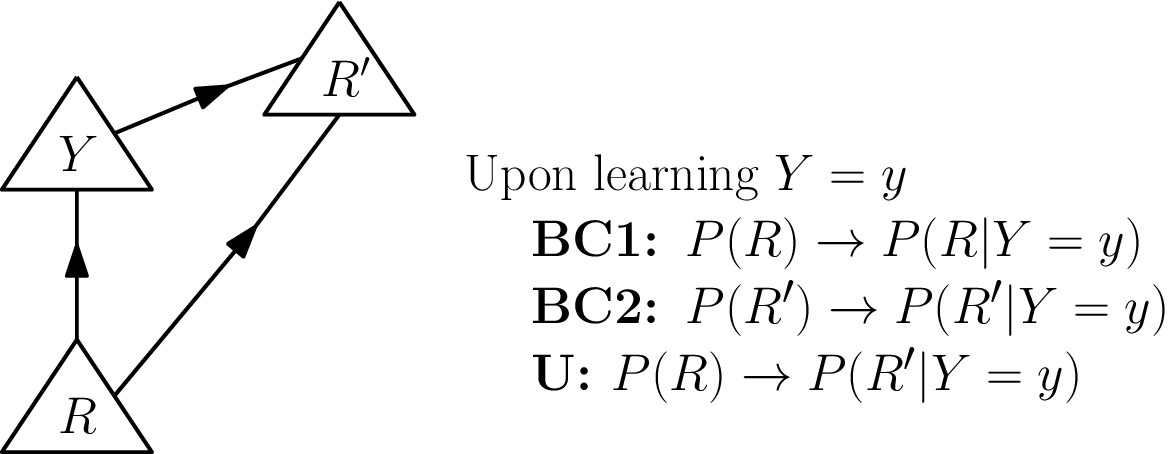}
      \label{fig:Cond:CInstrument}
    }
  \end{minipage}
  \caption{\label{fig:Cond:Instrument}Causal diagrams for the state
    update rules associated with quantum and classical
    measurements. \subref{fig:Cond:QInstrument} A quantum instrument,
    representing how the state of a quantum persistent system changes
    after a measurement.  \subref{fig:Cond:CInstrument} The classical
    analogue of a quantum instrument, representing how the state of a
    classical persistent system changes after a measurement.}
\end{figure}

It is useful to distinguish three kinds of update rules that might be
considered in this scenario, as defined in
fig.~\ref{fig:Cond:CInstrument}.  To describe the difference between
these rules, it is useful to introduce some terminology. A
distribution over $R$ or $R'$ is said to be a \emph{prior}
distribution if it is not conditioned on the value of $Y$, and it is a
\emph{posterior} distribution if it is conditioned on the value of
$Y$.  The temporal ordering that is implicit in this prior/posterior
terminology specifies whether the distribution characterizes the
knowledge you have \emph{before} learning the value of $Y$ or the
knowledge you have \emph{after} learning the value of $Y$.  In other
words, it refers to the time in your epistemological history, relative
to the event of learning $Y$.  On the other hand, the system's
configuration at the time before the occurrence of the measurement
will be called its \emph{initial} configuration and its configuration
at the time after the occurrence of the measurement its \emph{final}
configuration.  $R$ is the initial configuration and $R'$ the final.
The temporal ordering implicit in this latter distinction refers to
the time in the system's ontological history.

Strictly speaking, Bayesian conditioning is a rule that updates what
one knows about \emph{one and the same} variable upon the acquisition
of new information. In other words, it maps a prior distribution about
some variable to a posterior distribution for that same variable.
Consequently, rule \textbf{BC1}, which maps the prior distribution of
the initial configuration of the system to the posterior distribution
of the initial configuration of the system, is an instance of Bayesian
conditioning (it is analogous to updating a retrodictive state as
discussed in \S\ref{Cond:Meas}).  Rule \textbf{BC2} is also an
instance of Bayesian conditioning: it maps the prior distribution of
the final configuration of the system to the posterior distribution of
the final configuration of the system.  However, the rule \textbf{U}
is \emph{not} an instance of Bayesian conditioning because it maps the
prior distribution of the \emph{initial} configuration of the system
to the posterior distribution of the \emph{final} configuration.  In
other words, if one considers $R$ and $R'$ to be distinct variables,
then any map from a distribution over one of them to a distribution
over the other cannot be an instance of Bayesian conditioning.

We now return to the quantum scenario, this time using the conditional
states formalism and paying attention to the analogy with the
classical case.  The measurement is associated with the causal
conditional state $\varrho_{YA'|A}$ that corresponds to the quantum
instrument $\left \{ \mathcal{E}_y^{A'|A} \right \}$. In the quantum
conditional states formalism, there are analogues of each of the three
rules we considered above.  These are specified in
fig.~\ref{fig:Cond:QInstrument}.  The rule \textbf{QBC1} corresponds
to updating a retrodictive state, as considered in \S\ref{Cond:Meas}.
The von Neumann-L\"{u}ders projection postulate is clearly an instance
of rule \textbf{QU}.  If one stipulates that quantum Bayesian
conditioning is a rule that updates the quantum description of
\emph{one and the same} region upon acquiring new information, i.e.\
that it maps a prior state for a region to a posterior state for the
same region, then \textbf{QBC1} and \textbf{QBC2} are instances of
quantum Bayesian conditioning, but \textbf{QU} is not.

\subsubsection{State-Update Rules as a Combination of Belief
  Propagation and Bayesian Conditioning}

If \textbf{QU} is not an instance of Bayesian conditioning, then what
is its status within our framework?  We now show that it is a
composite of two operations: belief propagation followed by Bayesian
conditioning.

First consider the classical analogue.  The analogue of the projection
postulate (or any state update rule arising from an instrument) is
given by rule $\textbf{U}$: $P(R) \to P(R'|Y=y)$.  This can be
obtained by combining an instance of belief propagation, namely
\begin{equation}
  \label{eq:CnonselectiveBP}
  P(R) \to P(R') = \sum_R P(R'|R)P(R),
\end{equation}
followed by the rule $\textbf{BC2}$: $P(R') \to P(R'|Y=y)$, which is
an instance of Bayesian conditioning.  It is useful to express both
these steps in terms of the quantities that are given in the problem,
namely, the conditional $P(Y,R'|R)$ and the prior over $R$, $P(R)$.
The conditional probability distribution $P(R'|R)$ used in
eq.~\eqref{eq:CnonselectiveBP} is simply the marginal of $P(Y,R'|R)$,
i.e. $P(R'|R) = \sum_{Y} P(Y,R'|R)$. Meanwhile, the expression for the
conditional is $P(R'|Y)=P(R',Y)/P(Y)$, where $P(R',Y) = \sum_R
P(Y,R'|R)P(R)$.  Setting $Y=y$ gives the Bayesian conditioning step as
\begin{multline}
  P(R') \to P(R'|Y=y) = \\ \frac{\sum_{R} P(Y=y,R'|R)P(R)}{P(Y=y)}.
\end{multline}

The quantum analogue of this is straightforward.  Quantum state update
rules, such as the projection postulate, are of the form
$\textbf{QU}$: $\rho_A \to \varrho_{A'|Y=y}$.  This is simply a
sequential combination of quantum belief propagation
\begin{equation}
  \label{eq:nonselectiveBP}
  \rho_A \to  \rho_{A'} = \Tr[A]{\varrho_{A'|A} \rho_A},
\end{equation}
with quantum Bayesian conditioning via $\textbf{QBC2}$: $\rho_{A'} \to
\varrho_{A'|Y=y}$.

Again, it is useful to express each of these steps in terms of the
quantities that are given in the problem: the causal conditional state
$\varrho_{YA'|A}$ and the prior $\rho_A$.  The causal conditional
state $\varrho_{A'|A}$ used in eq.~\eqref{eq:nonselectiveBP} is simply
a reduced state of $\varrho_{YA'|A}$, i.e. $\varrho_{A'|A} =
\Tr[Y]{\varrho_{YA'|A}}$.  To gain some intuition for this step, we
translate it into conventional notation.  If the quantum instrument
associated with $\varrho_{YA'|A}$ is denoted $\{ \mathcal{E}_y^{A'|A}
\}$, then eq.~(\ref{eq:nonselectiveBP}) becomes
\begin{equation}
  \label{eq:Qnonselective}
  \rho_A \to \rho_{A'} = \mathcal{E}_{A'|A}(\rho_A),
\end{equation}
where
\begin{equation}
  \mathcal{E}_{A'|A} = \sum_y \mathcal{E}_y^{A'|A}.
\end{equation}
The map $\mathcal{E}_{A'|A}$ is the \emph{non-selective} update map.
It is the appropriate map to apply when you know that the measurement
has been performed, but you do not know which outcome occurred.  The
standard update map, which is appropriate when one also knows the
outcome, is the \emph{selective} update map.  The projection postulate
and its generalization to POVMs are instances of selective update
maps.  Applying the non-selective update map is just an instance of
quantum belief propagation.

Turning to the Bayesian conditioning step, we have
$\varrho_{A'|Y}=\rho_{A'Y} \Sprod \rho_Y^{-1}$, where $\rho_{A'Y} =
\Tr[A]{\varrho_{YA'|A} \rho_A}$.  Combining these and setting $Y=y$
gives
\begin{equation}
  \rho_{A'} \to \varrho_{A'|Y=y}= \Tr[A]{\varrho_{Y=y,A'|A} \rho_A}
  \Sprod \rho_{Y=y}^{-1}.
\end{equation}
In conventional notation, this translates into
\begin{equation}
  \rho_{A'} \to \rho_y^{A'} =
  \frac{\mathcal{E}_y^{A'|A}(\rho_A)}{P(Y=y)}.
\end{equation}
Given the expression for $\rho_{A'}$ in eq.~(\ref{eq:Qnonselective}),
we see that $\textbf{QBC2}$ is simply a transition from your prior
about the system output by the measurement, the result of applying the
\emph{non-selective} update map, to your posterior about the system
output by the measurement, the result of applying the \emph{selective}
update map and normalizing.  That this transition from non-selective
to selective updates should be regarded as analogous to Bayesian
conditioning has previously been argued by Ozawa \cite{Ozawa1997,
  Ozawa1998}.

In fact, the rule $\textbf{QBC2}$ is a particular example of the kind
of Bayesian conditioning considered in \S\ref{Cond:Prep}.  Every
measurement with output region $A'$ can be considered to define a
preparation of $A'$ for every outcome (assuming a fixed input state).
The set of states prepared is given by the components of
$\varrho_{A'|Y}$, which we can compute from the causal joint state
$\varrho_{YA'A}=\varrho_{YA'|A} \Sprod \rho_A$ by tracing over $A$ and
conditioning on $Y$.  The rule $\textbf{QBC2}$ is then just Bayesian
conditioning on $Y$, thought of as a classical preparation variable.

\subsubsection{No Information Gain Without Disturbance}

We have argued that neither the classical rule $\textbf{U}$ nor the
quantum state update rule $\textbf{QU}$ are instances of Bayesian
conditioning.  However, a skeptic might counter that our argument is
an artefact of our insistence that the system before and after the
measurement should be given different labels.  If the conditional
distribution $P(R'|R)$ in the belief propagation rule of
eq.~(\ref{eq:CnonselectiveBP}) has the form
\begin{equation}
  P(R'=r'|R=R)=\delta_{r',r},
\end{equation}
where $\delta_{r',r}$ is the Kronecker-delta function, then $R$ and
$R'$ are perfectly correlated and consequently $P(R'|Y=y)$ has
precisely the same functional form as $P(R|Y=y)$.  In this case, one
could say that the rule $\textbf{U}$ is \emph{effectively} just
Bayesian conditioning.

If $P(R'|R)$ is just a delta function, then we say that the
measurement is \emph{non-disturbing}.  Recall that for every $P(Y|R)$
that characterizes the outcome probabilities for a measurement, there
are many conditionals (i.e.\ classical instruments) $P(Y,R'|R)$ that
might characterize its transformative aspect and are consistent with
$P(Y|R)$.  It is not difficult to see that, among all such
conditionals, there is always one for which the measurement is
non-disturbing, namely,
\begin{equation}
  P(Y,R'=r'|R=r)=P(Y|R=r)\delta_{r',r}.
\end{equation}
Of course, there are also many ways of implementing a measurement of
$P(Y|R)$ such that it \emph{is} disturbing. Therefore, whilst the
update rule $\textbf{U}$ is not an instance of Bayesian conditioning
for every possible implementation of the measurement, there is always
at least one implementation such that it is effectively just Bayesian
conditioning.

The obvious question to ask at this point is whether it is possible to
implement a \emph{quantum} measurement in a non-disturbing way, such
that the associated quantum update rule $\textbf{QU}$ (perhaps the
projection postulate, perhaps some other rule) is effectively just
quantum Bayesian conditioning.  For the measurement to be
non-disturbing, the causal conditional state $\varrho_{A'|A}$ in the
belief propagation rule of eq.~(\ref{eq:nonselectiveBP}) would have to
be of the form
\begin{equation}
  \varrho_{A'|A}= \sum_{j,k} \Ket{j}\Bra{k}_A \otimes \Ket{k}\Bra{j}_{A'},
\end{equation}
i.e it would need to be the partial transpose of the (unnormalized)
maximally entangled state.  This corresponds to perfect correlation
between $A$ and $A'$ because it is Jamio{\l}kowski-isomorphic to the
identity channel.  If the belief propagation step in the rule
$\textbf{QU}$ were of this form, then $\varrho_{A'|Y=y}$ would have
the same functional form as $\varrho_{A|Y=y}$, and the overall quantum
update rule $\textbf{QU}$ would be effectively just Bayesian
conditioning.

Of course, we are only interested in the case where the measurement is
nontrivial.  If the measurement gives \emph{no} information about $A$,
then the posterior is the same as the prior and no real conditioning
has occurred.  Therefore, we restrict our attention to the case where
some information is gained.  In the language of conditional states,
the only kind of measurement that yields no information about the
input state is one associated with a causal conditional state that
factorizes, that is, one of the form $\varrho_{Y|A}= \rho_Y$ (recall
that, in our notation there is an implicit $\otimes I_A$ on the right
hand side of this equation).  In conventional notation, this
corresponds to a POVM of the form $\{P(Y=y)I_A\}$, which generates a
random outcome $Y=y$ from the distribution $P(Y)$ regardless of the
state of $A$.  We are interested in nontrivial measurements for which
$\varrho_{Y|A}$ does not factorize in this way.

With these definitions in hand, the answer to our question is a
resounding ``no'' --- a quantum state update rule of the form
$\textbf{QU}$ can \emph{never} be effectively just Bayesian
conditioning because, unlike the classical case, in quantum theory
information gain necessarily implies a disturbance. This prevents any
$\textbf{QU}$ rule, such as the projection postulate, from being pure
Bayesian conditioning.  Whilst this fact is well-known, it is
instructive to prove it in the conditional states formalism.

\begin{theorem}[No information gain without disturbance]
  \label{thm:InfoGainDisturbance}
  Consider a measurement described by an instrument associated with
  the causal conditional state $\varrho_{YA'|A}$.  It is impossible
  for this measurement to be both informative about $A$
  ($\varrho_{Y|A} \ne \rho_Y$) and non-disturbing ($\varrho_{A'|A}=
  \sum_{j,k} \Ket{j}\Bra{k}_A \otimes \Ket{k}\Bra{j}_{A'}$).
\end{theorem}

The proof is a causal analogue of the monogamy of entanglement (see
\cite{Leifer2006} for related ideas).

\begin{proof}
  The operator $\varrho_{YA'|A}$ is the partial transpose over $A$ of
  an acausal conditional state $\rho_{YA'|A}$.  Combining this with
  the maximally mixed state gives a valid tripartite acausal state via
  \begin{equation}
    \rho_{YA'A}=\rho_{YA'|A}\Sprod I_A/d =\rho_{YA'|A}/d,
  \end{equation}
  where $d$ is the dimension of $\Hilb[A]$.  The condition that the
  measurement be non-disturbing is equivalent to
  $\rho_{A'|A}=\sum_{j,k} \Ket{j}\Bra{k}_A \otimes
  \Ket{j}\Bra{k}_{A'}$, which implies that the tripartite state
  $\rho_{YA'A}$ should have a reduced state on $A'A$ that is maximally
  entangled.  Meanwhile, the condition that the measurement be
  informative is equivalent to $\rho_{Y|A} \ne \rho_Y$, which implies
  that the reduced state on $YA$ of $\rho_{YA'A}$ should not be a
  product state.  But, by the monogamy of entanglement, any tripartite
  state $\rho_{YA'A}$ for which $\rho_{AA'}$ is maximally entangled
  must have a product state for its reduced state $\rho_{YA}$.  Hence,
  both conditions cannot be satisfied simultaneously.
\end{proof}

We end with another, less obvious, disanalogy between the quantum and
classical cases that bears on the question of how to interpret the
quantum collapse rule in our framework.  In the previous section, we
showed that the classical update rule $\textbf{U}$ could be decomposed
into belief propagation followed by Bayesian conditioning.  We could
just as well have decomposed it in the opposite order: Bayesian
conditioning followed by belief propagation.  Specifically, we could
first apply $\textbf{BC1}$: $P(R) \to P(R|Y=y)$, and then propagate
the conditioned state via $P(R|Y=y) \to P(R'|Y=y) = \sum_R
P(R'|R,Y=y)P(R|Y=y)$.

It is natural to ask whether such a reverse-order decomposition is
possible in the quantum case.  That is, can $\textbf{QU}$:
$\rho_{A}\to \varrho_{A'|Y=y}$ be decomposed into $\textbf{QBC1}$:
$\rho_A \to \varrho_{A|Y=y}$, followed by belief propagation
$\varrho_{A|Y=y}\to \varrho_{A'|Y=y}$?  Perhaps surprisingly, this
cannot be done.  The belief propagation would have to have the form
$\varrho_{A'|Y=y}= \Tr[A]{\varrho_{A'|A,Y=y} \varrho_{A|Y=y}}$.  But
to compute $\varrho_{A'|A,Y=y}$ from $\varrho_{Y=y,A'|A}$ we need to
move $Y$ from the left of the conditional to the right while keeping
$A$ on the right.  Classical analogy suggests that this could be done
using a conditionalized form of the quantum Bayes' theorem, i.e.
$\varrho_{A'|A,Y=y}=\varrho_{Y=y,A'|A} \Sprod \varrho_{Y=y|A}^{-1}$.
Unfortunately, in the conditional states formalism, valid equations do
not necessarily remain valid when we conditionalize each term (this is
discussed further in \S\ref{Lim}).  In particular, the conditionalized
form of the quantum Bayes' theorem is not valid.  Not every intuition
from classical Bayesian inference carries over into the conditional
states formalism.

\section{Related Work}

\label{Related}

In this section, quantum conditional states are compared to other
proposals for quantum generalizations of conditional probability and
the conditional states formalism is compared to several recently
proposed operational reformulations of quantum theory.

\subsection{Comparison to Other Quantum Generalizations of Conditional
  Probability}

\label{Related:Comp}

Several quantum generalizations of conditional probability have been
proposed in the literature, so it is worth comparing their relative
merits to the conditional state formalism developed here.

Firstly, Cerf and Adami have proposed an alternative definition of an
acausal conditional state \cite{Cerf1997, Cerf1997a, Cerf1998} (their
definition does not extend to the causal case).  For a bipartite state
$\rho_{AB}$, the Cerf-Adami conditional state is
\begin{equation}
  \rho^{\left ( \infty \right )}_{B|A}  = \exp \left ( \log \rho_{AB} - \log
    \rho_A \otimes I_B \right ).
\end{equation}
This proposal has a close connection to the calculus of quantum
entropies, since the conditional von Neumann entropy of a state
$\rho_{AB}$ can be succinctly written as
\begin{equation}
  S(B|A) = -\Tr[AB]{\rho_{AB}\log \rho^{\left ( \infty \right )}_{B|A}},
\end{equation}
which is analogous to the classical formula for the conditional
Shannon entropy
\begin{equation}
  H(S|R) = -\sum_{R,S} P(R,S) \log P(S|R).
\end{equation}
Similar compact formulas hold for other information theoretic
quantities, such as the quantum mutual information and conditional
mutual information.

In \cite{Leifer2008}, Leifer and Poulin introduced a family of
conditional states, again restricted to the acausal case, indexed by a
positive integer $n$ and given by
\begin{equation}
  \label{eq:Disc:nCond}
  \rho^{\left ( n \right )}_{B|A} = \left ( \rho^{\frac{1}{n}}_{AB} \Sprod
    \rho_A^{-\frac{1}{n}} \right )^n.
\end{equation}
This unifies the Cerf-Adami conditional state with the definition used
in the present work in the sense that $\rho_{B|A} = \rho^{(1)}_{B|A}$
and $\rho^{(\infty)}_{B|A} = \lim_{n \rightarrow \infty}
\rho^{(n)}_{B|A}$.

The main concern of \cite{Leifer2008} was the generalization of
graphical models and belief propagation algorithms to quantum theory,
and their use in simulating many-body quantum systems and decoding
quantum error correction codes.  In this context, the $n=1$ and $n
\rightarrow \infty$ cases are particularly interesting.  The $n
\rightarrow \infty$ case is the natural one to use for simulating
many-body systems and it allows for a simple generalization of one
direction of the classical Hammersley-Clifford theorem, which
characterizes the states on Markov Networks.  On the other hand, the
$n=1$ case is more useful for decoding quantum error correction codes
and, as outlined in the present paper, it extends to the causal case
and has close connections to quantum preparations, measurements and
dynamics that are lacking for other values of $n$.  Given that
different applications work better with different definitions of the
conditional state, it is probably fair to say that there is no
uniquely compelling quantum generalization of conditional probability.

With this in mind, note that Coecke and Spekkens have outlined a broad
framework for generalizations of conditional probability within the
category theoretic approach to quantum theory \cite{Coecke2010}.
Encouragingly, it is possible to derive generalizations of Bayes'
theorem and conditioning abstractly within this framework, but it is
not yet clear what axioms within this framework are sufficient to
capture all the important aspects of conditioning.

The final quantum generalization of conditional probability to be
considered here is the \emph{quantum conditional expectation} (see
\cite{Umegaki1962} for the original paper and \cite{Petz1988,
  R'edei2007, Petz2008} for reviews).  This was proposed in the
context of quantum probability theory, which is a non-commutative
generalization of classical measure-theoretic probability within the
framework of operator algebras.  As such, it is well-defined for
systems with infinite dimensional Hilbert spaces as well as for
systems with an infinite number of degrees of freedom, for which there
is more than one unitarily inequivalent Hilbert space representation.
Also, it describes conditioning on an arbitrary algebra of
observables, rather than just on a tensor factor as has been
considered here.  Importantly for the present work, R{\'e}dei has
proposed an argument based on quantum conditional expectations
purporting to show that quantum theory cannot be understood as a
theory of Bayesian inference \cite{Redei1992} (see also
\cite{Valente2007}).

However, quantum conditional expectations have a major flaw that is
not shared by the conditional states formalism presented here.
Fortunately, the full operator-algebraic machinery is not necessary to
make the point; the case of finite dimensional Hilbert spaces and
conditioning on a tensor factor suffices.  Further details of the
general case and how the formalism used below follows from it can be
found in \cite{Petz2008}.

Classically, for a pair of random variables, $R$ and $S$, a
conditional expectation of $S$ given $R$ is a positive map
$\Phi_{R|S,R}$ from functions of $R$ and $S$ to functions of $R$ that
satisfies
\begin{equation}
  \label{eq:Disc:ClassCondExp}
  \Phi_{R|R,S}(f(R)) = f(R)
\end{equation}
for all functions $f(R)$ that are independent of $S$.  Any such map is
explicitly given by
\begin{equation}
  \label{eq:Disc:ClassCondExpExplicit}
  \Phi_{R|R,S}(f(R,S)) = \sum_S P(S|R) f(R,S),
\end{equation}
where $P(S|R)$ is a conditional probability distribution.

Starting from a joint state $P(R,S)$, one can obtain a conditional
expectation by plugging the associated conditional probability
$P(S|R)$ into eq. \eqref{eq:Disc:ClassCondExpExplicit}.  The main
point of this is that it allows the expectation value of any function
$f(R,S)$ to be computed from the marginal probability distribution
$P(R)$ via
\begin{equation}
\sum_R \Phi_{R|R,S}(f(R,S))P(R) = \sum_{R,S} f(R,S) P(R,S).
\end{equation}

The set of functions on $R$ and $S$ can be thought of as the dual
space to the set of probability distributions on $R$ and $S$, where
the linear functional $\hat{f}$ associated with $f(R,S)$ is given by
\begin{equation}
  \hat{f}(P(R,S)) = \sum_{R,S} P(R,S) f(R,S),
\end{equation}
i.e.\ it is the functional that maps the probability distribution
$P(R,S)$ to the expectation value of $f(R,S)$ with respect to
$P(R,S)$.  With respect to this identification, a conditional
expectation $\Phi_{R|R,S}$ has a dual map $\mathcal{E}_{R,S|R}$ that
maps the space of probability distributions over $R$ to the space of
probability distributions over $R$ and $S$.  This is given by
\begin{equation}
  \label{eq:Disc:Extend}
  \mathcal{E}_{R,S|R}(P(R)) = P(S|R)P(R)
\end{equation}
and is called a \emph{state extension} because every probability
distribution $P(R)$ gets mapped to a valid probability distribution
$P(R,S)= P(S|R)P(R)$ on a larger space.  In addition, state extensions
that are dual to conditional expectations satisfy
\begin{equation}
  \label{eq:Disc:Invariant}
  \sum_S \mathcal{E}_{R,S|R}(Q(R)) = Q(R)
\end{equation}
for every input distribution $Q(R)$.

In the finite dimensional, tensor factor case, the quantum conditional
expectation of $B$ given $A$ is a completely positive\footnote{The
  definition only calls for positivity, but it is a theorem that all
  quantum conditional expectations are completely positive (see
  \cite{Petz2008}).} linear map $\Phi_{A|AB}:\Lin[AB] \rightarrow
\Lin[A]$ that acts on the set of observables on $\Hilb[AB]$ and
satisfies
\begin{equation}
  \Phi_{A|AB} \left ( M_A \otimes I_B \right ) = M_A
\end{equation}
for all operators $M_A \in \Lin[A]$.  This is analogous to the
condition given in eq. \eqref{eq:Disc:ClassCondExp}.  The dual map
$\mathcal{E}_{AB|A}:\Lin[A] \rightarrow \Lin[AB]$ acts on states and
is a state extension, which means that
\begin{equation}
  \label{eq:Disc:QExtend}
  \tau_{AB} = \mathcal{E}_{AB|A} \left ( \tau_A \right )
\end{equation}
is a valid state for any state $\tau_A$ on $\Hilb[A]$.  In addition,
the state extensions that are dual to conditional expectations satisfy
\begin{equation}
  \label{eq:Disc:QInvariant}
  \Tr[B]{\mathcal{E}_{AB|A} \left ( \tau_A \right
    )} = \tau_A
\end{equation}
for every state $\tau_A$ on $\Hilb[A]$, which is analogous to
eq. \eqref{eq:Disc:Invariant}.

As in the classical case, one would like to associate every joint
state $\rho_{AB}$ with a conditional expectation, such that the dual
state extension $\mathcal{E}_{AB|A}$ satisfies
\begin{equation}
  \label{eq:Disc:Requirement}
  \mathcal{E}_{AB|A}(\rho_A) = \rho_{AB},
\end{equation}
i.e.\ it should give back the state that you started with when you
input the reduced state.  The analogous requirement is a key property
of classical conditional probability as it is what allows an arbitrary
joint state to be broken up into a marginal and a conditional that are
independent of one another.  Unfortunately, the fact that
eq. \eqref{eq:Disc:QInvariant} holds for every input state means that
this requirement can only be met for product states, i.e.\ states of
the form $\rho_{AB} = \rho_A \otimes \rho_B$.  This severely restricts
the applicability of quantum conditional expectations for describing
the correlations present in quantum states.  Indeed, they are
incapable of representing any correlations at all.  This fact is known
in the quantum probability literature (see \cite{Petz2008} example
9.6), but here is an elementary proof.

Ironically, the easiest way to show that $\rho_{AB}$ has to be a
product state is to use the conditional states formalism as outlined
in this paper.  The state extension $\mathcal{E}_{AB|A}$ is
Jamio{\l}kowski isomorphic to a causal conditional state
$\varrho_{AB|A'}$, where $A'$ has the same Hilbert space as $A$ and
the $'$ is just used to distinguish the input and output spaces.
Then, eq. \eqref{eq:Disc:Requirement} can be rewritten as
\begin{equation}
  \label{eq:Disc:CondRequirement}
  \rho_{AB} = \Tr[A']{\varrho_{AB|A'}\rho_{A'}},
\end{equation}
where $\rho_{A'}$ is the same state as $\rho_A$.  Similarly,
eq. \eqref{eq:Disc:QInvariant} can be rewritten as
\begin{align}
  \tau_A & = \Tr[A'B]{\varrho_{AB|A'}\tau_{A'}} \\
  & = \Tr[A']{\varrho_{A|A'}\tau_{A'}},\label{eq:Disc:CondInvariant}
\end{align}
for all $\tau_A$, where $\tau_{A'}$ is the same state as $\tau_A$ and
$\varrho_{A|A'} = \Tr[B]{\varrho_{AB|A'}}$. Since,
eq. \eqref{eq:Disc:CondInvariant} has to hold for every input state,
$\varrho_{A|A'}$ has to be Jamio{\l}kowski isomorphic to the identity
superoperator, so $\varrho_{A|A'} =
\Ket{\Phi^+}\Bra{\Phi^+}_{A|A'}^{T_{A'}}$.  Since this is pure,
monogamy of entanglement\footnote{The monogamy of entanglement is well
  known for positive operators.  The fact that it also applies to
  locally positive operators follows one of the results of
  \cite{Barnum2006}, which shows that monogamy applies to more general
  probabilistic theories, including one in which states are locally
  positive operators.\label{fn:Monogamy}} entails that
$\varrho_{AB|A'}$ must be of the form $\varrho_{AB|A'} =
\Ket{\Phi^+}\Bra{\Phi^+}_{A|A'}^{T_{A'}} \otimes M_B$ for some
operator $M_B$ on $\Hilb[B]$.  Substituting this into
eq. \eqref{eq:Disc:CondRequirement} gives
\begin{equation}
  \rho_{AB} = \rho_A \otimes M_B,
\end{equation}
which shows that $\rho_{AB}$ must be a product state and $M_B = \rho_B
= \Tr[A]{\rho_{AB}}$.  Therefore, conditional expectations associated
with joint states only exist for product states $\rho_{AB} = \rho_A
\otimes \rho_B$.

It is unclear why quantum probabilists have not regarded this as a
fatal flaw in their definition of quantum conditional expectation.
However, despite this problem, quantum conditional expectations are
still worthy objects of study as they come up in a variety of
contexts.  For example, the projection onto the fixed point set of a
completely positive map is a quantum conditional expectation.  The
point is just that the terminology ``quantum conditional expectation''
is an inaccurate way of describing the way that these maps relate to
quantum states.

For comparison, the acausal conditional state defined in the present
work can also be described as a map $\mathcal{F}_{AB|A}: \Lin[A]
\rightarrow \Lin[AB]$, similar to a state extension, with the crucial
difference that $\mathcal{F}_{AB|A}$ is not linear.  For a conditional
state $\rho_{B|A}$, the map $\mathcal{F}_{AB|A}$ is defined as
\begin{equation}
  \mathcal{F}_{AB|A} \left ( \rho_A \right ) = \rho_A^{\frac{1}{2}} \rho_{B|A}
  \rho_A^{\frac{1}{2}}.
\end{equation}
The nonlinearity allows this map to satisfy
eqs.~\eqref{eq:Disc:QInvariant} and \eqref{eq:Disc:Requirement}
without running into trouble with the monogamy of entanglement.
Crucially though, because this map is nonlinear, it does not have a
dual map that could be regarded as a conditional expectation.
Therefore, the close connection between conditional probabilities and
conditional expectations breaks down within this formalism.

R{\'e}dei's argument against a Bayesian interpretation of quantum
probabilities is based on a variant of Jeffrey conditioning with
respect to a quantum conditional expectation.  In the present context,
if the state of a region $A$ is updated from $\rho_A$ to
$\rho^{\text{post}}_A$, then R{\'e}dei asserts that the state for $AB$
should be updated to $\rho^{\text{post}}_{AB} = \mathcal{E}_{AB|A}
\left ( \rho^{\text{post}}_A \right )$, where $\mathcal{E}_{AB|A}$ is
a state extension derived from the prior state $\rho_{AB}$.  He then
shows that this update rule fails to satisfy an important stability
criterion in the infinite dimensional case.  Since his argument is
crucially based on the linearity of conditional expectations and their
duality with state extensions, it does not apply to the nonlinear maps
$\mathcal{F}_{AB|A}$ associated with quantum conditional states.
However, since the failure only occurs for infinite dimensional
systems, a full refutation will have to wait until the conditional
states formalism has been extended beyond the finite dimensional case
treated here.

To reiterate, it seems unlikely that there is a unique quantum
generalization of conditional probability that has properties
analogous to every single property of classical conditional
probability that is traditionally regarded as important.  For this
reason, it is important to keep applications in mind when defining
quantum conditionals, rather than working in a formal mathematical
vacuum.

\subsection{Operational Formalisms for Quantum Theory}

\label{Related:Op}

Recent efforts to replace the conventional formalism of quantum theory
with a new operational formulation --- typically in an effort to
provide an axiomatic derivation of quantum theory --- have much in
common with the work presented here.  Particularly cognate to our
approach is the work of Hardy \cite{Hardy2005, Hardy2007, Hardy2011},
the Pavia group \cite{Chiribella2009, Chiribella2011}, Oreshkov
\emph{et al.}  \cite{Oreshkov2011} and Coecke's group
\cite{Coecke2009, Coecke2009a}.

The reformulation of quantum theory presented by the Pavia group makes
heavy use of the Choi isomorphism between quantum operations and
bipartite states, and leverages this to represent quantum operations
by operators rather than maps.  Mathematically, this is also how we
achieve a unification of the treatment of acausally-related and
causally-related regions.  In particular, our
proposition~\ref{prop:CCS:Jamiol}, which specifies how to represent a
quantum channel in terms of conditional states, is the counterpart to
the expression of the action of a quantum channel in terms of the
\emph{link product}. Note that the Pavia group uses the Choi
isomorphism whereas we use the Jamio{\l}kowski-isomorphism.  The
latter has the advantage of being basis-independent, so that the
partial transposes that appear in the link product are absent in our
approach.

Hardy's latest work on reformulating quantum theory, using the
\emph{duotensor framework } \cite{Hardy2011}, also represents quantum
operations as operators in much the same way as is done in the quantum
combs framework and our own. Furthermore, Hardy's notion of
\emph{circuit trace} (an example of the causaloid product introduced
in \cite{Hardy2005, Hardy2011}) provides a unified way of representing
a composition of maps as well as tensor products of system states,
which is to say, a unified way of representing correlations between
acausally-related and causally-related regions.  The motivation for
Hardy's work on the causaloid product is very similar to the
motivation for our own, namely, to formulate quantum theory in a
manner that is even handed with regard to possible causal structures.

Recent work by Oreshkov, Costa and Brukner \cite{Oreshkov2011} also
represents correlations between acausally-related and causally-related
regions is a uniform manner by appealing to the Choi isomorphism.

One notable way in which we depart from all of these approaches is
that we treat classical systems internally to the formalism, on a par
with quantum systems, rather than as indices on operators representing
preparations and measurements.

Finally, we compare our approach to the categorical approach of
Coecke, where much of quantum theory (in particular the non-metrical
parts) is reformulated using the mathematical framework of symmetric
monoidal categories and the graphical calculus that can be defined for
these \cite{Coecke2009, Coecke2009a}.  Systems are the objects of the
category, while quantum states and quantum operations are the
morphisms.  The isomorphism between bipartite states and operations
also features prominently in this approach and arises from having a
\emph{compact structure} in the category.  Furthermore, classical
systems can be treated internally within this framework.  It should
also be noted that although Coecke's categorical framework is
typically used to formulate quantum theory as a theory of physical
processes, it can also be used to express our formulation of quantum
theory as a theory of inference.  For inference among acausally
related regions, this was done in \cite{Coecke2010}.  An extension of
this work to the case of inference among causally-related regions
should be instructive.

The mathematics of all of these approaches and our own are quite
analogous.  It is in the interpretational aspect that the greatest
differences are to be found. In the reformulations considered above,
quantum theory is given a rather minimalist interpretation -- it is
viewed as a framework for making predictions about the outcomes of
certain measurements given certain preparations.  Quantum states ---
however they are reformulated mathematically --- are taken to be
representations of preparation procedures, while quantum operations
are taken to be representations of transformation procedures.  These
approaches follow the interpretational tradition of operationalism.
By contrast, our work takes quantum states to represent the beliefs of
an agent about a spatio-temporal region and takes quantum operations
to represent belief propagation; it has an \emph{epistemological
}flavor rather than an \emph{operational }one.  For instance, the
notions that we deem to be most promising for making sense of the
quantum formalism are those one finds in textbooks on statistics and
inductive inference, such as Bayes' theorem, conditional
probabilities, statistical independence, conditional independence, and
sufficient statistics and not the notions that are common to the
operational approaches, such as measurements, transformations and
preparations.  In this sense, our approach is more closely aligned in
its philosophical starting point with \emph{quantum Bayesianism}, the
view developed by Caves, Fuchs and Schack \cite{Caves2002a, Caves2006,
  Fuchs2003a, Fuchs2010, Fuchs2010a}\footnote{Unlike the quantum
  Bayesians, however, we are not committed to the notion that the
  beliefs represented by quantum states concern the outcomes of future
  experiments. Rather, the picture we have in mind is of the quantum
  state for a region representing beliefs about the physical state of
  the region, even though we do not yet have a model to propose for
  the underlying physical states.}.

The particular merit of our epistemological approach is the strong
analogy that it affords between quantum inference and classical
probabilistic inference. It makes the conceptual content of various
quantum expressions more transparent than they would otherwise be.  It
also bolsters the view that quantum states ought to be interpreted as
states of knowledge.

Our take on how to incorporate causal assumptions into quantum theory
also has a rather different starting point than the works described
above.  The relation we posit between the notions of causation and
correlation is most informed by the work on causal networks (also
known as Bayesian networks), summarized in the textbooks of Pearl
\cite{Pearl2009} and of Glymour, Scheines and Spirtes
\cite{Spirtes2001}.  We ultimately hope to generalize this analysis of
causality by replacing classical probability theory with quantum
theory, understood as a theory of Bayesian inference
\cite{Leifer2008}.  If the goal is to develop a formalism for quantum
theory that is causally neutral then we argue that the causal network
approach holds an advantage over the operational approach.
Specifically, the problem with taking experimental operations as a
primitive notion is that they already have a notion of causal
structure built into them insofar as the output of an operation is
causally influenced by the input, but not vice-versa (similarly, if
the operation is a measurement, then the outcome is causally
influenced by the input but not vice-versa).  On the other hand,
elementary regions --- the primitive notion of the causal network
approach --- are precisely the sorts of thing that can enter into
arbitrary causal relations with one another, whilst not having any
intrinsic causal structure themselves.

\section{Limitations of the $\Sprod$-Product}

\label{Lim}

Using the $\Sprod$-product makes the conditional states formalism look
very similar to classical probability theory.  Often, by replacing
probabilities with operators and ordinary products with
$\Sprod$-products, one can obtain an equation that is valid in the
conditional states formalism from one that is valid for classical
conditional probabilities.  However, because the $\Sprod$-product is
non-associative and non-commutative, this does not always happen.  In
this section, the limitations of the $\Sprod$-product and the
disanalogies between classical conditional probabilities and quantum
conditional states are discussed.

\subsection{Conditionalized Equations}

\label{Lim:Cond}

In classical probability theory, if one takes a universally
valid\footnote{Universal validity means that the equation holds for
  the joint, marginal, and conditional probability distributions
  derived from \emph{any} joint probability distribution
  $P(R,S,\ldots)$, rather than holding only in special cases, e.g.\
  due to symmetries or degeneracies of a particular distribution.}
equation relating conditional, joint and marginal probability
distributions over a set of variables $R,S,\ldots$ and conditionalizes
every term on a disjoint variable $T$, then the equation that results
is still universally valid.  For example, the equation
\begin{equation}
  P(R,S) = P(S|R)P(R)
\end{equation}
generalizes to
\begin{equation}
  \label{eq:Lim:CondChain}
  P(R,S|T) = P(S|R,T)P(R|T),
\end{equation}
where a $T$ has been placed on the right of the $|$ in each term.

Unfortunately, the analogous property does not hold in the conditional
states formalism.  For example, for acausal states the analog of
eq. \eqref{eq:Lim:CondChain} would be
\begin{equation}
  \label{eq:Lim:QCondChain}
  \rho_{AB|C} = \rho_{B|AC} \Sprod \rho_{A|C}.
\end{equation}
Writing this out explicitly, the left hand side is
\begin{equation}
  \rho_C^{-\frac{1}{2}} \rho_{ABC} \rho_C^{-\frac{1}{2}},
\end{equation}
whereas the right hand side is
\begin{equation}
  \left ( \rho_C^{-\frac{1}{2}} \rho_{AC} \rho_C^{-\frac{1}{2}}\right
  )^{\frac{1}{2}} \rho_{AC}^{-\frac{1}{2}} \rho_{ABC}
  \rho_{AC}^{-\frac{1}{2}} \left ( \rho_C^{-\frac{1}{2}} \rho_{AC}
    \rho_C^{-\frac{1}{2}} \right )^{\frac{1}{2}}.
\end{equation}
Because the $\Sprod$-product is non-commutative, the terms involving
$\rho_{AC}$ and $\rho_{AC}^{-1}$ cannot be brought together and made
to cancel as they would in the classical case.  Therefore,
counterexamples to this rule can occur when $\rho_{AC}$ and $\rho_C$
do not commute.  For example, it is straightforward to verify that a
generalized W-state of the form
\begin{multline}
  \Ket{\psi}_{ABC} = \frac{1}{2} \left ( \Ket{001}_{ABC} +
    \Ket{010}_{ABC} \right ) \\ + \frac{1}{\sqrt{2}} \Ket{100}_{ABC}
\end{multline}
does not satisfy this rule.  The calculation is not especially
instructive, so it is omitted.

Note that a universally valid equation relating quantum conditional,
joint and marginal states is still universally valid if one
conditionalizes every term on a classical variable.  This follows from
the equality of the expressions for the left hand and right hand sides
of eq.~\eqref{eq:Lim:QCondChain} when $C$ is replaced by a classical
variable $T$.

\subsection{Limitations of Causal Joint States}

\label{Lim:CJS}

In \S\ref{CCS:TJS}, a causal joint state was defined as an operator of
the form $\varrho_{B|A} \Sprod \rho_A$.  This representation of two
causally-related regions highlights the symmetry with the acausal
case, since, up to a partial transpose, it is the same sort of
operator that would be used to represent two acausally-related
regions.  Based on this, the quantum Bayes' theorem was developed in a
way that is formally equivalent for acausally and causally-related
regions.  For this to work, we only needed causal joint states for two
causally-related regions.  However, since acausal states are not
limited to just two regions, it is natural to ask whether causal joint
states can be defined for more than two regions.  Unfortunately, the
naive generalization does not work for scenarios with mixed causal
structure, e.g.\ two causally-related regions that are both acausally
related to a third regions, and it also does not work for multiple
time-steps.

In the remainder of this section, these limitations are discussed and
a different definition of a causal joint state is suggested, which
works more generally, but does not exhibit the symmetry between the
acausal and causal scenarios for two regions.

\subsubsection{Mixed Causal Scenarios}

\label{Lim:CJS:MCS}

If region $B$ is in the causal future of region $A$, then they can be
assigned a causal joint state $\varrho_{AB} = \varrho_{B|A} \Sprod
\rho_A$.  The causal conditional state is Jamio{\l}kowski isomorphic
to the dynamical CPT map $\mathcal{E}_{B|A}$ and $\rho_A$ is the input
state.  Now suppose that there is a third region, $C$, that is
acausally-related to both $A$ and $B$, i.e.\ we start with a joint
state $\rho_{AC}$ of $A$ and $C$ and apply $\mathcal{E}_{B|A}$ to
region $A$, whilst doing nothing to $C$, to obtain $B$ and $C$ in a
joint state $\rho_{BC} = \left ( \mathcal{E}_{B|A} \otimes
  \mathcal{I}_C \right ) \left (\rho_{AC} \right )$.

One might think that a joint state of $ABC$ could be defined via
$\varrho_{ABC} = \rho_{B|A} \Sprod \rho_{AC}$.  We would like this
state to have the correct input and output states as marginals, so, in
particular, it should satisfy
\begin{equation}
  \label{eq:Lim:SProdWrong}
  \rho_{BC} = \mathcal{E}_{B|A} \left ( \rho_{AC} \right )=
  \Tr[A]{\varrho_{ABC}}.
\end{equation}
Unfortunately, this fails because theorem~\ref{thm:ACS:Jamiol} implies
that
\begin{equation}
  \label{eq:Lim:CJ}
  \rho_{BC} = \Tr[A]{\varrho_{B|A} \rho_{AC}},
\end{equation}
whereas
\begin{equation}
  \Tr[A]{\varrho_{ABC}} = \Tr[A]{\varrho_{B|A} \Sprod \rho_{AC}}.
\end{equation}
Because there is no trace over $C$, the cyclic property of the trace
cannot be used to equate these two expressions.

In fact, in addition to not having $\rho_{BC}$ as a reduced state,
$\varrho_{ABC}$ fails to correctly represent the correlations between
the causally-related regions $A$ and $B$ as well,
i.e. $\Tr[C]{\varrho_{ABC}} \neq \varrho_{AB}$.

To see the failure of both these conditions explicitly, consider an
example in which $\Hilb[A]$, $\Hilb[B]$ and $\Hilb[C]$ are all of
dimension $d$, the input state $\rho_{AC} = \frac{1}{d}
\Ket{\Phi^+}\Bra{\Phi^+}_{AC}$ is a maximally entangled state, and
$\mathcal{E}_{B|A}$ is the identity superoperator.  Then, the output
state is $\rho_{BC} = \frac{1}{d} \Ket{\Phi^+}\Bra{\Phi^+}_{BC}$, and
the causal joint state is $\varrho_{AB} = \frac{1}{d}
\Ket{\Phi^+}\Bra{\Phi^+}_{AB}^{T_A}$.

In this case, explicitly calculating the operator $\varrho_{ABC} =
\varrho_{B|A} \Sprod \rho_{AC}$ gives
\begin{equation}
  \varrho_{ABC} = \frac{1}{d} \Ket{\Phi^+}\Bra{\Phi^+}_{AC} \otimes
  \frac{I_B}{d}.
\end{equation}
Whilst the reduced state $\rho_{AC}$ gives the correct input state,
$\rho_{BC}$ is not the output state and $\varrho_{AB}$ is not the
causal joint state.  In fact, no operator of the form $\varrho_{ABC} =
\varrho_{B|A} \Sprod \rho_{AC}$ can ever satisfy all three conditions
$\rho_{AC} = \frac{1}{d}\Ket{\Phi^+}\Bra{\Phi^+}_{AC}$, $\varrho_{AB}
= \frac{1}{d} \Ket{\Phi^+}\Bra{\Phi^+}_{AB}^{T_A}$ and $\rho_{BC} =
\frac{1}{d} \Ket{\Phi^+}\Bra{\Phi^+}_{BC}$ simultaneously.  This is
because $\varrho_{ABC}$ is a locally positive operator, and hence it
must satisfy the monogamy of entanglement, but requiring all three
bipartite reduced states to be maximally entangled would violate
monogamy and, in fact, only one of the three conditions can be
satisfied.  Because of this, we have to move beyond locally positive
operators in order to faithfully represent all the correlations.

The validity of eq. \eqref{eq:Lim:CJ} suggests that using an ordinary
product instead of a $\Sprod$-product could be a better way of
defining a joint state in this scenario.  Indeed, the operator
$\varrho_{B|A} \rho_{AC}$ does have all the correct bipartite
marginals.  For example, in our example of maximally entangled input
and identity map dynamics, explicit calculation gives
\begin{equation}
  \varrho_{B|A} \rho_{AC} = \frac{1}{d} \sum_{j,k,m=1}^d
  \Ket{j}\Bra{k}_A \otimes \Ket{m}\Bra{j}_B \otimes \Ket{m}\Bra{k}_C.
\end{equation}
This violates the monogamy of entanglement, which is allowed because
it is not a locally positive operator.  In fact, it is not even
Hermitian.

Whilst this may turn out to be a useful representation, it has a
number of disadvantages compared to the $\Sprod$-product.  First of
all, it is non-unique because the operator $\rho_{AC} \varrho_{B|A}$
gives an equally good account of all the correlations.  One could even
take combinations of the two operators, such as
\begin{equation}
  \frac{1}{2} \left ( \varrho_{B|A} \rho_{AC} + \rho_{AC}
    \varrho_{B|A} \right ),
\end{equation}
which might be useful because it is Hermitian.

Secondly, given an arbitrary operator $M_{AB}$, it is not clear how to
check whether it is of the form $\varrho_{B|A} \rho_{A}$ without
running over all possible input states $\rho_{A}$ and checking whether
$M_{AB} \rho_{A}^{-1}$ is a valid causal conditional state.  In
contrast, when using the $\Sprod$-product, we know that an operator is
of the form $\varrho_{B|A} \Sprod \rho_A$ iff it is the partial
transpose of a valid acausal state on $AB$, which is a straightforward
condition to check.

Finally, a related issue is that, when using the ordinary product
instead of the $\Sprod$-product, the set of possible causal joint
states depends on the causal direction.  For an evolution from $A$ to
$B$, a causal joint state would be of the form $\varrho_{B|A}\rho_A$,
but for an evolution from $B$ to $A$ it would be of the form
$\varrho_{A|B} \rho_B$.  These define two different sets of operators,
so we lose the symmetry that was used to obtain Bayes' theorem in the
causal case.

\subsubsection{Multiple Time-Steps}

\label{Lim:CJS:MTS}

For three space-like separated regions, a Markovian joint state
$\rho_{ABC}$, where $A$ and $C$ are conditionally independent given
$B$, can always be decomposed via the chain rule into
\begin{equation}
  \label{eq:Lim:Markov}
  \rho_{ABC} = \rho_{C|B} \Sprod \left (\rho_{B|A} \Sprod
    \rho_A\right )
\end{equation}

For three time-like separated regions, the analogue of Markovianity is a
two time-step dynamics where the first CPT map is
$\mathcal{E}_{B|A}:\Lin[A]\rightarrow \Lin[B]$ and the second is
$\mathcal{E}_{C|B}:\Lin[B]\rightarrow \Lin[C]$, i.e.\ it has no direct
dependence on $A$.  It is natural to ask whether this situation can be
represented by a tripartite causal joint state that can be decomposed
in a manner similar to eq. \eqref{eq:Lim:Markov}.

Suppose that the causal conditional states isomorphic to
$\mathcal{E}_{B|A}$ and $\mathcal{E}_{C|B}$ are $\varrho_{B|A}$ and
$\varrho_{C|B}$.  If the input state is $\rho_A$ then the first
time-step is represented as
\begin{align}
  \rho_B & = \mathcal{E}_{B|A} \left ( \rho_A \right ) \\
  & = \Tr[A]{\varrho_{B|A} \rho_A}
\end{align}
and the second time-step is represented as
\begin{align}
  \rho_C & = \mathcal{E}_{C|B} \left ( \rho_B \right ) \\
  & = \Tr[B]{\varrho_{C|B} \rho_B}.
\end{align}
It follows that
\begin{align}
  \rho_C & = \mathcal{E}_{C|B} \circ \mathcal{E}_{B|A} \left ( \rho_A \right ) \\
  & = \Tr[B]{\varrho_{C|B} \Tr[A]{\varrho_{B|A} \rho_A}},
\end{align}
The question is whether the complete dynamics might also be
representable as
\begin{align}
  \rho_C   & = \Tr[AB]{\varrho_{C|B} \Sprod \left ( \varrho_{B|A} \Sprod \rho_A
    \right )},
\end{align}
which in turn would suggest that $\varrho_{ABC} = \varrho_{C|B} \Sprod
\left ( \varrho_{B|A} \Sprod \rho_A \right )$ might be a good
candidate for a tripartite causal joint state.

This fails because the expression $\varrho_{ABC} = \varrho_{C|B}
\Sprod \left ( \varrho_{B|A} \Sprod \rho_A \right )$ is not well
defined.  To see this, expand the first $\Sprod$-product to obtain
\begin{equation}
  \varrho_{ABC} = \sqrt{\varrho_{B|A} \Sprod \rho_A} \varrho_{C|B}
  \sqrt{\varrho_{B|A} \Sprod \rho_A}.
\end{equation}
The term $\sqrt{\varrho_{B|A} \Sprod \rho_A}$ is not well defined
because $\varrho_{B|A} \Sprod \rho_A$ is not a positive operator, but
only locally positive, so it may have negative eigenvalues.  Hence, it
does not have a unique square root.  This could be remedied by
adopting a convention for square roots of Hermitian operators, such as
demanding that the square root of each negative eigenvalue has
positive imaginary part.  The resulting tripartite operator
$\varrho_{ABC}$ would then have the correct reduced states, $\rho_A$,
$\rho_B$ and $\rho_C$, representing the state of the system at each
time-step.  However, it is not related to a tripartite state of three
acausally-related regions via partial transposes, so the symmetry that
motivates the use of the $\Sprod$-product representation is lost.
This loss of symmetry resonates with previous work showing that
tripartite ``entanglement in time'' is not isomorphic to ordinary
tripartite entanglement \cite{Taylor2004}.

As with the evolution of a subsystem, these problems can be remedied
by using the ordinary product instead of the $\Sprod$-product to
represent time evolutions, but it is subject to the same disadvantages
that were discussed in that context.

\section{Open Questions}

\label{Open}

\subsection{The Quantum Conditionals Problem}

Monogamy of entanglement is a key feature that distinguishes classical
from quantum information.  There is a closely related difference
between classical conditional probability distributions and acausal
quantum conditional states that deserves further investigation.
Classically, if $P(R)$ is a probability distribution and $P(S|R)$,
$P(T|R,S)$ are conditional probability distributions then
\begin{equation}
  P(R,S,T) = P(T|R,S) P(S|R) P(R)
\end{equation}
is always a valid probability distribution.  Furthermore, the
distribution so defined has the correct marginal and conditional
states in the sense that
\begin{align}
  \sum_{S,T} P(R,S,T) & = P(R) \\
  \frac{\sum_T P(R,S,T)}{\sum_{S,T} P(R,S,T)} & = P(S|R) \\
  \frac{P(R,S,T)}{\sum_T P(R,S,T)} & = P(T|R,S),
\end{align}
whenever the left hand sides are well defined.

In the quantum case, the analogous properties do not hold.  Although a
tripartite state $\rho_{ABC}$ can always be decomposed as
\begin{equation}
  \label{eq:Disc:Chain}
  \rho_{ABC} = \rho_{C|AB} \Sprod \rho_{B|A} \Sprod \rho_A
\end{equation}
via the chain rule, one cannot start with an arbitrary reduced state
$\rho_A$ and two arbitrary conditional states, $\rho_{B|A}$ and
$\rho_{C|AB}$, and expect there to be a joint state $\rho_{ABC}$ that
has these conditional and reduced states.

For example, suppose that $B$ and $C$ are conditionally independent of
$A$, i.e. $\rho_{C|AB} = \rho_{C|A}$.  Now, suppose that $\rho_A$ is
chosen to have more than one nonzero eigenvalue and $\rho_{B|A}$ is
chosen to be maximally entangled, e.g. $\rho_{B|A} =
\Ket{\Phi^+}\Bra{\Phi^+}_{B|A}$.  This implies that the reduced state
$\rho_{AB} = \rho_{B|A} \Sprod \rho_A$ is pure and entangled.  If, in
addition, $\rho_{C|A}$ is chosen to be $\rho_{C|A}=
\Ket{\Phi^+}\Bra{\Phi^+}_{C|A}$, then the reduced state $\rho_{AC} =
\rho_{C|A} \Sprod \rho_A$ should also be pure and entangled.  However,
monogamy of entanglement says that this is impossible, so these
choices of conditional state are not compatible.  Determining the full
set of constraints on coexistent conditional states for three
acausally-related regions would be an interesting problem, as would
determining the computational complexity of the $n$-party
generalization.

\section{Conclusions}

\label{Conc}

The formalism of quantum conditional states presented in this paper
provides a step towards a formalism for quantum theory that is
independent of causal structure, as a theory of probabilistic
inference ought to be, and provides a closer analogy between quantum
theory and classical probability theory.  There is significant
potential to use these results to simplify and generalize existing
approaches to problems in quantum information theory.  As an example
of this, in a companion paper \cite{Leifer2011} we provide an approach
to the problems of compatibility and pooling of quantum states that is
based on a principled application of Bayesian conditioning and is a
direct generalization of existing approaches to the classical versions
of these problems.  It seems unlikely that the possibility of this
approach would have been noticed within the conventional quantum
formalism.  However, this is only the beginning and we anticipate
applications to quantum estimation theory and to quantum cryptography,
such as studying the relationship between cryptography protocols that
employ different causal arrangements to achieve the same task.  As it
stands, the formalism is limited to two disjoint elementary quantum
regions and the most pressing problem is to generalize it to arbitrary
causal scenarios.  This is a topic of ongoing work.

\begin{acknowledgments}
  We thank David Poulin for useful discussions about conditional
  states.  M.L. would like to thank UCL Department of Physics and
  Astronomy for their hospitality.  Part of this work was completed
  whilst M.L. was a postdoctoral fellow at the Institute for Quantum
  Computing, University of Waterloo, Canada, where he was supported in
  part by MITACS and ORDCF.  M.L. also acknowledges support from The
  Foundational Questions Institute through the Grant RFP1-06-006.
  Research at Perimeter Institute is supported by the Government of
  Canada through Industry Canada and by the Province of Ontario
  through the Ministry of Economic Development and Innovation.
\end{acknowledgments}

\appendix

\section{Proofs of Theorems}
\label{Appendix:Proofs}

{ 
\renewcommand{\thetheorem}{\ref{thm:ACS:Jamiol}}
\begin{theorem}[Jamio{\l}kowski Isomorphism]
  Let $\mathfrak{E}_{B|A}:\Lin[A] \to \Lin[B]$ be a linear map and let
  $M_{AC} \in \Lin[AC]$ be a linear operator, where $\Hilb[C]$ is a
  Hilbert space of arbitrary dimension.  Then, the action of
  $\mathfrak{E}_{B|A}$ on $\Lin[A]$ (tensored with the identity on
  $\Lin[C]$) is given by
  \begin{equation}
    \label{eq:App:reverseJamiol}
    (\mathfrak{E}_{B|A} \otimes \mathcal{I}_C) \left ( M_{AC} \right )
    = \Tr[A]{\rho_{B|A} M_{AC}},
  \end{equation}
  where $\rho_{B|A}\in \Lin[AB]$ is given by
  \begin{equation}
    \label{eq:App:Jamiol}
    \rho_{B|A} \equiv (\mathfrak{E}_{B|A'} \otimes \mathcal{I}_{A}) \left (
      \sum_{j,k} \Ket{j}\Bra{k}_A \otimes \Ket{k}\Bra{j}_{A'} \right ).
  \end{equation}
  Here, $A'$ labels a second copy of $A$, $\mathcal{I}_A$ is the
  identity superoperator on $\Lin[A]$, and $\{\Ket{j} \}$ is an
  orthonormal basis for $\Hilb[A]$.

  Furthermore, the operator $\rho_{B|A}$ is an acausal conditional
  state, i.e.\ it satisfies definition~\ref{def:ACS:ACS}, if and
  only if $\mathfrak{E}_{B|A}\circ T_A$ is CPT, where $T_A: \Lin[A]
  \to \Lin[A]$ denotes the linear map implementing the partial
  transpose relative to some basis.
\end{theorem}
\addtocounter{theorem}{-1}
}

To prove this theorem, it is useful to make use of the connection
between the Choi and the Jamio{\l}kowski isomorphisms.  The map that
is Choi-isomorphic to an operator $\rho_{B|A}$ is given by
\begin{equation}
  \label{eq:App:reverseChoi}
  (\mathcal{E}_{B|A} \otimes \mathcal{I}_{C}) \left ( M_{AC} \right ) = \Bra{\Phi^+}_{AA'}
  \rho_{B|A'} M_{AC} \Ket{\Phi^+}_{AA'},
\end{equation}
where $\Ket{\Phi^+}_{AA'} = \sum_j \Ket{jj}_{AA'}$ is a canonical
maximally entangled state defined with respect to a preferred basis
$\{\Ket{j} \}$ for $\Hilb[A]$.

The operator is recovered from the map via
\begin{equation}
  \label{eq:App:Choi}
  \rho_{B|A} \equiv \left ( \mathcal{E}_{B|A'} \otimes \mathcal{I}_A
  \right )\left ( \Ket{\Phi^+}\Bra{\Phi^+}_{AA'} \right ).
\end{equation}

Because
\begin{equation}
  \sum_{j,k} \Ket{j}\Bra{k}_A \otimes \Ket{k}\Bra{j}_{A'}
  = \left ( \Ket{\Phi^+}\Bra{\Phi^+}_{AA'} \right )^{T_A},
\end{equation}
eqs.~\eqref{eq:App:Choi} and \eqref{eq:App:Jamiol} differ only by
whether one uses the projector onto the maximally entangled state
(Choi) or the partial transpose thereof (Jamio{\l}kowski) and the two
isomorphic maps to $\rho_{B|A}$ are related by $\mathfrak{E}_{B|A} =
\mathcal{E}_{B|A} \circ T_A$, where $T_A$ is the partial transpose
operation with respect to the basis used to define the Choi
isomorphism.

The equivalence of eq.~\eqref{eq:App:reverseChoi} and
eq.~\eqref{eq:App:reverseJamiol} is established as follows.
\begin{multline}
  \left ( \mathcal{E}_{B|A} \otimes
    \mathcal{I}_C \right ) \left ( M_{AC} \right
  ) = \\
  \Bra{\Phi^+}_{AA'}\rho_{B|A'} M_{AC} \Ket{\Phi^+}_{AA'} \\
  = \sum_{j,k} \Bra{jj}_{AA'} \rho_{B|A'} M_{AC}
  \Ket{kk}_{AA'} \\
  = \sum_{j,k} \Bra{j}_{A'} \rho_{B|A'} \Ket{k}_{A'} \Bra{j}_A
  M_{AC} \Ket{k}_A \\
  = \sum_{j,k} \Bra{j}_{A} \rho_{B|A} \Ket{k}_{A} \Bra{k}_A
  M_{AC}^{T_A} \Ket{j}_A \\
  = \Tr[A]{\rho_{B|A} M_{AC}^{T_A}} \\
  = \left ( \left [ \mathfrak{E}_{B|A} \circ T_A \right ]
    \otimes \mathcal{I}_{C} \right ) \left ( M_{AC} \right )
\end{multline}

\begin{proof}[Proof of theorem~\ref{thm:ACS:Jamiol}]
  Eq.~(\ref{eq:App:reverseJamiol}) is derived from
  eq.~(\ref{eq:App:Jamiol}) as follows:
   \begin{multline}
     \left ( \mathfrak{E}_{B|A} \otimes \mathcal{I}_C \right ) \left (
       M_{AC} \right ) \\
     = \left ( \mathfrak{E}_{B|A'} \otimes \mathcal{I}_C \right )
     \left ( \left [ \sum_k \Ket{k}\Bra{k}_{A'} \right ] M_{A'C}
       \left [ \sum_j \Ket{j} \Bra{j}_{A'}  \right ] \right ) \\
     = \left ( \mathfrak{E}_{B|A'} \otimes \mathcal{I}_C \right )
     \left ( \sum_{j,k} \Bra{k}_{A'} M_{A'C}
       \Ket{j}_{A'} \Ket{k}\Bra{j}_{A'} \right ) \\
     = \left ( \mathfrak{E}_{B|A'} \otimes \mathcal{I}_C \right )
     \left ( \sum_{j,k} \Tr[A]{\Ket{j}\Bra{k}_{A} M_{AC}}
       \Ket{k}\Bra{j}_{A'} \right ) \\
     = \Tr[A]{\left ( \mathfrak{E}_{B|A'} \otimes \mathcal{I}_C \right
       ) \left ( \sum_{j,k} \Ket{j}\Bra{k}_A \otimes \Ket{k}\Bra{j}_{A'} \right )
       M_{AC}} \\
     = \Tr[A]{\rho_{B|A} M_{AC}}
   \end{multline}

   Now suppose that $\rho_{B|A}$ is an acausal conditional state,
   i.e.\ it is positive and $\Tr[B]{\rho_{B|A}} = I_A$.  To show that
   the Jamio{\l}kowski-isomorphic map composed with a partial
   transpose, $\mathfrak{E}_{B|A} \circ T_A$, is trace-preserving, note
   that $T_A$ is trace-preserving, so it suffices to show that
   $\mathfrak{E}_{B|A}$ is trace-preserving.  This proceeds as follows:
   \begin{align}
     \Tr[B]{\mathfrak{E}_{B|A} \left ( M_A \right )} & =
     \Tr[B]{\Tr[A]{\rho_{B|A} M_A}} \\
     & = \Tr[A]{\Tr[B]{\rho_{B|A}} M_A} \\
     & = \Tr[A]{I_A M_A} \\
     & = \Tr[A]{M_A}.
   \end{align}
   To show that $\mathfrak{E}_{B|A} \circ T_A$ is completely positive, note that
   it is equal to the Choi-isomorphic map $\mathcal{E}_{B|A}$, so it
   suffices to show that the latter is completely positive.  By
   definition
   \begin{equation}
     \mathcal{E}_{B|A} \otimes \mathcal{I}_C \left ( \rho_{AC} \right
     ) = \Bra{\Phi^+}_{AA'} \rho_{B|A'} \otimes \rho_{AC}\Ket{\Phi^+}_{AA'}
   \end{equation}
   and this is a positive operator for arbitrary positive $\rho_{AC}$,
   where $\Hilb[C]$ can have any dimension.

   Conversely, suppose $\mathfrak{E}_{B|A} \circ T_A$ is CPT.  Then,
   $\mathfrak{E}_{B|A}$ is also trace preserving, so
   \begin{align}
     \Tr[B]{\rho_{B|A}} & = \Tr[B]{\sum_{j,k} \Ket{j}\Bra{k}_A \otimes
       \mathfrak{E}_{B|A} \left ( \Ket{k}\Bra{j}_{A'} \right )} \\
     & = \sum_{j,k} \Ket{j}\Bra{k} \otimes \Tr[B]{\mathfrak{E}_{B|A}
       \left ( \Ket{k}\Bra{j}_{A'} \right )}  \\
     & = \sum_{j,k} \Ket{j}\Bra{k}_A \otimes
     \Tr[A']{\Ket{k}\Bra{j}_{A'}} \\
     & = \sum_{j,k} \Ket{j}\Bra{k}_A \delta_{j,k} \\
     & = \sum_j \Ket{j}\Bra{j}_A \\
     & = I_A.
   \end{align}
   Also,
   \begin{equation}
     \rho_{B|A} = \left ( \mathcal{E}_{B|A'} \otimes \mathcal{I}_A \right ) \left (
       \Ket{\Phi^+}\Bra{\Phi^+}_{AA'} \right ),
   \end{equation}
   and this is a CPT map acting on a positive operator, so
   $\rho_{B|A}$ is positive.
\end{proof}

{
\renewcommand{\thetheorem}{\ref{thm:CCS:States}}
\begin{theorem}
  Let $\sigma_{A|X}$ be a hybrid operator, so that by
  eq.~(\ref{eq:CCS:HybridOp}) it can be written as
  \begin{equation}
    \label{eq:App:States}
    \sigma_{A|X} = \sum_x \rho^A_x \otimes \Ket{x}\Bra{x}_X,
  \end{equation}
  for some set of operators $\{ \rho^A_x\}$.  Then, $\sigma_{A|X}$
  satisfies the definition of both an acausal and a causal
  conditional state for $A$ given $X$, iff each of the components
  $\rho^A_x$ is a normalized state on $\Hilb[A]$.
\end{theorem}
\addtocounter{theorem}{-1} }
\begin{proof}
  Suppose $\sigma_{A|X}$ has the form of eq.~\eqref{eq:App:States} for
  a set of normalized states $\{\rho^A_x\}$.  Then it is clearly
  positive and satisfies $\Tr[A]{\sigma_{A|X}} = I_X$ because
  $\Tr[A]{\rho^A_x} = 1$ for every $x$.  Therefore, $\sigma_{A|X}$ is
  an acausal conditional state.  On the other hand,
  $\sigma_{A|X}$ is invariant under partial transpose on $X$, so it
  is also a causal conditional state.

  Conversely, suppose that $\sigma_{A|X}$ is an acausal conditional
  state.  This means that it is positive and satisfies
  $\Tr[A]{\sigma_{A|X}} = I_X$.  Positivity means that
  $\Bra{\psi}_{AX} \sigma_{A|X} \Ket{\psi}_{AX} \geq 0$ for all
  $\Ket{\psi}_{AX} \in \Hilb[AX]$.  If $\sigma_{A|X}$ has the form of
  eq.~\eqref{eq:App:States} then taking $\Ket{\psi}_{AX} =
  \Ket{\phi}_A \otimes \Ket{x}_A$ gives $\Bra{\phi}_A \rho^A_x
  \Ket{\phi}_A \geq 0$.  By varying over all $\Ket{\phi}_A \in
  \Hilb[A]$, this implies that each $\rho^A_x$ is a positive operator.
  To prove that these operators are normalized note that
  \begin{equation}
    \Tr[A]{\sigma_{A|X}} = \sum_x \Tr[A]{\rho^A_x} \Ket{x}\Bra{x}_X.
  \end{equation}
  This is an eigendecomposition of $\Tr[A]{\sigma_{A|X}}$ with
  eigenvalues $\Tr[A]{\rho^A_x}$ and if this is the identity operator
  then each of these eigenvalues must be $1$.

  On the other hand, if $\sigma_{A|X}$ is a causal conditional state,
  then its partial transpose over $X$ must be positive and satisfy
  $\Tr[A]{\varrho^{T_X}_{A|X}} = I_X$.  However, operators of the form
  of eq.~\eqref{eq:App:States} are invariant under partial
  transpose on $X$, so the same argument applies.
\end{proof}

{
\renewcommand{\thetheorem}{\ref{thm:CCS:POVM}}
\begin{theorem}
  Let $\sigma_{Y|A}$ be a hybrid operator so that it can be written in
  the form
  \begin{equation}
    \label{eq:App:POVM}
    \sigma_{Y|A} = \sum_y \Ket{y}\Bra{y}_Y \otimes  E^A_y,
  \end{equation}
  for some set of operators $\{ E^A_y \}$.  Then, $\sigma_{Y|A}$
  satisfies the definition of both an acausal and a causal
  conditional state for $Y$ given $A$ iff the components $E^A_y$ form
  a POVM on $\Hilb[A]$, i.e.\ each $E^A_y$ is positive and $\sum_y
  E^A_y = I_A$.
\end{theorem}
\addtocounter{theorem}{-1} }
\begin{proof}
  Suppose $\sigma_{Y|A}$ has the form of eq.~\eqref{eq:App:POVM}
  for a POVM $\{E^A_y\}$.  Then it is clearly positive and satisfies
  $\Tr[Y]{\sigma_{Y|A}} = \sum_y E^A_y = I_A$.  Therefore,
  $\sigma_{Y|A}$ is an acausal conditional state.  On the other
  hand
  \begin{equation}
    \sigma_{Y|A}^{T_A} = \sum_y \Ket{y}\Bra{y}_Y \otimes \left (E^A_y
    \right )^{T_A}.
  \end{equation}
  is also positive because the positive operators $E^A_y$ remain
  positive under the transpose.  Also, $\Tr[Y]{\sigma_{Y|A}^{T_A}} =
  \sum_y \left (E^A_y\right )^{T_A} = I_A^{T_A} = I_A$.  Therefore,
  $\sigma_{Y|A}$ is also a causal conditional state.

  Conversely, suppose that $\sigma_{Y|A}$ is an acausal conditional
  state.  This means that it is positive and satisfies
  $\Tr[Y]{\sigma_{Y|A}} = I_A$.  By the same argument used in the
  proof of theorem~\ref{thm:CCS:States}, positivity implies that, if
  $\sigma_{Y|A}$ is of the form of eq.~\eqref{eq:App:POVM}, then each
  of the components $E^A_y$ must be positive.  Since
  $\Tr[Y]{\sigma_{Y|A}} = \sum_y E^A_y$, the components must form a
  POVM.

  On the other hand, if $\sigma_{Y|A}$ is an acausal conditional state
  then, by the argument just given, its partial transpose over $A$
  must be of the form of eq.~\eqref{eq:App:POVM} for some POVM
  $\{E^A_y\}$.  This means that $\sigma_{Y|A}$ itself can be written
  as
  \begin{equation}
    \sigma_{Y|A} = \sum_y \Ket{y}\Bra{y}_Y \otimes \left (E^A_y
    \right )^{T_A},
  \end{equation}
  but since the operators $\left ( E^A_y \right )^{T_A}$ form a POVM
  whenever $\{E^A_y\}$ is a POVM, $\sigma_{Y|A}$ is of the required
  form.
\end{proof}

{
\renewcommand{\thetheorem}{\ref{thm:CCS:CondBP}}
\begin{theorem}
  Let $\mathcal{E}_{B|A}$, $\mathcal{E}_{C|B}$ and $\mathcal{E}_{C|A}$
  be linear maps such that $\mathcal{E}_{C|A} = \mathcal{E}_{C|B}
  \circ \mathcal{E}_{B|A}$.  Then, the Jamio{\l}kowski isomorphic
  operators, $\varrho_{B|A}$, $\varrho_{C|B}$ and $\varrho_{C|A}$
  satisfy
  \begin{equation}
    \label{eq:App:CondBP}
    \varrho_{C|A} = \Tr[B]{\varrho_{C|B}\varrho_{B|A}}.
  \end{equation}
  Conversely, if three operators satisfy eq.~\eqref{eq:App:CondBP},
  then the Jamio{\l}kowski isomorphic maps satisfy $\mathcal{E}_{C|A}
  = \mathcal{E}_{C|B} \circ \mathcal{E}_{B|A}$.
\end{theorem}
\addtocounter{theorem}{-1} }

\begin{proof}
  By definition, the Jamio{\l}kowski isomorphic operator to
  $\mathcal{E}_{C|A}$ is
  \begin{align}
    \varrho_{C|A} & = (\mathcal{E}_{C|A'} \otimes \mathcal{I}_A) \left
      ( \sum_{j,k}
      \Ket{j}\Bra{k}_A \otimes \Ket{k}\Bra{j}_{A'} \right ) \\
    & = \sum_{j,k} \Ket{j}\Bra{k} \otimes \left [ \mathcal{E}_{C|B} \circ
    \mathcal{E}_{B|A'} \left ( \Ket{k}\Bra{j}_{A'} \right ) \right ]
  \end{align}
  Applying theorem~\ref{thm:CCS:Jamiol} to $\mathcal{E}_{B|A'}$ gives
  \begin{multline}
    \varrho_{C|A}  = \\ \sum_{j,k} \Ket{j}\Bra{k}_A \otimes
    \mathcal{E}_{C|B} \left ( \Tr[A']{\varrho_{B|A'}
        \Ket{k}\Bra{j}_{A'}} \right ) \\
     = \sum_{j,k} \Ket{j}\Bra{k}_A \otimes \mathcal{E}_{C|B} \left (
      \Bra{j}_{A'}\varrho_{B|A'} \Ket{k}_{A'} \right ),
  \end{multline}
  and then applying the same theorem to $\mathcal{E}_{C|B}$ gives
  \begin{multline}
    \varrho_{C|A} = \\ \sum_{j,k} \Ket{j}\Bra{k}_A \otimes
      \Bra{j}_{A'}\Tr[B]{\varrho_{C|B}\varrho_{B|A'}} \Ket{k}_{A'}.
  \end{multline}
  Since $A'$ is a dummy label in this expression, it can be changed to
  $A$ and then
  \begin{align}
    \varrho_{C|A} & = \sum_{j,k} \Ket{j}\Bra{j}_A
    \Tr[B]{\varrho_{C|B}\varrho_{B|A}} \Ket{k}\Bra{k}_A \\
    & = \Tr[B]{\varrho_{C|B}\varrho_{B|A}}.
  \end{align}

  For the converse direction, we have
  \begin{align}
    \mathcal{E}_{C|A} \left ( M_A \right ) & = \Tr[A]{\varrho_{C|A}
      M_A} \\
    & = \Tr[A]{\Tr[B]{\varrho_{C|B}\varrho_{B|A}}M_A} \\
    & = \Tr[B]{\varrho_{C|B} \Tr[A]{\varrho_{B|A}M_A}} \\
    & = \Tr[B]{\varrho_{C|B} \mathcal{E}_{B|A} \left ( M_A \right )}
    \\
    & = \mathcal{E}_{C|B} \left ( \mathcal{E}_{B|A} \left ( M_A \right
      ) \right ) \\
    & = \mathcal{E}_{C|B} \circ \mathcal{E}_{B|A} \left ( M_A \right ).
  \end{align}
\end{proof}

\bibliography{quantumpooling}

\end{document}